\documentclass[10pt, letterpaper]{article}

\usepackage[utf8]{inputenc} 
\usepackage[T1]{fontenc}
\usepackage{microtype}
\usepackage[margin=1in, top=0.8in]{geometry}
\usepackage{mathpazo} 
\usepackage[round]{natbib}

\usepackage{titlesec}
\titleformat{\section}{\large\bfseries\raggedright}{\thesection}{1em}{}
\titlespacing*{\section}{0pt}{1.2ex plus 0.3ex minus 0.2ex}{1.0ex plus 0.2ex minus 0.2ex}
\titleformat{\subsection}{\normalsize\bfseries\raggedright}{\thesubsection}{1em}{}
\titlespacing*{\subsection}{0pt}{1.0ex plus 0.3ex minus 0.2ex}{0.5ex plus 0.2ex}
\titleformat{\subsubsection}{\normalsize\bfseries\raggedright}{\thesubsubsection}{1em}{}
\titlespacing*{\subsubsection}{0pt}{0.8ex plus 0.2ex minus 0.2ex}{0.3ex plus 0.1ex}
\titleformat{\paragraph}[runin]{\normalsize\bfseries}{\theparagraph}{1em}{}
\titlespacing*{\paragraph}{0pt}{0.8ex plus 0.2ex minus 0.2ex}{0.5em}

\renewenvironment{abstract}%
{%
  \vspace*{0.075in}%
  \begin{center}%
    {\large \bfseries Abstract}%
  \end{center}%
  \vspace{-2em}%
  \begin{quote}%
}
{
  \end{quote}%
  \vspace{1ex}%
}

\makeatletter
\def\blfootnote{\gdef\@thefnmark{}\@footnotetext}
\makeatother

\usepackage{url}           
\usepackage{booktabs}
\usepackage{nicefrac}  
\usepackage[table]{xcolor}
\usepackage{enumitem}   
\usepackage{amsmath, amsthm, amsfonts, amssymb, mathtools}
\usepackage{graphicx}
\usepackage{mdframed}
\usepackage{float}
\usepackage{caption}
\usepackage[toc,page,header]{appendix}
\usepackage{wrapfig}
\usepackage{algorithm, algpseudocode}
\usepackage{multicol}
\usepackage{etoolbox}
\usepackage{multirow}
\usepackage[scaled]{beramono} 
\usepackage[most]{tcolorbox}

\usepackage{notation}

\usepackage[
    linktocpage=true,
    plainpages=false,
    pdfpagelabels
]{hyperref}
\hypersetup{
  colorlinks=true,
  citecolor=blue,
  urlcolor=cyan,
  linkcolor=magenta,
}

\definecolor{light-gray}{gray}{0.93}

\setlist{itemsep=0.1em, topsep=0em, leftmargin=2.5em}

\usepackage{thmtools, thm-restate}

\declaretheoremstyle[
  spaceabove=1.5ex plus 0.2ex minus 0.2ex, 
  spacebelow=0.5ex,                        
  headfont=\bfseries,                     
  bodyfont=\normalfont                  
]{tighterrem}

\declaretheoremstyle[
  spaceabove=1.5ex plus 0.2ex minus 0.2ex, 
  spacebelow=0.5ex,                        
  headfont=\bfseries,                     
  bodyfont=\it                     
]{tighterthm}

\newtheorem*{theorem*}{Theorem}
\newtheorem*{informaltheorem*}{Informal Theorem}

\theoremstyle{definition}

\renewenvironment{proof}{\noindent\textbf{Proof.}\,}{\qed}

\title{  
    \bf
    Linear Last-Iterate Convergence Rates for \\
    Optimistic Learning in Harmonic Games
}

\date{}

\author{
  \begin{tabular}{c}
    John Lazarsfeld$^{1}$
  \end{tabular}
  \hfil
  \hspace*{1em}
  \begin{tabular}{c}
    Anas Barakat$^{2}$
  \end{tabular}
  \hfil
  \hspace*{1em}
  \begin{tabular}{c}
    Jacob Abernethy$^{1, 3}$
  \end{tabular}
  \hfil
  \hspace*{1em}
    \begin{tabular}{c}
    Vidya Muthukumar$^{1}$
  \end{tabular}
  \vspace{0.2em} \\ 
  \begin{tabular}{c}
    Georgios Piliouras$^{4}$
  \end{tabular}
  \hspace*{1em}
  \begin{tabular}{c}
    Andre Wibisono$^{5}$
  \end{tabular}
    \hspace*{1em}
  \begin{tabular}{c}
    Juba Ziani$^{1}$
  \end{tabular}
}

\begin{document}
\maketitle
\blfootnote{
    Affiliations:\;\;    
    $^{1}$ Georgia Tech \;\;
    $^{2}$ SUTD \;\;
    $^{3}$ Google Research\;\;
    $^{4}$ Google DeepMind \;\;
    $^{5}$ Yale \;\;
  \vspace*{0.1em}
}
\blfootnote{%
Contact: \qquad
\texttt{jlazarsfeld@gatech.edu,
    anas\_barakat@sutd.edu.sg, 
    prof@gatech.edu, 
    vmuthukumar8@gatech.edu \\
    \hspace*{8em}
    gpil@deepmind.com, 
    andre.wibisono@yale.edu, 
    jziani3@gatech.edu.
}}


\begin{abstract}
\vspace*{0.5em}
We prove that the last iterate of Optimistic Follow-The-Regularized-Leader
(OFTRL) converges at a 
\textit{linear rate}
to a Nash equilibrium in multiplayer harmonic games.
Our result holds 
under a mild game-regularity condition and for a large family of regularizers,
including entropy, Tsallis, and log-barrier.
Our analysis is based on a dual perspective
of OFTRL: we prove that the interaction between
optimism and the global structure of harmonic
games leads to the dissipation of an 
energy function
along the dual iterates,
for which we establish tight quantitative bounds.
Our result also holds for the
Extragradient FTRL variant.
\end{abstract}

\newpage

\section{Introduction}
\label{sec:intro}

The last-iterate behavior of online learning algorithms 
in strategic multiagent settings has attracted 
growing attention in the study of learning in games. 
Unlike regret or time-average results,
a last-iterate guarantee concerns the
actual strategies played by agents over time,
and the distinction is often consequential:
while standard no-regret algorithms
that are run for
$T$ iterations may enjoy quantitative convergence
guarantees for \emph{time-averaged} play,
the \emph{$T$-th iterate itself} may fail to
stabilize and can even be
repelled from equilibrium and spiral out towards 
the boundary of the simplex (see, e.g., \cite{bailey2018multiplicative,
daskalakis-et-al18training-gans,cheung2019vortices}).

Partly motivated by this gap, the past decade has
seen substantial work on the design and analysis of 
\textit{optimistic}  online learning 
algorithms~\citep{rakhlin2013optimization}.
Broadly speaking, optimistic methods use recent 
feedback to predict future losses.
Within this line of work, 
\textit{Optimistic Follow-the-Regularized-Leader} (OFTRL) 
is a prominent family of methods which augments the 
classical FTRL framework with an optimistic prediction 
and naturally accommodates constrained, non-Euclidean geometries
~\citep{shalev2012online}. With entropic regularization,
OFTRL recovers the popular
\textit{Optimistic Multiplicative Weights Update} (OMWU)
algorithm.
In the narrow, yet canonical setting
of two-player zero-sum games, the last-iterate behavior
of OMWU is now comparatively well understood:
asymptotic convergence in the bimatrix
case was shown by~\cite{DaskalakisP19},
subsequently strengthened to a \textit{linear}
rate by~\cite{wei2021linear} and
extended to the extensive-form setting by
\cite{lee2021last}.
More recent work has also given
a finer geometric characterization of the
last-iterate rates in the same bimatrix setting~\citep{lazarsfeld2026}.

Despite this detailed picture in two-player zero-sum games,
the last-iterate behavior of OFTRL beyond zero-sum settings,
and beyond just OMWU, remains much less understood.
This is in contrast to 
the wealth of
regret guarantees for OMWU in
two-player zero-sum \citep{tsuchiya26tight-2pszs} and multiplayer general-sum games
\citep{syrgkanis2015fast,daskalakis2021near,anagnostides2022near,
soleymani2025faster}.
This gap raises a broader question: 
which structural properties of general multiplayer games 
enable fast \textit{last-iterate} convergence of
FTRL-type dynamics to equilibrium?

Toward this direction,
\cite{legacci2024no} recently showed
that some
of the dynamical properties
of OFTRL in two-player zero-sum games extend to the broader class of \textit{multiplayer harmonic games}.
These games capture strategic interactions
between agents with conflicting interests,
and they form one of the fundamental components in the
decomposition of normal-form games into their
potential and harmonic parts
\citep{candogan2011flows, abdou2022decomposition}.
While two-player zero-sum games with
an interior equilibrium form a subclass,
harmonic games need not be constant-sum,
may admit non-convex equilibrium sets, and can induce 
nonlinear, non-monotone game operators.
Nevertheless,
\cite{legacci2024no} proved
in this broader setting
that the last iterate of both OFTRL and the related
\textit{Extragradient FTRL} (EGFTRL) variant
\textit{converges asymptotically}
to the game's Nash equilibrium set.
However, establishing \textit{convergence rates} 
has remained open. We thus ask:
\vspace*{-0.25em}
\begin{center}
    \it
    How fast do optimistic FTRL-type
    dynamics converge in harmonic games,\\
    and what structural mechanisms govern their convergence?
\end{center}

\subsection{Our Contributions and Techniques}

The main contributions of our paper
give quantitative answers to the above question.

\paragraph{Last-iterate convergence rate for
OFTRL, EGFTRL.}
We prove that both
OFTRL and EGFTRL
\textit{converge at a linear last-iterate rate to a 
Nash equilibrium in multiplayer harmonic games}
(Theorem~\ref{thm:main-last-iterate}).
Specifically, when instantiated with
a suitable regularizer $Q$,
and constant stepsize $\eta$,
the iterates $\{w_t\}$ of the algorithms
converge to an equilibrium $w^\star$
and satisfy
\begin{equation*}
    \textbf{Theorem~\ref{thm:main-last-iterate} (Informal):}
    \;\;
    \D_{Q}(w^\star, w_{t+1})
    \le
    2 \kappa \D_{Q}(w^\star, w_0)
    \cdot
    \exp(- \gamma \eta^2 \cdot t)
    \;\;\text{for all $t \ge 0$} \;.
\end{equation*}
Here, $\D_Q$ is the Bregman divergence
induced by $Q$, and $\gamma > 0$ and $\kappa >0$
are instance-dependent constants.
Note that convergence rates
in $\D_Q$ also imply rates in the $\Gap$ function
(Remark~\ref{rem:gap-vs-bregman}).
This result establishes the first global,
quantitative last-iterate guarantees
for FTRL variants in harmonic games.
Moreover, in contrast to prior work specific to two-player
zero-sum games~\citep{DaskalakisP19, wei2021linear, lazarsfeld2026},
we do not assume existence of a \textit{unique} Nash equilibrium, and we
additionally give an exact characterization of
the equilibrium selected by the dynamics as a function
of the initialization,
strengthening the asymptotic
convergence result of~\cite{legacci2024no}.
Our result holds for a broad class of Legendre
regularizers (including negative Shannon entropy,
Tsallis entropies of order $q \in (0, 1)$,
and the log-barrier function), and under a mild
non-degeneracy condition on the game.
This condition always holds for the linear
polymatrix setting, as well as in other non-linear
settings.

\paragraph{Dual and variational perspective.}
Our analysis is based on a dual
perspective on OFTRL and EGFTRL.
We prove that the convex
conjugate of the regularizer plays the role 
of an \textit{energy function} that
dissipates along the dual iterates
(Lemma~\ref{lem:energy-dissipation-main}).
This energy dissipation corresponds to a contraction of the
Bregman divergence from equilibrium along the primal
iterates (Proposition~\ref{prop:bregman-shifted-energy}).
Dissipation
stems from the action of optimism
on the global structure of harmonic games.
Structurally, the \textit{loss operators} of
these games (defined in~\eqref{eq:G-def})
satisfy a \textit{weighted}
Minty variational inequality
\citep{minty1967generalization}
that holds globally with equality,
a property shown by~\cite{legacci2024no}
and restated in Proposition~\ref{prop:harmonic-mixed}.
In continuous time, due to this global harmonic structure, 
energy is conserved along the FTRL flow in a scaled coordinate system.

However, in discrete time, we show that an 
idealized \textit{Implicit FTRL} (IFTRL) 
discretization of the FTRL flow \textit{dissipates} energy.
By coupling their iterates with IFTRL,
we prove that dissipation
persists under both OFTRL and EGFTRL with a sufficiently small
constant stepsize.
Our dual viewpoint and proof techniques are
inspired from the recent work of~\cite{lazarsfeld2026} who studied OMWU 
in two-player zero-sum games. However, the special
skew-symmetric structure of zero-sum game operators
is not available in harmonic games, and we develop
a different coupling analysis technique with the implicit method.
Moreover, our analysis and convergence rates
apply to \textit{multiplayer} games with 
possibly \textit{multiple} equilibria,
and they hold for both Optimistic \textit{and} Extragradient FTRL
with regularizers beyond negative Shannon entropy.

\paragraph{Quantitative separation for vanilla FTRL.}
Our techniques additionally establish
quantitative guarantees for
vanilla (non-optimistic) FTRL
in polymatrix harmonic games:
here, we prove that the FTRL iterates \textit{move away}
from equilibrium at a geometric rate
within any fixed Bregman divergence sublevel set
(Theorem~\ref{thm:ftrl-diverge-main}).
This extends and makes quantitative
the asymptotic result
of~\cite{bailey2018multiplicative}
that was proven for two-player zero-sum games.
Taken together, our results give a unified
quantitative picture of the last-iterate behavior
of FTRL-type dynamics under the global harmonic structure.
We illustrate our results further
in Figure~\ref{fig:intro}.

\captionsetup[figure]{labelfont=footnotesize}
\begin{figure}[t!]
\centering
\vspace*{-1.5em}
\includegraphics[width=0.85\textwidth]{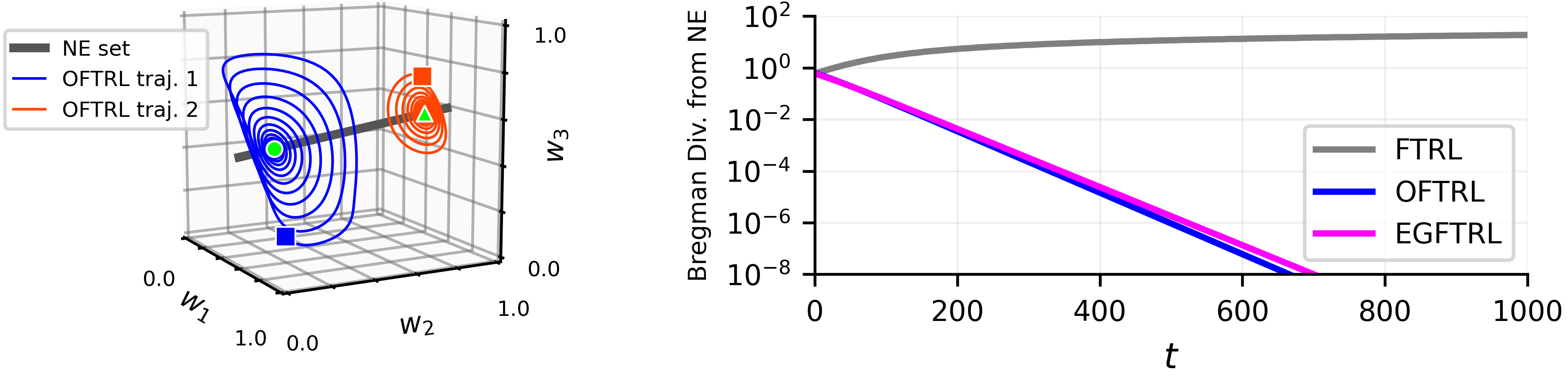}
\vspace*{-1em}
\caption{%
    \footnotesize
    Variants of FTRL (with
    entropy regularizer) on 
    the 3-player harmonic game
    in~\ref{app:harmonic-ex1}.
    Left: two trajectories of OFTRL
    converging to 
    their distinct analytic limits (in green)
    from Theorem~\ref{thm:main-last-iterate} (i).
    Right: Bregman
    divergence from NE for 
    the $t$-th iterate of
    OFTRL, EGFTRL, and FTRL from a 
    fixed initialization.
    OFTRL and EGFTRL converge linearly
    (Theorem~\ref{thm:main-last-iterate}(ii)),
    while FTRL diverges 
    at a geometric rate (Theorem~\ref{thm:ftrl-diverge-main}).
    }
\label{fig:intro}
    \vspace*{-1em}
\end{figure}

\smallskip
\subsection{Related Work}

Compared to the FTRL family, the
last-iterate behavior of
optimistic and extragradient
\textit{gradient-descent-ascent} methods is better understood
(see, e.g.,~\cite{mokhtari-et-al20unified,
golowich2020tight,cai2022finite,gorbunov-et-al22mvi}).
In constrained settings, however, these Euclidean schemes are generally distinct from FTRL dynamics (see, e.g., \citet{FGKLLZ24}).
For FTRL and their closely related variants, the iterate-based guarantees most relevant to our setting are either asymptotic
\citep{mertikopoulos-et-al18optim-MDextra-mile, hsieh2021adaptive},
concern only the \textit{best iterate}
\citep{anagnostides2022near},
or provide \textit{local} convergence rates
\citep{lei-et-al21,azizian2024rate}.
We discuss several of these lines of 
related work in more detail below.

\paragraph{Optimistic and extragradient methods.}
Optimistic and extragradient
methods\footnote{
Throughout, we colloquially say
\textit{optimism} to
describe both types of variants.
} have a long history in optimization \citep{korpelevich1976extragradient,popov1980modification,nemirovski2004prox,nesterov2007dual}. 
Our proof strategy of coupling OFTRL and Extragradient FTRL with IFTRL 
is reminiscent of the unified analysis proposed in \cite{mokhtari-et-al20unified}. 
Showing that both OGDA and the extragradient method 
can be interpreted as approximate variants of the proximal point method, 
\citet{mokhtari-et-al20unified} prove linear last-iterate convergence rates for both 
optimistic and extragradient methods. In contrast to their \textit{unconstrained} 
convex-concave min-max problem, our setting is \textit{constrained}. The 
interpretation of the extragradient method as an approximation of the Proximal Point 
method was first discussed and developed in \cite{nemirovski2004prox}. 

\paragraph{Last-iterate convergence in monotone variational inequalities and beyond.}
Last iterate convergence results for optimistic gradient and extragradient methods 
were proved in \citet{golowich2020tight,golowich-et-al20last-it-slower} and later 
tightened under more relaxed conditions 
\citep{gorbunov-et-al22last-it-mon-vi-extragd,gorbunov-et-al22mvi,cai-et-al22tight}. 
A line of works also established 
convergence guarantees \textit{beyond} monotone variational inequalities (see, e.g., 
\cite{gorbunov-et-al23,pethick-et-al25} and the references therein).

\paragraph{Other works on optimistic online learning methods.}
Optimistic variants of online learning algorithms
have also been studied for their connections
to accelerated methods in optimization
\citep{WA07,wang2024no}.
Recent work has also 
analyzed the optimistic FTRL
methods in symmetric cone games~\citep{barakat-et-al26tmlr-scg}.

\paragraph{Dual and energy-based perspectives
on learning in games.}
There has been a growing body of work 
in the past decade that studies learning dynamics
in games using dual and physics-inspired perspectives,
in particular for FTRL-type dynamics;
see, e.g.,~\cite{mertikopoulos2018cycles},~\cite{bailey2018multiplicative},
~\cite{bailey2019fast},~\cite{bailey2019multi},
~\cite{wibisono2022alternating}, 
~\cite{azizian2024rate},~\cite{lazarsfeld2025optimism,lazarsfeld2025fp},~\cite{katona2024symplectic},~\cite{BAILEY2026676},
and the references therein. 

\paragraph{Learning in harmonic games.}
We refer the reader to \cite{legacci2024geometric}
for a discussion on the dynamics of harmonic games within a dynamics-driven 
decomposition of games, where it is further shown that the continuous-time 
exponential weight dynamics are Poincaré recurrent in harmonic games. 
Besides~\cite{legacci2024no},
several recent works have developed convergence results in harmonic games. 
Under a weighted Minty condition extending the 
standard Minty condition in variational inequalities, 
\cite{zhang-et-al25scale} provide \textit{best-iterate} 
guarantees for parameter-free variants of predictive regret matching. 
In particular, this result applies to harmonic games
which satisfy the aforementioned weighted Minty condition. 
Taking a computational perspective, \cite{anagnostides-et-al25polynomial-vi-minty} 
provide a polynomial-time algorithm for 
computing approximate Nash equilibria in harmonic games.

\paragraph{Last-iterate convergence 
results for FTRL-type dynamics.}
There has been a series 
of last-iterate results obtained
for optimistic and extragradient FTRL
variants in several continuous settings,
but only with asymptotic guarantees;
see, e.g.,
\cite{mertikopoulos-et-al18optim-MDextra-mile},
\cite{lei-et-al21}, \cite{hsieh2021adaptive},
\cite{fasoulakis-et-al22flbr-mwu}.
Recent works of~\cite{FGKLLZ24,cai-et-al26separation} 
have also established a quantitative
separation in the last-iterate, best-iterate,
and random-iterate guarantees
of OFTRL dynamics in two-player zero-sum games.
Recently, \cite{fasoulakis-et-al2026flbr-mwu-last-it-rate}
also obtained a linear last-iterate convergence
rate in two-player zero-sum games 
for a certain extragradient variant of MWU.
\cite{wang-et-al26last-it-sym-gen-sum-2x2}
proved last-iterate guarantees
for the~\textit{vanilla} variant of MWU
in two-player, two-action \textit{symmetric games},
which comes in contrast to
the divergence result 
for vanilla FTRL that we prove for harmonic games.

\paragraph{Further results on last-iterate convergence in min-max optimization.} 
Recently, beyond existing results in two-player zero-sum games,
\cite{orabona26last-it-cv-omwu} proved asymptotic last-iterate convergence of OMWU 
with a small enough constant stepsize in convex-concave saddle point problems. 
We refer the reader to Section 2 in \cite{orabona26last-it-cv-omwu} for a more 
detailed discussion of related work on existing last-iterate convergence results.
Beyond static two-player zero-sum games,~\cite{wu-et-al25sicon-lastit-lqg}
consider finite-horizon zero-sum linear quadratic dynamical games 
and establish a linear last-iterate convergence rate in terms of the 
objective function-value gap for a \textit{nested} natural policy gradient algorithm. 
Their proof technique relies on using a gradient domination inequality
proved for the primal function. We note that their algorithm 
is nested and the two players need to take turns whereas 
we consider a multiplayer setting where players make decisions simultaneously.


\section{Preliminaries and Setting}
\label{sec:prelims}

\paragraph{Notation.}
For any integer $k \ge 1$, let $[k] = \{1, \dots, k\}$.
Let $\Delta_k \subset \R^k$ denote the probability simplex, and let 
$\relint(\Delta_k) = \{q \in \Delta_{k} : q(j) > 0, \forall j \in[k] \}$
denote its relative interior.
We say that $q$ is \textit{interior} when $q \in \relint(\Delta_k)$.
Let $\1_k \in \R^k$ and $I_k \in \R^{k \times k}$
denote the $k$-dimensional all-ones vector
and identity matrix, respectively.
Let $\langle \cdot, \cdot \rangle$ be the standard
$\ell_2$ inner product and $\|v\|_2 = \sqrt{\langle v, v\rangle}$. 
If $M$ is positive semi-definite, we write $M \succeq 0$.
For a linear subspace $\calU \subseteq \R^k$, 
let $\Pi_{\calU} : \R^k \to \calU$ denote the orthogonal 
projection onto $\calU$.
For $g: \calU \to \calU$ and $u \in \calU$,
let $J_{g}(u): \calU \to \calU$ denote the Jacobian
of $g$ at $u$ when it exists.
For a convex set $\calC \subseteq \R^k$ and a  differentiable convex
$f: \calC \to\R$,
the \textit{Bregman divergence} 
$\D_f : \calC \times \relint(\calC) \to \R$ is given by
$\D_f(y', y) = f(y') - f(y) - \langle \nabla f(y), y'-y \rangle$.

\subsection{Multiplayer Normal-Form Games}

Finite normal-form games consist of $N \ge 2$ players.
Each player $i \in [N]$
has a set $\calA_i := \{\alpha_j\}_{j \in [d_i]}$ of $d_i$ actions,
and under a joint action profile  
$a = (a_1, \dots, a_N) \in \calA := \prod_{i \in [N]} \calA_i$,
player $i$ incurs a loss $\ell_i(a)$, 
where $\ell_i: \mathcal{A} \to \mathbb{R}$ is the player's loss function.  
Each player~$i$ can choose
a mixed strategy $w_i \in \calW_i := \Delta_{d_i}$
over $\calA_i$.
Using the joint space 
$\calW := \prod_{i \in [N]} \calW_i$,
we overload notation 
and let $\ell_i: \calW \to \R$ denote
player $i$'s expected loss defined as $\ell_i(w) = \E_{a \sim w} [\ell_i(a)]$ 
for any joint strategy profile  
$w = (w_1, \dots, w_N) \in \calW\,.$ 
We set $D := d_1 + \dots + d_N$
and use the standard 
notation $\calA_{-i} := \prod_{j \neq i \in [N]} \calA_{j}$
and $\calW_{-i} := \prod_{j \neq i \in [N]} \Delta_{d_j}$.

\paragraph{Nash equilibrium.} 
Our primary solution concept will be the celebrated Nash equilibrium:
For $\eps \ge 0$, $w^\star \in \calW$ is an $\eps$-approximate Nash equilibrium ($\eps$-NE) if $\max_{i \in [N]}\,
\max_{w_i \in \calW_{i}} 
\{\ell_i(w^\star) - \ell_i(w_i, w_{-i}^{\star})\} \le \eps$.
Due to~\cite{nash1951non}, every finite normal-form game has at
least one ($\eps=0$) NE.

\subsection{Harmonic Games}
\label{sec:prelims:harmonic}

We focus on the class of
multiplayer \textit{harmonic games} defined as follows.

\begin{restatable}
    [Harmonic game]
    {define}{defharmonic}
	\label{def:harmonic-finite}
	A finite normal-form game is \textit{harmonic} if 
	for each player $i \in [N]$, there exists a positive measure 
	$\mu_i = (\mu_{i}(1), \dots, \mu_{i}(d_i)) \in (0, \infty)^{d_i}$ such that
	\begin{equation}
		\textstyle
		\sum_{i \in [N]} \,
		\sum_{j \in [d_i]} \, 
		\mu_{i}(j) \cdot \big( 
			\ell_i(a) - \ell_i(\alpha_{j}, a_{-i})
		\big) = 0 
		\;\;\text{for all $a = (a_i, a_{-i}) \in \calA$} \;.
		\label{eq:harmonic-finite}
	\end{equation}
\end{restatable}

Definition~\ref{def:harmonic-finite} is the \textit{weighted}
definition of harmonic games (introduced by~\cite{abdou2022decomposition}
and considered in~\cite{legacci2024no}),
extending the unweighted definition of~\cite{candogan2011flows}.
Intuitively,~\eqref{eq:harmonic-finite} imposes a global balance condition: at every action profile, the ${\mu_i}$-weighted sum of the loss differences induced by all unilateral deviations is zero. 

\begin{restatable}
	[Examples]
	{rem}{remharmonicexamples}
	\label{rem:harmonic-examples}
    Harmonic games (Definition~\ref{def:harmonic-finite})
	are in general multiplayer, 
	non-constant-sum, and can have 
	non-convex sets of Nash equilibria. 
	We provide such families of examples in Section~\ref{app:harmonic:examples}.
    An important subclass
    is two-player zero-sum games with
	at least one interior NE
    (but not every zero-sum game is harmonic);
    See Proposition~\ref{prop:zsg-interior}.
\end{restatable}

\subsection{Variational Characterization of Harmonic Games} 
\label{sec:harmonic:variational}

Our analysis will rely on the variational formulation of finite games through their associated game operator~\citep{facchinei2003finite}. In this language, the harmonicity condition of Definition~\ref{def:harmonic-finite} admits an equivalent characterization that will be central to the analysis that follows. We introduce the relevant variational notation and discuss this characterization below.

\paragraph{Joint loss operator.}
For a finite normal-form game, 
we define its joint loss operator $G: \calW \to \R^D$. 
For any $w \in \calW$,
the $i$-th block-coordinate $G_i(w)$ of $G(w)$ 
is given by  
$g_i: \calW_{-i} \to \R^{d_i}$,
where $g_i(w_{-i}) = \big(
	\E_{a_{-i} \sim w_{-i}}[\ell_i(\alpha_j, a_{-i})]
\big)_{j \in [d_i]} \in \R^{d_i}$.
In particular, as $\ell_i(w) = \langle w_i, g_i(w_{-i}) \rangle$, 
observe that each $g_i(w_{-i}) = \nabla_{w_i} \ell_i(w)$
is exactly the gradient of the $i$-th player's
loss with respect to their own action. 
The map $G$ is then the concatenation of these
individual loss gradients,~i.e.
\begin{equation}
	G(w) := (
		\nabla_{w_1} \ell_1(w),
		\dots,
		\nabla_{w_N} \ell_N(w)
	) \in \R^D
	\;\;\text{for all $w \in \calW$}\;.
	\label{eq:G-def}
\end{equation}
We associate each game 
with its operator $G$ and make standard regularity assumptions.
For all $i \in [N]$, let each $\calW_i$
be equipped with norm $\|\cdot \|_{i}$ and dual norm 
$\|\cdot\|_{i, *}$.
Then we define the global product norms
$\|w\| = (\sum_{i=1}^N \|w_i\|^2_{i})^{1/2}$
for $w \in \calW$
and $\|z\|_{*} = (\sum_{i=1}^N \|z_i\|^2_{i, *})^{1/2}$
for $z \in \R^D$.
We then assume that 
$G$ is $L$-Lipschitz and $B$-bounded with respect to these norms.

\smallskip
\begin{restatable}[Lipschitz and bounded operator]
	{ass}{assoperator}
	\label{ass:operator}
	There exist constants $L \ge 1$ and $B \ge 1$ such that
	(i) $\|G(w') - G(w)\|_{*} \le L\|w'-w\|$
	and (ii) $\|G(w)\|_{*} \le B$
	for all $w, w' \in \calW$.
\end{restatable}

Assumption~\ref{ass:operator} is always satisfied when the losses $\ell_i$ are smooth, 
with the global constants $L$ and $B$ inherited 
from the corresponding regularity bounds on the losses.

\paragraph{Variational characterization of NE.}
The operator $G$ also provides a standard variational
characterization of Nash equilibrium: $w^\star \in \calW$ is a Nash
equilibrium (for $\eps=0$) if and only if
\begin{equation}
	\langle G(w^\star), w- w^\star \rangle \ge 0
	\;\;\text{for all $w \in \calW$} \;.
	\label{eq:nash-svi}
\end{equation}
We will thus denote by $\NE(G) \subseteq \calW$ the set of all $w^\star$ 
satisfying~\eqref{eq:nash-svi}.
Moreover, for any point $w \in \calW$,
we can measure its distance to a Nash equilibrium
via the \textit{gap function}: for any $w \in \calW$, let  
\begin{equation}
	\textstyle
	\Gap(w) 
	:= 
	\max_{\wt w \in \calW} \,
	\langle G(w), w  - \wt w \rangle \;.
	\label{eq:gap}
\end{equation}
If $\Gap(w') \le \eps$, then 
$w'$ is an $\eps$-NE.

\paragraph{Variational characterization of harmonic games.}
As shown by \cite{legacci2024no}, 
harmonic games from Definition~\ref{def:harmonic-finite} also 
admit the following variational characterization
in terms of $G$:

\begin{restatable}[\cite{legacci2024no}, Proposition A.4]
	{prop}{propharmonicmixed}
	\label{prop:harmonic-mixed}
	A game with loss operator $G: \calW \to \R^D$
	is harmonic if and only if there exist 
	(i) a matrix $M = \Diag(m_1 I_{d_1}, \dots, m_N I_{d_N})$
	where $m_i > 0$ for all $i \in [N]$
	and (ii) an interior point $w' \in \relint(\calW)$ such that:
	\begin{equation}
		\langle M G(w), w- w' \rangle = 0 
		\;\;\text{for all $w \in \calW$} \;.
		\label{eq:harmonic-wmeq}
	\end{equation}
\end{restatable}

Intuitively,~\eqref{eq:harmonic-wmeq}
says that, relative to a center $w'$, 
the $M$-weighted operator is
orthogonal to all displacements $w-w'$,
expressing in variational terms
the weighted cancellation of unilateral deviations
in~\eqref{eq:harmonic-finite}.
We give more background in Section~\ref{app:prelims:harmonic},
and we highlight here several additional points: 
\begin{itemize}[
	leftmargin=1.5em,
	topsep=-0.25em,
	itemsep=0.05em
]
\smallskip
\item
\textit{Equilibrium condition}:
Any $w' \in \calW$ satisfying the proposition
is an NE (see Lemma~\ref{lem:wmeq-ne}), and thus~\eqref{eq:harmonic-wmeq} is simultaneously
a structural property of $G$ as well as an equilibrium condition. 
\item
\textit{Weight matrix}:
We call $M$ from the proposition 
a \textit{weights matrix}, 
and note that  each $m_i = \|\mu_i\|_1$ 
for $\mu_i$ from Definition~\ref{def:harmonic-finite}.
For fixed $M$, we let $m_{\max} = \max_{i \in [N]} m_i$
and $m_{\min} = \min_{i \in [N]} m_i$.

\item
\textit{Operator properties:}
Operators satisfying~\eqref{eq:harmonic-wmeq}
are generally non-monotone and nonlinear, even for multiplayer normal-form games.
The case of linear $G$
(meaning $G(w) = Pw$)
corresponds
to \textit{polymatrix games}
with the harmonic property;
see Section~\ref{app:harmonic:examples}.
\end{itemize}

\smallskip
\smallskip
\paragraph{Weighted Minty equality.} 
Note that~\eqref{eq:nash-svi} is \textit{Stampacchia's variational inequality} while~\eqref{eq:harmonic-wmeq} is 
a weighted Minty variational \textit{equality} for the operator $G$ 
(see Section~\ref{app:prelims:variational} for background).
For this reason, it is useful to define the set of 
\textit{all solutions} of~\eqref{eq:harmonic-wmeq},
which we dub \textit{weighted Minty equality} points.
For a fixed weights matrix $M$, let 
\begin{equation}
	\WMEQ_{M}(G) 
	= 
	\{w' \in \calW \,:\,
	\langle MG(w), w- w' \rangle = 0
	\;\text{for all $w \in \calW$}
	\}\;.
	\label{eq:wmeq}
\end{equation}
By Proposition~\ref{prop:harmonic-mixed}, a game is harmonic
iff
$\WMEQ_M(G) \cap \relint(\calW) \neq \emptyset$
for a weights matrix~$M$.
Due to Lemma~\ref{lem:wmeq-ne},
we have $\WMEQ_M(G) \subseteq \NE(G)$,
but in general~$\NE(G) \neq \WMEQ_M(G)$. 
Moreover, we note 
the generality of this variational perspective
in terms of the underlying game.

\smallskip
\begin{restatable}
	[Generality of variational perspective]
	{rem}{remgeneralityvariational}
	\label{rem:generality-variational}
	Our main results
 	only use the characterization
 	of Proposition~\ref{prop:harmonic-mixed}
 	under the conditions on $G$ from 
 	Assumption~\ref{ass:operator}.
 	More generally, our results extend beyond
    normal-form harmonic games to the setting where the individual
    loss functions $\ell_i: \calW \to \R$ 
    are each convex in their $i$'th coordinate $w_i$,
    where $G$ is the resulting
 	pseudo-gradient map 
 	in~\eqref{eq:G-def},
    and where 
    $\WMEQ_M(G) \cap \relint(\calW) \neq \emptyset$
 	for a weight matrix $M$.
 	By Proposition~\ref{prop:harmonic-mixed},
 	normal-form games
 	satisfying Definition~\ref{def:harmonic-finite} are one such instantiation.
\end{restatable}

\section{A Global View of 
Optimistic and Extragradient FTRL}
\label{sec:algorithms}

We now introduce the global dynamics of the family of
OFTRL and EGFTRL algorithms
where each player $i \in [N]$
uses a strictly convex regularizer $Q_i : \calW_i \to \R$
and a fixed stepsize $\eta_i > 0$.
We can then define the induced
\textit{global regularizer} $Q: \calW \to \R$
and \textit{global stepsize} $\eta > 0$,
which together 
reduce the simultaneous 
playerwise updates into a single primal-dual system.
We develop this global view here and 
defer the full details on the 
playerwise updates to Section~\ref{app:ftrl-prelims:system}.

\paragraph{Global stepsize and regularizer.}  
Let $\eta := \max_{i\in [N]} \eta_i >0$,
and let  $Q: \calW \to \R$ be the separable regularizer
that absorbs the player-specific stepsizes via
\begin{equation}
	\textstyle
	Q(w)
	= 
	\sum_{i=1}^N
	\frac{\eta}{\eta_i} Q_i(w_i)
	\;\;\text{for all $w = (w_1, \dots, w_N) \in \calW$} \;.
	\label{eq:global-reg}
\end{equation}

\paragraph{Global primal and dual updates.}
Using $\eta$ and $Q$ from~\eqref{eq:global-reg},
the global primal update for FTRL-type algorithms 
admits a similar structure to the standard playerwise definitions. 
Formally, fix initial $w_0 \in \relint(\calW)$.
Then for $t \ge 1$, let $w_{t} := (w_{t, i})_{i\in [N]} \in \calW$
be the \textit{global primal iterate} given by
\begin{equation}	
	\textstyle
	w_{t} 
	= 
	\argmin_{w \in \calW} \;
	\{ 
		- \langle w, z_{t} \rangle + Q(w) 
	\} 
	\;\;\;
	\text{for all $t \ge 1$.}
	\label{eq:ftrl-primal}
\end{equation}
Here, $z_{t} = (z_{t, i})_{i\in[N]} \in \R^D$ 
is a variant-specific \textit{global dual iterate}
that accumulates operator vectors scaled by~$\eta$.
For example, the vanilla FTRL variant sets 
$z_t := z_0 - \eta \sum_{k=0}^{t-1} G(w_k)$ for
all $t \ge 1$ and initial $z_0 \in \R^D$.
More generally, for initial $z_0 \in \R^D$ and all $t \ge 0$,
the dual sequences $\{z_t\}$ for
the vanilla (FTRL), implicit (IFTRL),
optimistic (OFTRL), and extragradient (EGFTRL)
variants can each be written recursively
via the update rules:
\begin{align}
	z_{t+1}
	&= z_t - \eta \, G(w_t) 
	\label{eq:ftrl}
	\tag{FTRL}\\
	z_{t+1}
	&= z_t - \eta \, G(w_{t+1})
	\label{eq:iftrl}
	\tag{IFTRL} \\
	z_{t+1} 
	&= z_t - 2\eta \, G(w_t) + \eta \, G(w_{t-1})  
	\label{eq:oftrl}
    \tag{OFTRL} 
    \\
    z_{t+1/2}
    &= 
    z_{t} - \eta \, G(w_t) 
    \,,\quad
    z_{t+1} 
    = z_t - \eta \, G(w_{t+1/2}) \;.
	\label{eq:egftrl}
    \tag{EGFTRL} 
\end{align}
Here, we use for~\eqref{eq:oftrl} 
the notational convention $w_{-1} := w_0 \in \relint(\calW)$,
and for~\eqref{eq:egftrl} the convention that
$w_{t+1/2}$ is the minimizer of the
objective in~\eqref{eq:ftrl-primal} with respect to $z_{t+1/2}$.

\smallskip
\begin{restatable}
    [Intuition for optimism]
    {rem}{remoptimismintuition}
    Intuitively, both OFTRL and EGFTRL
    augment the vanilla FTRL update
    in an attempt to \textit{predict} the next vector $G(w_{t+1})$. 
    Thus, both methods can be viewed as
    approximations of the idealized (unimplementable)
    IFTRL update (see also Section~\ref{sec:proof-overview:energy}).
\end{restatable}

\subsection{Assumptions on regularizer}
\label{sec:algorithms:assumptions}

We study instantiations of
FTRL variants with regularizers $Q$ satisfying
the following property:

\begin{restatable}[Relatively Legendre]
	{ass}{asslegendre}
	\label{ass:legendre}
	$Q$ is 
	(i) strictly convex and twice continuously differentiable on
	$\relint(\calW)$, and (ii) $\|\nabla Q(w) \| \to + \infty$
	as $w$ approaches the relative boundary of
	$\calW$.
\end{restatable}

Let $H := Q^* : \R^D \to \R$ denote the
\textit{convex conjugate} of $Q$,
which we refer to as an \textit{energy function}
following prior work~\citep{wibisono2022alternating}.
Under Assumption~\ref{ass:legendre},
$H$ is continuously differentiable and 
admits the following key primal-dual relationship
for FTRL variants
(see Section~\ref{app:ftrl-prelims:primal-dual}):

\smallskip
\begin{restatable}
	[Primal-dual invariant]
	{prop}{propftrlprimaldual}
	\label{prop:ftrl-primal-dual-invariant}
	Let $\{w_t\}$ be the iterates
	from~\eqref{eq:ftrl-primal} 
	using the dual sequence $\{z_t\}$
	of~\eqref{eq:ftrl},~\eqref{eq:iftrl},
	\eqref{eq:oftrl}, or~\eqref{eq:egftrl}.
	Then $w_{t} = \nabla H(z_t) \in \relint(\calW)$
	for all $t \ge 1$.
\end{restatable}

\begin{restatable}[OMWU]
	{rem}{remlegendreexamples}
	\label{rem:legendre-examples}
    The (componentwise) negative Shannon entropy
    is the canonical $Q$ satisfying
	Assumption~\ref{ass:legendre}.
	Here, the energy $H = Q^*$ is the \textit{log-sum-exp} 
	function, and $\nabla H (z): \R^D \to \relint(\calW)$
	is the \textit{softmax map}.
	This yields the MWU algorithm 
	and its optimistic variants.
    See Section~\ref{app:ftrl-prelims:examples}.
\end{restatable}

Our main convergence results 
also require two additional assumptions
on $Q$ and $H$.

\begin{restatable}[1-strong-convexity]
	{ass}{assregsc}
	\label{ass:reg-sc}
	 $\forall w \in \relint(\calW), w' \in \calW$, 
	$\D_Q(w',w) \ge \frac{1}{2} \|w'-w\|^2$.
\end{restatable} 
\begin{restatable}
	[Local Hessian stability]
	{ass}{asslhs} 
	\label{ass:lhs}
	$H$ is twice continuously differentiable, 
	and $\exists\; \omega \ge 1$ s.t.
	\begin{equation}
		\forall z, z' \in \R^D:\;\;
		\exp(-\omega  \|z-z'\|_{*})\cdot
		\nabla^2 H(z') 
		\preceq
		 \nabla^2 H(z)
		\preceq
		 \exp(\omega \|z-z'\|_{*}) \cdot
		\nabla^2 H(z') 
		\;.
		\label{eq:lhs}
		\tag{LHS}
	\end{equation}
\end{restatable}

Assumption~\ref{ass:reg-sc} is standard
and readily generalizes
to $K$-strong-convexity for $K > 0$.
Assumption~\ref{ass:lhs} specifies 
a local notion of third-order 
smoothness related to
\textit{generalized self-concordance}~\citep{sun2019generalized},
which intuitively controls the variation of the local dual
geometry induced by $H$.
Due to the separability of $Q$,
Assumptions~\ref{ass:legendre},
\ref{ass:reg-sc}, \ref{ass:lhs}
all hold whenever each $Q_i$ and $H_i := Q^*_i$
satisfy the componentwise properties,
including for the following canonical choices:

\begin{restatable}[Examples]
    {rem}{remlegendrereg}
    \label{rem:legendre-reg-examples}
    Regularizers satisfying
    Assumptions~\ref{ass:legendre},~\ref{ass:reg-sc},
    and~\ref{ass:lhs}
    include negative Shannon entropy,
    Tsallis entropy of order $q \in (0,1)$, 
    and the log-barrier regularizer;
    see Section~\ref{app:ftrl-prelims:examples}.
\end{restatable}

\subsection{Geometry of Primal and Dual Iterates}
\label{sec:ftrl:geometry}

Our analysis of FTRL variants also 
relies on several geometric properties
that we introduce here:

\paragraph{Dual update space and reachable primal set.}
Each dual iterate of the FTRL variants satisfies
$z_t - z_0 \in \calZ := \Span(G(\calW))$,
and thus we call the linear subspace $\calZ \subseteq \R^D$
the \textit{dual update space}.
Similarly, for initializations
with $w_0 = \nabla H(z_0)$, 
each $w_t \in  \calR(w_0) := \{\nabla H(z_0 + z) : z \in \calZ\}$,
and thus we refer to 
$\calR(w_0) \subseteq \relint(\calW)$ as the \textit{reachable primal set}.

\paragraph{Effective dual space and stationarity map.}
Importantly, not every dual direction in $\calZ$ affects updates
to the primal point under $\nabla H$. 
To make this precise, define $\calS \subseteq \R^D$ 
and $\calS^\bot \subseteq \R^D$ by
\begin{equation}
    \textstyle
    \calS := 
    \prod_{i=1}^N \Span(\1_{d_i})
 	\;\;\text{and}\;\;
	\calS^\bot := 
	\{ v \in \R^D : \langle v, s \rangle = 0\,, \forall s \in \calS\}\,.
	\label{eq:calS-def}
\end{equation}
The linear subspace $\calS$ and its orthogonal complement
$\calS^\bot$ are the \textit{normal space} and \textit{tangent space}
of the affine hull of $\calW$, respectively (see Section~\ref{app:prelims:variational}).
Under Assumption~\ref{ass:legendre},
a key invariant is that 
$\nabla H(z_0+z) = \nabla H(z_0 + \Pi_{\calS^\bot}z)$
for $z \in \calZ$
(Proposition~\ref{prop:invariance-corr}). 
Thus, we let $\calU := \Pi_{\calS^\bot}\calZ$
denote the \textit{effective dual space},
and for $z_0 \in \R^D$ and $w_0 = \nabla H(z_0)$, we define the 
\textit{effective stationarity map} 
\begin{equation}
    \pi_0: \calU \to \calU \;,\;\;\;
	\pi_0(u) = \Pi_{\calS^\bot} G(\nabla H(z_0 + u))
	\;\;\text{for all $u \in \calU$} \;.
	\label{eq:pi0-def}
\end{equation} 
Importantly, $\pi_0(u) = 0$ 
if and only if $\nabla H(z_0 + u) \in \NE(G)$
(Proposition~\ref{prop:pi0-zeros-ne}), and thus $\|\pi_0(u)\|$ 
acts as a first-order measure of distance to equilibrium 
in effective dual coordinates. 
We will also assume the following
regularity condition on the zeros of $\pi_0$:

\begin{restatable}[Regularity]
{ass}{assrestrictedregularity}
	\label{ass:restricted-regularity}
	For initial $z_0 \in \R^D$
	and $w_0 = \nabla H(z_0)$,
    assume
	$\calR(w_0) \not \subseteq \NE(G)$,
    and that 
	every zero of $\pi_0$ is \textit{regular}, i.e.,   
	if $\pi_0(u) = 0$ for $u \in \calU$, then 
    the Jacobian $J_{\pi_0}(u)$ is nonsingular.     
\end{restatable}

\begin{restatable}
	{rem}{remrestrictedregularity}
	\label{rem:restricted-regularity}
    Assumption~\ref{ass:restricted-regularity}
    is a first-order
    nondegeneracy 
    condition
    whose main consequence is a local linear
    error bound
    around 
    each reachable 
    Nash equilibrium.  
    The residual
    $\|\pi_0(u)\|$ controls the distance 
    to that equilibrium
    in the effective dual space.
    For linear $G$
	(polymatrix games)
	the assumption always holds
    (Proposition~\ref{prop:linear-regular}),
    and it also holds 
    for nonlinear examples
    (see Section~\ref{app:harmonic-ex2}).
\end{restatable}

\section{Main Result: 
Linear Convergence Rates in Harmonic Games}
\label{sec:main-results}

Our main theorem proves that, for harmonic
games, OFTRL and EGFTRL converge
at a linear last-iterate rate to an NE
$w^\star \in \WMEQ_M(G)$ in Bregman divergence $\D_Q(w^\star, \cdot)$.
To state the result, we define 
the \textit{$M$-weighted regularizer}
given by
$Q_M(w)
= 
\sum_{i=1}^N m_i \cdot (\frac{\eta}{\eta_i}Q_i(w_i))$.

\begin{tcolorbox}[
	colback=gray!06,
	boxsep=-0.5mm,
	arc=0.01em,
	colframe=white,
	boxrule=-0.1em,
	left=0.5em,
	right=0.5em,
]
\begin{restatable}
	[Linear last-iterate convergence rate]
	{thm}{thmmainlast}
	\label{thm:main-last-iterate}
	\, Let $G: \calW \to \R^D$ be a loss operator
	such that $\WMEQ_M(G) \cap\, \relint(\calW)\neq \emptyset$
	for a weight matrix $M$,
	and let $\kappa_M = \frac{m_{\max}}{m_{\min}}$.
	Let Assumptions~\ref{ass:operator},~\ref{ass:legendre},
	\ref{ass:reg-sc}, \ref{ass:lhs}, and~\ref{ass:restricted-regularity} hold.
	Then for 
	\eqref{eq:oftrl}~or \eqref{eq:egftrl}\,
	with initialization $w_0 \in \relint(\calW)$
	and with stepsize $\eta \in (0, 1/(880LB\omega \kappa_M)]$,
	there exists a constant $\gamma(w_0) > 0$ such that
	the following hold:
	\begin{align*}
		(i)&\;\;
		\text{There exists a unique limit }
		w_{\infty} := \lim\nolimits_{t\to\infty} w_t 
		= 
		\argmin\nolimits_{w^\star \in \WMEQ_M(G)}\,
		\D_{Q_M}(w^\star, w_0) \;. \\[0em]
		(ii)&\;\;
		\text{For all $t \ge 0$:}\quad
		\D_{Q}(w_{\infty}, w_{t+1})
		\;\le\;
		2 \kappa_M
		 \cdot
		 \D_{Q}(w_{\infty}, w_0)
		\cdot 
		\exp(
			- \gamma(w_0) \eta^2 \cdot t
		) \;.
	\end{align*}
\end{restatable}
\end{tcolorbox}

\smallskip
Theorem~\ref{thm:main-last-iterate}
also gives a linear convergence rate 
in $\Gap(w_t)$ due to the following relationships:

\begin{restatable}
	[Gap vs. Bregman divergence]
	{rem}{remgapvsbregman}
	\label{rem:gap-vs-bregman}
	\, Let $w^\star \in \WMEQ_M(G)$.
	Then for all $w \in \relint(\calW)$:
	\begin{equation}
		\textstyle
		\frac{1}{2(L \rho)^2} \cdot 
		\Gap(w)^2
		\;\le \;
		\frac{1}{2} \|w^\star - w\|^2
		\;\le\;
		\D_Q(w^\star, w) \;,
		\label{eq:gap-relations}
	\end{equation}
	where $\rho := \max_{w', w'' \in \calW} \|w' - w''\|$
	is the diameter of $\calW$. 
	Thus, convergence in $\D_Q(w^\star, \cdot)$
	implies convergence in $\Gap(\cdot)$.
	The full derivation of these relationships
	is in Proposition~\ref{prop:gap-vs-bregman}.
\end{restatable}

Before giving an overview of the proof in Section~\ref{sec:main-proof-overview},
we make several additional remarks:

\begin{restatable}[Equilibrium selection]
	{rem}{remmainthmeqselection}
	\label{eq:main-thm-eq-selection}
    \,Part (i) of Theorem~\ref{thm:main-last-iterate}
    proves that the
	limit points of OFTRL and EGFTRL are 
	in $\WMEQ_M(G) \subseteq \NE(G)$,
	along with an exact initialization-dependent characterization
	in terms of the $M$-weighted regularizer $Q_M$.
	This strengthens the result of~\cite{legacci2024no},
	who showed only asymptotic convergence to the
	larger set $\NE(G)$.
    Note that the theorem does not assume
    uniqueness of NE, in contrast to
    prior last-iterate results for OMWU
    in two-player zero-sum games
    \citep{DaskalakisP19, wei2021linear, lazarsfeld2026}. 
\end{restatable}

\begin{restatable}[Instance-dependent rate]
	{rem}{remuniversallastiterate}
	\label{rem:main:universal-rate}
	The constant $\gamma(w_0) > 0$ in part (ii) of
	Theorem~\ref{thm:main-last-iterate} depends on $G$,
	on properties of the regularizer $Q$
	and energy function $H = Q^*$,
	and on the initial $w_0 \in \relint(\calW)$.
    As in prior results
    for two-player zero-sum games
    \citep{wei2021linear,lazarsfeld2026},
    the constant
    $\gamma(w_0)$ is thus instance-dependent,
    but uniform over all steps $t \ge 1$.
\end{restatable}

\paragraph{High-level proof strategy.}
The proof has
two main components:
(i) proving \textit{dissipative} bounds
on the change in energy along the dual iterates,
and (ii), proving a structural \textit{operator domination}
inequality that relates the magnitude of dissipation back 
to primal Bregman divergence from equilibrium.
This high-level strategy is similar to \cite{lazarsfeld2026} 
for the special case of OMWU in two-player zero-sum games,
with a unique, interior NE, but our result 
is substantially more general.
We do not leverage the 
special skew-symmetric structure of zero-sum game operators,
do not assume uniqueness of NE,
and our analysis simultaneously holds for both
OFTRL \textit{and} EGFTRL for a broad class of regularizers
beyond negative Shannon entropy. 

\paragraph{Implications of analysis for vanilla FTRL.}
Our techniques for proving Theorem~\ref{thm:main-last-iterate}
simultaneously establish a geometric rate of 
\textit{divergence} away from equilibrium
under \textit{vanilla} FTRL.
Specifically, we prove the following result 
for linear operators $G$
(which capture polymatrix games):
\begin{tcolorbox}[
	colback=gray!06,
	boxsep=-0.5mm,
	arc=0.01em,
	colframe=white,
	boxrule=-0.1em,
	left=0.3em,
	right=0.3em,
]
\begin{restatable}[Divergence of FTRL]
	{thm}{thmftrldiverge}
	\label{thm:ftrl-diverge-main}
    \, Let $G: \calW \to \R^D$ be a linear loss operator 
    such that
	$\WMEQ_M(G) \cap\, \relint(\calW) \neq \emptyset$
	for a weight matrix $M$,
	and let $\kappa_M = \frac{m_{\max}}{m_{\min}}$.
	Consider~\eqref{eq:ftrl}
	with $\eta \in (0, 1/(10 B\omega \kappa_M)]$
	and 
	$w_0 \notin \NE(G)$.
    Then for $w^\star \in \NE(G) \cap \relint(\calW)$ 
    and $\alpha > \D_{Q_M}(w^\star, w_0)$, 
    letting 
    $T_{\alpha}(w^\star) = 
    \sup\{t \ge 0: \D_{Q_M}(w^\star, w_t) \le \alpha\}$,
    there exists $\beta(w_0, w^\star, \alpha) > 0$
    such that
	\begin{equation*}
	\textstyle
	\D_Q(w^\star, w_{t}) 
	\ge 
	\kappa_M^{-1}\cdot
	\D_Q(w^\star, w_0) \cdot 
	(1 + \beta(w_0, w^\star, \alpha) \cdot \eta^2)^t
	\;\;\text{for all $0 \le t \le T_{\alpha}(w^\star)$}\;.
	\end{equation*}
\end{restatable}
\end{tcolorbox}

\smallskip
The proof is in Section~\ref{app:vanilla-ftrl}
and is similar to that of Theorem~\ref{thm:main-last-iterate}.
Moreover, we note the following:

\begin{restatable}[Quantitative separation]
	{rem}{remftrldivergence}
	\label{rem:ftrl-divergence}
    Theorem~\ref{thm:ftrl-diverge-main}
    shows a separation between~\eqref{eq:ftrl} and 
    the~\eqref{eq:oftrl} and~\eqref{eq:egftrl} variants.
	In the case of linear $G$ 
    (which captures polymatrix games), 
    the last-iterate of FTRL moves 
    away from all interior Nash equilibria
    and toward the simplex boundary,
    and within any fixed Bregman sublevel set 
    with respect to a fixed interior NE,  
    it does so at a geometric rate. 
    This is in contrast to last-iterate
    convergence rates to NE for OFTRL and EGFTRL. 
    The theorem generalizes and makes quantitative
	the divergence result of~\cite{bailey2018multiplicative}
    beyond the case of OMWU in two-player zero-sum games. 
\end{restatable}

\section{Overview of Proof of Theorem~\ref{thm:main-last-iterate}}
\label{sec:main-proof-overview}

\vspace*{-0.2em}
In this section, we give a technical overview
of the proof of Theorem~\ref{thm:main-last-iterate}
(the full proof is in 
Section~\ref{app:last-iterate-proof}).
The starting point of our analysis is to introduce
a change of basis that allows for
studying OFTRL and EGFTRL
over scaled spaces,
and using a \textit{symmetrized} operator $V$.

\paragraph{Reduction to unweighted system.}
Fix $G: \calW \to \R^D$
with $\WMEQ_M(G) \neq \emptyset$ for a
weight matrix $M$. 
Define the scaled primal domain 
$\calX := M^{1/2} \calW \subseteq \R^D$,
and for each $w \in \calW$ and $z \in \R^D$,
let $x = M^{1/2} w$ and $y = M^{1/2} z$.
Then let $V : \calX \to \R^D$ be the 
$M$-symmetrized operator, where 
\begin{equation}
	V(x) = M^{1/2} G(M^{-1/2} x)
	\;\;\text{for all $x \in \calX$}\;.
	\label{eq:V-def}
\end{equation}
The key property of $V$ is that
all solutions of the \textit{weighted}
Minty variational equality of~\eqref{eq:harmonic-wmeq} under $G$
also satisfy the following \textit{unweighted} Minty variational equality
under $V$:

\begin{tcolorbox}[
	colback=gray!06,
	boxsep=-0.5mm,
	arc=0.01em,
	colframe=white,
	boxrule=-0.1em,
	left=0.5em,
	right=0.5em,
]
\begin{restatable}[Unweighted Minty equality]
	{prop}{propmeqsymmetrized}
	\label{prop:symmetrized-meq} \,
	Let $V: \calX \to \R^D$ be the operator from~\eqref{eq:V-def}.
	For any solution $w^\star \in \WMEQ_M(G)$,
	let $x^\star := M^{1/2} w^\star \in \calX$. 
	Then it holds that:
	\vspace*{-0.25em}
	\begin{equation}
		\langle V(x), x- x^\star \rangle = 0
		\;\;\text{for all $x \in \calX$} \;.
		\label{eq:minty-equality}
		\tag{MEQ}
	\end{equation}	
\end{restatable}
\end{tcolorbox}
We write $\MEQ(V) \subseteq \calX$
to denote the set of solutions $x^\star$ satisfying~\eqref{eq:minty-equality}.
As a straightforward corollary:
$w^\star \in \WMEQ_M(G)$ if and only if
$x^\star = M^{1/2} w^\star \in \MEQ(V)$ 
(see Section~\ref{app:reduction}). 
From an algorithmic standpoint, the change of basis 
also yields the following equivalence:

\begin{tcolorbox}[
	colback=gray!06,
	boxsep=-0.5mm,
	arc=0.01em,
	colframe=white,
	boxrule=-0.1em,
	left=0.5em,
	right=0.5em,
]
\begin{restatable}[Symmetrized dynamics]
	{prop}{propsymmetrizeddynamics}
	\label{prop:reduction-algorithmic} \,
	Let $V: \calX \to \R^D$ be the operator 
	from~\eqref{eq:V-def}.
	Then there exist a global regularizer $R: \calX \to \R$
	and energy function $F = R^*$ such that,
	if $\{(w_t, z_t)\}$ is an~\eqref{eq:oftrl} ~or~ \eqref{eq:egftrl}~
	trajectory for $(G, H, \eta)$, 
	then letting $x_t = M^{1/2} w_t$
	and $y_t = M^{1/2} z_t$ for all $t \ge 0$,
	it follows that $\{(x_t, y_t)\}$
	is an~\eqref{eq:oftrl}~or~\eqref{eq:egftrl}~ trajectory 
	for $(V, F, \eta)$. 
\end{restatable}
\end{tcolorbox}

See Section~\ref{app:reduction} for the proof and 
the exact definitions of $R$ and $F$.
There, we also show
for all $w^\star \in \WMEQ_M(G)$,
$x^\star = M^{1/2} w^\star \in \MEQ(V)$,
and all $w \in \relint(\calW)$, $x = M^{1/2} w \in \relint(\calX)$ that:
\begin{equation}
	m_{\min} \cdot \D_Q(w^\star, w)
	\;\le\;
	\D_R(x^\star, x)
	\;\le\;
	m_{\max} \cdot 
	\D_Q(w^\star, w)  \;.
	\label{eq:DQ-vs-DR}
\end{equation}

\paragraph{Intuition for reduction.}
While $G$ satisfies the \textit{$M$-weighted} variational equality
of~\eqref{eq:wmeq}, the iterates $\{z_t\}$ are updated via 
$G$ in an \textit{unweighted} manner. 
The reduction specifies the coordinate system
under which the resulting operator  satisfies the
unweighted equality~\eqref{eq:minty-equality}
while also preserving the OFTRL/EGFTRL update structure. 
Analyzing $\D_Q(w^\star, w_t)$ thus
reduces to controlling $\D_R(x^\star, x_t)$.

\subsection{Energy Dissipation and Continuous-Time Intuition}
\label{sec:proof-overview:energy}

To prove $\D_R(x^\star, x_t) \to 0$
under the scaled primal iterates $\{x_t\}$, 
our analysis is based on establishing the 
dissipation of a \textit{shifted} energy function, related to $\D_R$, 
along the 
scaled dual iterates $\{y_t\}$.  

\paragraph{Shifted energy function.}
For $x^{\star} \in \calX$, let
$F_{x^\star} : \R^D \to \R$ be  
the \textit{shifted energy function}: 
\begin{equation}
	F_{x^\star}(y) = F(y) - \langle x^\star , y \rangle 
	\;\;\text{for all $y \in \R^D$}\;.
	\label{eq:shifted-energy}
\end{equation}	
We then have the following
key relationship between 
$\D_R$ and $F_{x^{\star}}$
(see the proof in
Section~\ref{app:reduction}):
\begin{tcolorbox}[
	colback=gray!06,
	boxsep=-0.5mm,
	arc=0.01em,
	colframe=white,
	boxrule=-0.1em,
	top=0.3em,
	left=0.5em,
	right=0.5em, 
]
\begin{restatable}[Change in shifted energy and Bregman divergence]
{prop}{propshiftedenergy}
	\label{prop:bregman-shifted-energy}
	Let $y, y' \in \R^D$, and  let
	$x = \nabla F(y), x' = \nabla F(y') \in \relint(\calX)$. 
	Then for any $x^{\star} \in \relint(\calX)$, 
    it holds that:
	\begin{equation*}
		\textstyle
	    \D_R(x^{\star}, x') - \D_R(x^{\star}, x) 
		\;=\;
		F_{x^{\star}}(y') - F_{x^{\star}}(y) \;.
	\end{equation*}
\end{restatable}
\end{tcolorbox}

\paragraph{Dissipation of shifted energy.}
Given Proposition~\ref{prop:bregman-shifted-energy},
it then suffices to establish 
\textit{dissipative} bounds 
under OFTRL and EGFTRL
on the one-step change in $F_{x^\star}$
for $x^\star \in \MEQ(V)$.
This first requires the following notation.
For all $y \in \R^D$, 
the Hessian $\nabla^2 F(y)$
is positive semi-definite and induces 
a local degenerate inner product 
and local semi-norm:
for $u, v \in \R^D$ 
define $\langle u, v \rangle_{y}
:= \langle u, \nabla^2 F(y) v \rangle$
and
$\|v\|_{y} := \sqrt{\langle v, \nabla^2 F(y) v \rangle}$
(see Section~\ref{app:error-machinery:seminorms} 
for further background).
Then, our main technical lemma shows the dissipation 
of $F_{x^\star}$ along the dual iterates,
quantified under the semi-norm:
\begin{tcolorbox}[
	colback=gray!06,
	boxsep=-0.5mm,
	arc=0.01em,
	colframe=white,
	boxrule=-0.1em,
	top=0.3em,
	left=0.5em,
	right=0.5em,
]
\begin{restatable}[Dissipation of shifted energy]
	{lem}{lemenergydissipationmain}
	\label{lem:energy-dissipation-main}
	Let Assumptions~\ref{ass:operator}, 
	\ref{ass:legendre},
	\ref{ass:reg-sc}, \ref{ass:lhs} hold.
	Let $\{(x_t, y_t)\}$ be the \textnormal{OFTRL}
	or \textnormal{EGFTRL} iterates from 
	Proposition~\ref{prop:reduction-algorithmic}
	with $\eta \in (0, 1/(880 L B \omega \kappa_M)]$.
	Then for $t \ge 1$:
	\begin{equation*}
	\textstyle
		- \frac{7}{8} \cdot \eta^2 \|V(x_t)\|^2_{y_t}
		\;\le\;
		F_{x^\star}(y_{t+1}) - 
		F_{x^\star}(y_t)
		\;\le\;
		-\frac{1}{8} \cdot \eta^2 \|V(x_t)\|^2_{y_t}\,, 
       \quad\text{$\forall \; x^\star \in \MEQ(V)$.}
	\end{equation*}
\end{restatable}
\end{tcolorbox}

\smallskip
Note that these energy dissipation bounds
hold simultaneously for all
points $x^\star \in \MEQ(V)$,
and the magnitude of dissipation is
quantified under the geometry of 
$\nabla^2 F(y_t)$ at the current dual iterate. 
Moreover, these bounds are tight 
(with the lower and upper inequality matching
up to an absolute constant), and
thus when combined with Proposition~\ref{prop:bregman-shifted-energy},
the lemma implies 
that the magnitude of dissipation
exactly controls the change in $\D_R(x^\star, x_t)$
under the iterates
for all $x^\star \in \MEQ(V)$.

\paragraph{Intuition from continuous time.}
The starting point for understanding the dissipation in
Lemma~\ref{lem:energy-dissipation-main} is the continuous-time analogue 
of FTRL, which is the first-order ODE that we refer to as \textit{FTRL flow}:
\begin{equation}
	\dot y_t = - V(x_t)
	\quad\text{and}\quad
	x_{t} = \nabla F(y_t) \;.
	\tag{FTRL flow}
	\label{eq:ftrl-flow}
\end{equation}
Importantly, along~\eqref{eq:ftrl-flow}, 
it holds for every $x^\star \in \MEQ(V)$
that the Bregman divergence from $x^{\star}$
(and equivalently
the shifted energy by Proposition~\ref{prop:bregman-shifted-energy})
is \textit{exactly conserved}. 
Indeed, differentiating
Proposition~\ref{prop:bregman-shifted-energy}
along the~\eqref{eq:ftrl-flow}
trajectory and using the chain rule,
we have for all initializations and
all $x^\star \in \MEQ(V)$ that:
\begin{align}
	\textstyle
	\tfrac{d}{dt} \D_R(x^\star, x_t)
	= 
	\tfrac{d}{dt} F_{x^\star}(y_t)
	&=
	\langle \dot y_t, \nabla F_{x^*}(y_t) \rangle  \nonumber \\
	&= 
	- \langle V(x_t), \nabla F(y_t) - x^\star \rangle 
	= - \langle V(x_t), x_t - x^\star \rangle 
	= 0 \;.
	\label{eq:flow-conservation}
\end{align}
Here, the final equality holds 
exactly due to the structural~\eqref{eq:minty-equality}
property. This conservation
was shown in~\cite{legacci2024no} and is 
similar to other continuous-time conservation results 
for two-player zero-sum games.
A consequence of~\eqref{eq:flow-conservation}
is that the behavior of
discrete-time FTRL variants can be
better understood when viewed as discretizations
of the~\eqref{eq:ftrl-flow}:

\paragraph{FTRL and Implicit FTRL.}
Taking a forward (resp. backward) \textit{Euler discretization}
of~\eqref{eq:ftrl-flow} with stepsize $\eta > 0$ then
yields the vanilla~\eqref{eq:ftrl}
(resp. implicit~\eqref{eq:iftrl}) update for the symmetrized system.
A similar calculation to \eqref{eq:flow-conservation} in discrete time
 yields for all $x^\star \in \MEQ(V)$:
\begin{equation}
	F_{x^\star}(y_{t+1}) - F_{x^\star}(y_t)
	= 
	\begin{cases}
	\;\;\;\D_F(y_{t+1}, y_t) \;\ge\; 0 
	\quad\text{for FTRL} \\
	-\D_F(y_t, y_{t+1}) \;\le\; 0 
	\quad\text{for IFTRL} \;.
	\end{cases} 
	\label{eq:forward-backward-energy}
\end{equation}	
In other words, the continuous-time energy conservation 
under~\eqref{eq:ftrl-flow} becomes 
energy \textit{expansion} under the forward discretization~\eqref{eq:ftrl},
and energy \textit{dissipation} under the backward discretization~\eqref{eq:iftrl}
(see Proposition~\ref{prop:ftrl-iftrl-energy} 
for the proof).
The resulting $\D_F$ (resp. $-\D_F$) term
captures the \textit{second-order} forward (resp. backward) discretization error.  
Under this perspective, \eqref{eq:oftrl} and~\eqref{eq:egftrl}
can both be viewed as implementable approximations
of the idealized~\eqref{eq:iftrl}, which leads 
to the dissipation in Lemma~\ref{lem:energy-dissipation-main}.
More concretely, the proof of the lemma follows by 
coupling OFTRL and EGFTRL at each $t \ge 1$ 
with an \textit{implicit} step,
which we sketch below:

\paragraph{Proof sketch of 
Lemma~\ref{lem:energy-dissipation-main}.}
To set up the coupling, we establish the following proposition:
\begin{tcolorbox}[
	colback=gray!06,
	boxsep=-0.5mm,
	arc=0.01em,
	colframe=white,
	boxrule=-0.1em,
	top=0.3em,
	left=0.5em,
	right=0.5em,
]
\begin{restatable}[Implicit step coupling]
	{prop}{propenergycouplingmeta}
    \label{prop:meta-coupling}
	Let $\{(x_t, y_t)\}$ be a sequence
	satisfying $x_t = \nabla F(y_t)$
	for all $t \ge 0$.
	Further, define the implicit update $\wh y_{t+1} := y_t - \eta V(x_{t+1})$
	for $t \ge 0$. Then for $t \ge 1$:
	\vspace*{-0.1em}
    \begin{equation*}
		F_{x^\star}(y_{t+1}) - F_{x^\star}(y_t)
		\;=\;
		\underbrace{- \D_F(y_t, y_{t+1})}_{\textnormal{Dissipation term}}
		+ 
	\underbrace{\langle x_{t+1} - x^\star, y_{t+1} - \wh y_{t+1} \rangle}_{\textnormal{Coupling error term $C_t$}}\,,
        \;\;\;
        \forall \; x^\star \in \MEQ(V)\;.
	\end{equation*}
\end{restatable}
\end{tcolorbox}

Fixing any $x^\star \in \MEQ(V)$, 
Proposition~\ref{prop:meta-coupling}
then shows that the one-step change
in shifted-energy exactly decomposes
into a \textit{dissipation} term
(that matches the one-step change 
of~\eqref{eq:iftrl} from~\eqref{eq:forward-backward-energy}),
and a \textit{coupling error term} $C_t$.
Under the dual iterates of OFTRL and EGFTRL, 
we prove that both
$- \D_F(y_t,y_{t+1}) \lesssim - \eta^2 \|V(x_t)\|^2_{y_t}$
and that 
$C_t \lesssim \eta^2 \|V(x_t)\|^2_{y_t}$,
and for a sufficiently small constant stepsize $\eta$, 
the constant scaling of the negative dissipation 
term dominates that of $C_t$.
In particular, to control the error terms,
we rely on consequences of the local 
Hessian stability property~\eqref{eq:lhs}
(Assumption~\ref{ass:lhs}).
This establishes the upper bound of the lemma,
and the matching lower bound follows similarly. 
The full proof is given in Section~\ref{app:energy:main-lemma-proof}.
~\hfill $\blacksquare$

\subsection{Operator Domination Inequality}
\label{sec:main:op-dom}

Next, we prove a structural inequality
relating the local seminorms
from Lemma~\ref{lem:energy-dissipation-main}
back to $\D_R$.
For this, we extend the notation
from Section~\ref{sec:ftrl:geometry}
to the present, symmetrized coordinates:
for $y_0 \in \R^D$ and $x_0 = \nabla F(y_0) \in \relint(\calX)$, 
let $\calR(x_0) := \{\nabla F(y_0 + y) : y \in \Span(V(\calX))\}$,
and let $\NE(V) \subseteq \calX$
denote the solutions of the Stampacchia variational inequality
\eqref{eq:nash-svi} for the operator $V$.

\begin{tcolorbox}[
	colback=gray!06,
	boxsep=-0.5mm,
	arc=0.01em,
	colframe=white,
	boxrule=-0.1em,
	left=0.5em,
	right=0.5em,
]
\begin{restatable}[Operator domination on reachable set]
	{prop}{propopdom}
	\label{prop:op-dom}
	Let Assumption~\ref{ass:restricted-regularity} hold.
	Fix $y_0 \in \R^D$ and $x_0 = \nabla F(y_0)$
    such that $\calR(x_0) \not \subseteq \NE(V)$. 
	Let $x^\star \in \MEQ(V) \cap \calR(x_0)$,
	and let $\calK \subset \relint(\calX)$ be any compact set
	such that $x^\star \in \calK$. 
	Then there exists a positive constant 
    $\wt \gamma(x_0, \calK) > 0$ such that
	\begin{equation*}
    	\|V(x)\|^2_{y}
        \;\ge\;
        \wt \gamma(x_0, \calK) \cdot \D_R(x^\star, x) \,,
        \quad
         \forall \; y \in \R^D\,
        \text{s.t. $x=\nabla F(y) \in \calK \cap \calR(x_0)$}\;.
	\end{equation*}
\end{restatable}
\end{tcolorbox}

Proposition~\ref{prop:op-dom}
establishes a structural
inequality
between the dissipation term
$\|V(x)\|^2_{y}$
and the Bregman divergence 
$\D_R(x^\star, x)$ that is
analogous in spirit to 
a gradient-domination inequality
in optimization. 
In general, such a relationship
need not hold uniformly 
over the entire primal space.
However, under the regularity 
condition of 
Assumption~\ref{ass:restricted-regularity},
we show that a local
error bound around
a reachable MEQ point
$x^\star \in \MEQ(V)  \cap \calR(x_0)$
can be made uniform over
$\calK \cap \calR(x_0)$ 
for any compact subset 
 $\calK \subset \relint(\calX)$ 
 containing $x^\star$.
This yields a positive domination
constant $\wt \gamma(x_0, \calK)$
for which the inequality 
holds uniformly over $\calK \cap \calR(x_0)$,
but with dependencies on the initialization $x_0$
and the set $\calK$ (and, intrinsically,
on geometric properties of the
regularizer over this set). 
We give the full details
of the proof in Section~\ref{app:op-dom-details}.

\paragraph{Combining the pieces.}
To combine the pieces,
we show that the dissipation bounds of 
Lemma~\ref{lem:energy-dissipation-main}
imply that all  OFTRL and EGFTRL iterates
lie in a compact set.
In Lemma~\ref{lem:eq-selection}
(see Section~\ref{app:eq-selection}),
we use this compactness to then prove
the asymptotic equilibrium selection result of 
part (i) of the theorem.
Compactness then further
allows us to apply the operator domination inequality 
of Proposition~\ref{prop:op-dom}
along the OFTRL and EGFTRL iterates.  
Together with the upper bound
of Lemma~\ref{lem:energy-dissipation-main},
this recursively establishes
the linear convergence rate 
of part (ii).

\medskip
\section{Conclusion}
\label{sec:conclusion}

This work establishes linear last-iterate convergence rates for OFTRL and EGFTRL 
in harmonic games, and more broadly for operators satisfying the global (weighted)
Minty equality.  
Several directions remain open,
including adaptive stepsize variants, 
and a finer understanding of the instance dependence of the linear rate.
More broadly, it would be interesting to 
investigate whether
our techniques can be used 
in general games through their harmonic and potential decompositions.

\medskip
\subsection*{AI use statement}

In this work, we used generative AI tools to proofread and suggest editing changes for the technical writing and proofs. We also used generative AI tools for verifying properties of example regularizers and for generating and verifying examples of harmonic games. Finally, we used generative AI tools to suggest edits of our manuscript to improve readability, and to run simulations of the algorithms we study. We take responsibility for the final content of this work, including all proofs.

\medskip
\subsection*{Acknowledgements.}
VM gratefully acknowledges the support of the NSF through 
grants CCF-2239151 and IIS-2212182, and an Amazon Research Award. 
JZ gratefully acknowledges the support of the NSF through
grants CAREER-2336236 and III-2504990.
AW is supported by NSF awards CCF-2403391 and CAREER CCF-2443097.

\newpage
\setcounter{tocdepth}{1}
\renewcommand*\contentsname{Table of Contents}
{\setlength{\parskip}{-0.5em}\tableofcontents}

\appendix

\newpage

\section{Additional Preliminaries}
\label{app:prelims}

In this section, we provide additional preliminaries on variational analysis, harmonic games and properties of the Bregman divergence.

\subsection{Variational Preliminaries}
\label{app:prelims:variational}

\smallskip
Here, we provide
additional preliminaries
on the variational
perspective discussed in Section~\ref{sec:prelims}.

\subsubsection{Overview of spaces} 

Recall that the joint strategy space is given by $\calW = \calW_1 \times \dots \times \calW_N \subseteq \R^D$,
where each $\calW_i = \Delta_{d_i}$. 
Moreover, recall from~\eqref{eq:calS-def} that 
we define the linear subspaces $\calS$ and $\calS^\bot$
by 
\begin{equation*}
    \textstyle
    \calS := 
    \prod_{i=1}^N \Span(\1_{d_i})
 	\;\;\text{and}\;\;
	\calS^\bot := 
	\{ v \in \R^D : \langle v, s \rangle = 0\,, \forall s \in \calS\}\,.
\end{equation*}
In particular, we note that 
$\calS$ and $\calS^\bot$ 
are the \textit{normal}
and \textit{tangent} spaces 
of the affine hull $\aff(\calW)$.
In particular, recall that 
$\aff(\calW) = w' + \calS^\bot$
for any $w' \in \calW$. Further background can be found in, e.g.,~\citet[Section 6]{rockafellar98var-book}. Then, we also have the following result for which we provide a proof for completeness.  

\begin{restatable}
{lem}{lemnormalspace}
    \label{lem:normal-diff}
    For any $w' \in \calW,$ 
    $\calS^\bot = 
    \Span\{w - w' : w \in \calW\} \;.$
\end{restatable}

\begin{proof}
Let $w' \in \calW.$ 
We first prove that $w - w' \in \calS^\bot$
for any $w \in \calW$, 
which implies that the span is also a subset 
of the subspace $\calS^\bot$.  
For every $i\in[N]$ and every
$s_i\in\Span(\1_{d_i})$, there exists $c_i\in\R$ such that
$s_i=c_i\1_{d_i}$. Since $w_i,w_i'\in\Delta_{d_i}$, we have
$\langle s_i,w_i-w_i'\rangle
=c_i\langle\1_{d_i},w_i-w_i'\rangle =c_i(1-1)=0.$
Hence, for every
$s=(s_i)_{i\in[N]}\in\calS$,
$\langle s, w - w'\rangle = 
\sum_{i=1}^N \langle s_i, w_i - w_i'\rangle = 0$, 
and thus $w - w' \in \calS^\bot$ by definition of $\calS^\bot$.

We now prove the remaining inclusion. 
Let $v \in \calS^\bot$.
Since $\aff(\calW) = w' + \calS^\bot$, 
we have $w' + v \in \aff(\calW)$. Then by definition 
of $\aff(\calW)$, there exist $(\lambda_j)_{j \in [k]} \in \mathbb{R}^k$ 
for some $k \in \mathbb{N}$ and $(w^j)_{j \in [k]} \in \calW^{k}$ 
such that $\sum_{j=1}^k \lambda_j = 1$ and 
$w' + v = \sum_{j=1}^k \lambda_j w^j$. 
Hence $v = \sum_{j=1}^k \lambda_j (w^j - w') 
\in \Span\{w - w' : w \in \calW\}$.
\end{proof}

\smallskip
\subsubsection{Properties of Global Loss Operator}

Recall that $G: \calW \to \R^D$
is the operator from~\eqref{eq:G-def}.
By definition, $G$ is continuous. 
As is standard for normal-form games,
we further assume that $G$ 
satisfies the following global
Lipschitz and boundedness properties:

\smallskip
\assoperator*

\noindent
We note that boundedness follows from continuity of
$G$ and compactness of the joint space $\calW$.

Moreover, recall the variational
characterization of Nash equilibria
from~\eqref{eq:nash-svi}:
$w^\star \in \calW$ is a
Nash equilibrium if and only if
\begin{equation}
    \langle G(w^\star), w -w^\star \rangle
    \ge 0 
    \;\;\;\text{for all $w \in \calW$} \;.    
    \label{eq:NE}
    \tag{NE}
\end{equation}
The proof of this equivalence is standard
and follows from the first-order optimality 
conditions for each player’s loss minimization problem,
with the strategies of all other players fixed, 
as required by the definition of a Nash equilibrium 
(see, e.g.,~\cite{facchinei2003finite}, Proposition 1.4.2). 
Recall that we then write
$\NE(G)\subset \calW$ 
to denote the set of all solutions\footnote{
    A reader familiar with variational
    inequalities will recognize~\eqref{eq:nash-svi}
    (and equivalently~\eqref{eq:NE}) as
    simply the set of Stampacchia VI (SVI) solutions
    for $G$. For the setting of 
    Section~\ref{sec:prelims},
    these two sets coincide. 
} $w^\star$ of~\eqref{eq:NE}.

\smallskip
\paragraph{Interior Nash Equilibria 
and Normal Operator Vectors.} ~

In the following proposition,
we establish an equivalence between 
interior Nash equilibria $w^\star \in \relint(\calW)$
and points $w \in \calW$ 
with normal operator values $G(w) \in \calS$:

\smallskip
\begin{restatable}[NE vs. Normal Operator Value]{prop}{propsvinormalrelate}
	\label{prop:ne-vs-normal}
	Let $G: \calW \to \R^D$ and 
    let $w' \in \calW$. Then, the following statements hold: 
	\begin{enumerate}[
		label={(\roman*)},
        leftmargin=2em,        
	]
	\item If 
	$G(w') \in \calS$, then $w' \in \NE(G)$.

	\item
    If $w' \in \relint(\calW) \cap \NE(G)$,
    then $G(w') \in \calS$ and 
    $\langle G(w'), w - w' \rangle = 0$ for any $w \in \calW$.
\end{enumerate}
\end{restatable}

\smallskip

\begin{proof}
	We prove the two claims separately:

    \noindent
    \textit{Proof of (i).}
    If $G(w') \in \calS$, then for all $v \in \calS^\bot$, $\langle G(w'), v \rangle = 0$ by definition of $\calS$ and $\calS^\bot$. 
	As $w-w' \in \calS^\bot$ for all 
    $w \in \calW$ (by Lemma~\ref{lem:normal-diff}), 
    this further implies that
	$\langle G(w'), w-w'\rangle = 0$ for all $w \in \calW$,
	which means that $w' \in \NE(G)$.

    \smallskip
    \noindent
	\textit{Proof of (ii).}
    Suppose now that $w' \in \relint(\calW) \cap \NE(G)$. 
	Since $w' \in \relint(\calW)$, for every $v \in \calS^\bot$, there 
	exists a sufficiently small 
    $\eps > 0$ such that $w' + \eps v \in \calW$ and $w' - \eps v \in \calW$. 
	Then, since $w' \in \NE(G)$, 
	it follows from \eqref{eq:NE} that for all $v \in \calS^\bot$, 
	\begin{align*}
		\eps \langle G(w'), v \rangle = 
        \langle G(w'), (w'+ \eps v) - w' \rangle &\ge 0\,, \\
		\quad
		- \eps \langle G(w'), v \rangle = 
        \langle G(w'), (w' - \eps v) - w' \rangle &\ge 0\,. 
	\end{align*}
	As $\eps > 0$, we obtain  
	$\langle G(w'),  v \rangle = 0$ for all 
    $v \in \calS^\bot$, 
    and thus 
	$G(w') \in \calS$.
\end{proof}

\medskip
\subsubsection{Polymatrix Games and Linear Operators}
\label{app:prelims:polymatrix}

For the normal-form game setting introduced 
in Section~\ref{sec:prelims}, 
we have a linear characterization of $G$ 
when the underlying game is a
\textit{polymatrix} game
(e.g., \cite{daskalakis2009network}).

\smallskip
\begin{restatable}
[Linear operators and polymatrix games]
{prop}{proplinearpolymatrix}
\label{prop:linear-polymatrix}
Consider a polymatrix game and let $G: \calW \to\R^D$ 
be the corresponding operator defined in~\eqref{eq:G-def}. 
Then, there exists a matrix $P \in \R^{D \times D}$ such that $G(w) = Pw$ for every $w \in \calW.$
\end{restatable}

\begin{proof}
By definition of a polymatrix game,
for every $(i,j)\in[N]\times[N]$ there exists a matrix
$A^{i,j}\in\R^{d_i\times d_j}$, with $A^{i,i}=0$, 
such that the loss $\ell_i: \calW \to \R$ of player $i$ is given by
\begin{equation}
\textstyle
\ell_i(w)
=
\sum_{j=1, j \neq i}^N
\langle w_i, A^{i,j}w_j\rangle
\;\;
\text{for all $w\in\calW$}\,.
\end{equation}
Hence, by the definition of $G$ in~\eqref{eq:G-def}, 
$G_i(w) =  \nabla_{w_i} \ell_i(w) = \sum_{j\neq i}^N A^{i, j} w_j$
for all $i \in [N]$.
Now let $P\in\R^{D\times D}$ be the block matrix 
whose $(i,j)$-block is $P_{i,j} = A^{i,j}$. 
Since $A^{i,i} = 0$, the $i$'th
block of $Pw$ is $(Pw)_i = \sum_{j\neq i}^N A^{i, j} w_j = G_i(w)$.
Therefore, $G(w) = Pw$ for every $w \in \calW.$
\end{proof}

\smallskip
\subsection{Preliminaries on Harmonic Games}
\label{app:prelims:harmonic}

\smallskip
In this section, we give several further
preliminaries on harmonic games. 

\subsubsection{Details on the Variational 
Characterization of Harmonic Games}
\label{app:prelims:harmonic-variational}

As introduced in Section~\ref{sec:prelims:harmonic},
harmonic games from Definition~\ref{def:harmonic-finite}
admit the following
variational characterization
in terms of the operator $G$,
a result due to~\cite{legacci2024no} that we restate below.

\smallskip
\propharmonicmixed*

Recall from~\eqref{eq:wmeq} 
that for a weights matrix $M$
satisfying Proposition~\ref{prop:harmonic-mixed},
\begin{equation}
	\WMEQ_{M}(G) 
	= 
	\{w' \in \calW \,:\,
	\langle MG(w), w- w' \rangle = 0
	\;\text{for all $w \in \calW$}
	\}\;.
	\label{eq:wmeq-app}
\end{equation}
Moreover, we further have 
by the proposition that 
$\WMEQ_M(G) \cap \relint(\calW) \neq \emptyset$.
Note that $\WMEQ_M(G)$
can be viewed as the solution set to
an $M$-weighted \emph{Minty variational equality}.
Importantly, this implies that 
(1) the set $\WMEQ_M(G)$ is convex,
and (2) $\WMEQ_M(G) \subseteq \NE(G)$. 
We prove these properties in what follows:

\smallskip
\begin{restatable}
    [Convexity of WMEQ set]
    {lem}{lemconvexitywmeq}
    \label{lem:wmeq-convex}
    Let $G: \calW \to \R^D$
    satisfy 
    $\WMEQ_M(G) \neq \emptyset$
    for a weights matrix $M$. 
    Then, the set $\WMEQ_M(G)$ defined in \eqref{eq:wmeq-app} is convex.
\end{restatable}

\begin{proof}
    Let $w', w'' \in \WMEQ_M(G)$.
    For $\lambda \in [0, 1]$,
    let $w_{\lambda} := \lambda w' + (1-\lambda) w''$.
    To prove convexity of $\WMEQ_M(G)$, 
    we show that $w_{\lambda} \in \WMEQ_M(G)$ for 
    any $\lambda \in [0, 1]$.  
    
    By convexity of $\calW$,
    we have 
    $w_{\lambda} \in \calW$ for all $\lambda \in [0, 1]$.
    Then for any $w \in \calW$, observe that
    \begin{align*}
        \langle M G(w), w-w_{\lambda} \rangle
        = 
        \lambda \langle M G(w), w-w' \rangle
        + 
        (1-\lambda) \langle MG(w), w-w''\rangle
        = 0\,, 
    \end{align*}
    where the last equality follows from the fact 
    that $w', w'' \in \WMEQ_M(G)$. 
    Thus $w_{\lambda} \in\WMEQ_M(G)$.
\end{proof}

\medskip
\begin{restatable}
    [WMEQ points are NE]
    {lem}{lemwmeqne}
    \label{lem:wmeq-ne}
    Let $G: \calW \to \R^D$
    satisfy 
    $\WMEQ_M(G) \neq \emptyset$
    for a weights matrix $M$. 
    Then $\WMEQ_M(G) \subseteq \NE(G)$,
    and $G(w^\star) \in \calS$ 
    for
    every $w^\star \in \WMEQ_M(G)$.
\end{restatable}

\begin{proof}
    Let $w^\star \in \WMEQ_M(G)$.
    For each player $i \in [N]$
    and any $w_i \in \calW_i$,
    we then apply~\eqref{eq:wmeq-app}
    at the point $w = (w_i,  w^\star_{-i})$.
    By the block diagonal structure of 
    $M = \Diag(m_1 I_{d_1}, \dots, m_N I_{d_N})$,
    for $m_i > 0$, 
    and since $G_i(w) = G_i(w^\star)$
    by definition of $G$ in~\eqref{eq:G-def},
    this yields at the $i$-th component:
    \begin{equation}
        \textstyle
        0 = 
        m_i 
        \langle G_i(w), w_i- w^\star_i \rangle
        = 
        m_i 
        \langle G_i(w^\star),w_i - w^\star_i
        \rangle
        = 0\;.
    \end{equation}
    Dividing by $m_i> 0$ 
    then gives 
    $\langle G_i(w^\star), w_i -w^\star_i \rangle= 0$
    for all $w_i \in \calW_i$.
    As this latter equality holds for all players
    simultaneously, summing over
    $i \in [N]$ then yields
    \begin{equation}
        \langle G(w^\star), w-w^\star \rangle= 0
        \;\;\;\text{for all $w \in \calW$}\;.
        \label{eq:wmeqne01}
    \end{equation}
    Thus $w^\star$ satisfies~\eqref{eq:NE}
    with equality and so 
    $w^\star \in \NE(G)$.

    Moreover, as 
    $\calS^\bot = 
    \Span\{w-w^\star: w \in \calW\}$
    by Lemma~\ref{lem:normal-diff}, 
    expression~\eqref{eq:wmeqne01}
    also immediately implies that
    $G(w^\star) \in \calS$,
    which concludes the proof.
\end{proof}

\subsubsection{Harmonic games
with linear operators}

When the operator $G$ 
satisfying Proposition~\ref{prop:harmonic-mixed}
is~\textit{linear},
which is the case for 
polymatrix games (Proposition~\ref{prop:linear-polymatrix}), the sets $\WMEQ_M(G)$
and $\NE(G)$ coincide. 

\smallskip
\begin{restatable}
    [NE and WMEQ for linear G]
    {prop}{propwmeqlinear}
    \label{prop:ne-wmeq-linear}
    Let $G: \calW \to \R^D$
    satisfy 
    $\WMEQ_M(G) \cap \relint(\calW) 
    \neq \emptyset$
    for a weights matrix $M$.
    Suppose
    $G(w) = Pw$ for some 
    $P \in \R^{D\times D}$.
    Then $\WMEQ_M(G) = \NE(G)$.
\end{restatable}

\begin{proof}
    As $\WMEQ_M(G) \subseteq \NE(G)$
    by Lemma~\ref{lem:wmeq-ne},
    it suffices to show that 
    $\NE(G) \subseteq \WMEQ_M(G)$
    under the assumptions of the proposition.
    For this, let $w^\star \in \NE(G)$,
    and let $\wt w \in \WMEQ_M(G) \cap 
    \relint(\calW)$, which exists by assumption.  
    To show that $w^{\star} \in \WMEQ_M(G)$, 
    we prove in what follows that for any $w \in \calW$, 
    $\langle MG(w), w - w^{\star} \rangle = 0$, 
    as per the definition of the set $\WMEQ_M(G)$ in \eqref{eq:wmeq-app}.

    First, fix any $w \in \calW$. Then observe that we can write
    \begin{equation}
        \langle MG(w), w - w^{\star} \rangle =
        \langle MG(w), w - \wt w \rangle 
        + \langle MG(w), \wt w - w^{\star} \rangle
        = 
        \langle MG(w) , \wt w - w^\star \rangle \;.
        \label{eq:proof-master-decomp}
    \end{equation}
    Here, the final equality follows from the 
    fact that $\wt w \in \WMEQ_M(G)$, and so
    $\langle MG(w), w - \wt w \rangle = 0$.
    The remainder of the proof is thus
    devoted to establishing that 
    $\langle MG(w), \wt w - w^\star \rangle = 0$.

    For this, observe that since $\wt w \in \WMEQ_M(G)$, 
    then evaluating the WMEQ condition in \eqref{eq:wmeq-app} 
    at $w^{\star}$ gives
    \begin{equation}
     \label{eq:zero-eq1}
        \langle MG(w^{\star}), w^{\star} - \wt w \rangle = 0\,.
    \end{equation}
    Moreover, since $w^{\star} \in \NE(G)$, the 
    playerwise Nash inequalities and the positivity of the weights 
    $(m_i)_{i\in[N]}$ (defining the matrix $M$) imply: 
    \begin{equation}
    \label{eq:ineq1}
        \langle MG(w^{\star}), w -  w^{\star} \rangle \geq 0\,,\quad 
        \forall w \in \calW\,.
    \end{equation}
    Therefore, combining \eqref{eq:zero-eq1} and \eqref{eq:ineq1} yields: 
    \begin{equation}
    \label{eq:ineq2}
        \langle MG(w^{\star}), w -  \wt w \rangle \geq 0\,,
        \quad \forall w \in \calW\,.
    \end{equation}
    Now, since $\wt w \in \relint(\calW)$, 
    we have that for any $h \in \calS^\bot$, 
    there exists a sufficiently small $\eps > 0$ such that 
    $\wt w + \eps h \in \calW$ and $\wt w - \eps h \in \calW$. 
    Substituting these two points in \eqref{eq:ineq2} 
    gives two inequalities: $\langle MG(w^\star), h \rangle \geq 0$
    and $\langle MG(w^\star), h \rangle \leq 0$. 
    Hence, $\langle MG(w^\star), h \rangle = 0$ for every $h \in \calS^\bot$.
    This implies that $MG(w^\star) \in(\calS^\bot)^\bot=\calS$. 

    Now let $v = \frac 12 (w + w^{\star})$,
    and observe that $v \in \calW$ by convexity of $\calW$.
    By linearity of $G$, we also have
    $G(v) = \frac{1}{2} (G(w) + G(w^\star))$. 
    Moreover, since $\wt w \in \WMEQ_M(G)$, it follows 
    again from the definition of $\WMEQ_M(G)$ in \eqref{eq:wmeq-app} 
    used at $v \in \calW$ that 
    $\langle MG(v), v - \wt w \rangle = 0.$ 
    Multiplying this identity by 4, substituting
    in the definition of $v$, and expanding then yields
    \begin{equation}
        0 = \langle MG(w), w - \wt w\rangle + \langle MG(w), w^{\star} - \wt w\rangle + \langle MG(w^{\star}), w - \wt w\rangle + \langle MG(w^{\star}), w^{\star} - \wt w\rangle\,.
        \label{eq:last-decomp}
    \end{equation}
    Since $\wt w \in \WMEQ_M(G)$ and  $w, w^{\star} \in \calW$, 
    the first and fourth terms in~\eqref{eq:last-decomp} vanish.
    Moreover, the third term also vanishes
    due to the previously-established 
    fact that $MG(w^\star) \in \calS$ 
    and since $w-\wt w \in \calS^\bot$ 
    (which follows from Lemma~\ref{lem:normal-diff}). 
    We thus conclude from~\eqref{eq:last-decomp} that
    $\langle MG(w), w^\star - \wt w \rangle = 0$.
    Substituting this identity back
    into~\eqref{eq:proof-master-decomp}
    then yields 
    $\langle MG(w), w - w^{\star} \rangle = 0$,
    which concludes the proof. 
    \end{proof}

\medskip
\paragraph{Zero-sum games with an
interior equilibrium.} 
When $G$ is the loss operator
associated with a two-player
zero-sum game with payoff matrix 
$A \in \mathbb{R}^{d_1 \times d_2}$, 
we further show that
if $\NE(G) \cap \relint(\calW)$
is nonempty 
(meaning there exists an interior
equilibrium)
then $G$ is harmonic
with identity weights matrix $I_D$.

\begin{restatable}[Zero-sum
games with interior NE]
    {prop}{propzsginterior}
    \label{prop:zsg-interior}
    Let $G: \calW \to \R^D$
    be the loss operator
    associated to a zero-sum game
    with payoff matrix $A \in \R^{d_1 \times d_2}$.
    Suppose that 
    $\NE(G) \cap \relint(\calW) \neq \emptyset$.
    Then $G$ satisfies 
    Proposition~\ref{prop:harmonic-mixed}
    with weights matrix $M = I_D$.
\end{restatable}

This result was shown 
by~\cite{legacci2024no}
(Proposition A.10), 
but we give a separate proof here:

\smallskip
\begin{proof}
    For a zero-sum game with payoff
    matrix $A$, 
    the operator $G$ is the
    linear operator $G(w) = Kw$ 
    for the skew-symmetric
    matrix $K = -K^\top \in \R^{D \times D}$
    defined as follows: 
    \begin{equation*}
        \textstyle
        K = \begin{pmatrix}
            0 & A \\
            -A^\top & 0 
        \end{pmatrix} \;.
    \end{equation*}
    See, e.g.,~\cite{wibisono2022alternating, lazarsfeld2026}.
    Moreover, recall by skew-symmetry
    that $\langle x, K x \rangle =
    \langle Kx, x \rangle = 0$
    for any $x\in\R^D$.
    Let $w^\star \in \NE(G) \cap \relint(\calW)$. 
    Then, by Proposition~\ref{prop:ne-vs-normal}, 
    we have for all $w \in \calW,$ 
    \begin{align*}
        0 = 
        \langle G(w^\star), w-w^\star \rangle
        &= 
        \langle Kw^\star, w-w^\star \rangle\\
        &=
        \langle 
        Kw, w-w^\star
        \rangle
        - 
        \langle K(w-w^\star), w-w^\star \rangle  \\
        &= 
        \langle Kw, w-w^\star \rangle \\
        &= 
        \langle G(w), w-w^\star \rangle \;,
    \end{align*}
    where the fourth equality holds
    due to skew-symmetry of $K=-K^\top$ which gives $\langle K(w-w^\star), w-w^\star \rangle = 0$.
    Thus, $w^\star \in \WMEQ_M(G)$ 
    with weights matrix $M = I_D$,
    which by Proposition~\ref{prop:harmonic-mixed}
    shows that the underlying game
    is harmonic. 
\end{proof}

Note that not every
zero-sum game is harmonic
(see, e.g.,~\cite{candogan2011flows, abdou2022decomposition, legacci2024no}), 
and Proposition~\ref{prop:zsg-interior}
requires the assumption
that the game has an interior equilibrium.

\subsection{Relationship Between Gap Function and Bregman Divergence}

\smallskip
Here, we establish 
the relationships between
the Bregman divergence and
the $\Gap$ function,
reported in 
Remark~\ref{rem:gap-vs-bregman}:

\smallskip
\begin{restatable}[Gap vs. Bregman divergence from WMEQ point]
	{prop}{propgapvsbregman}
	\label{prop:gap-vs-bregman}
    Let $G: \calW \to \R^D$
    satisfy Assumption~\ref{ass:operator} and suppose that $\WMEQ_M(G) \neq \emptyset$
    for some weights matrix $M$. 
    Let $Q: \calW \to \R$
    be a regularizer satisfying
    Assumption~\ref{ass:reg-sc}.
    	Define $\rho := \max_{w', w'' \in \calW} \|w' - w''\|$.
	Then for every $w^\star \in \WMEQ_M(G)$ and
    every $w \in \relint(\calW)$,
	\begin{equation*}
		\textstyle
		\frac{1}{2(L\rho)^2} \cdot \Gap(w)^2
		\;\le\;
		\frac{1}{2} \|w - w^\star \|^2
		\;\le\; 
		D_Q(w^\star, w) \;.
	\end{equation*}
\end{restatable}

\begin{proof} 
Let $w^\star \in \WMEQ_M(G)$. 
For the first inequality, recall by 
Lemma~\ref{lem:wmeq-ne} 
that $w^\star \in \WMEQ_M(G)$ implies
$G(w^\star) \in \calS$.
Thus $\langle G(w^\star), w - w' \rangle = 0$
for all $w, w' \in \calW$. 
Thus for all $w, w' \in \calW$, we have
\begin{align}
	\langle G(w), w - w' \rangle 
	= 
	\langle G(w) - G(w^\star), w - w' \rangle  
	&\le
	\|G(w) - G(w^\star) \|_{*} \cdot  \|w-w'\| 
	\le 
	L \rho \cdot \|w - w^\star\|  \;,
	\label{eq:gvbd-01}
\end{align}
where the final inequality comes from 
the $L$-Lipschitzness of $G$ 
and the definition of $\rho$. 
Then using the definition of 
$\Gap(w) = \max_{w' \in \calW} \langle G(w), w-w'\rangle$, 
the first inequality of the proposition follows by
squaring both sides of~\eqref{eq:gvbd-01} and rearranging.
The second inequality of the proposition
follows directly from 1-strong-convexity
of $Q$ with respect to the norm $\|\cdot \|$.
\end{proof}

\medskip
Finally, recall that for any
$w\in \calW$, if $\Gap(w) \le \eps$, 
then $w$ is an $\eps$-approximate NE
(see, e.g.,~\cite{facchinei2003finite},
Section 1.5.3).

\medskip
\subsection{Examples of Harmonic Games}
\label{app:harmonic:examples}

In this section, we give several
additional examples of harmonic games
that go beyond the canonical
subclass of two-player zero-sum games
with an interior NE (e.g., see Proposition~\ref{prop:zsg-interior}).
In particular, we give two non-trivial examples
of harmonic games that are
(1) multiplayer,
(2) non-constant sum, and 
(3) have multiple equilibria.

\smallskip
\subsubsection{Example 1: 3-player 
    polymatrix game with
    WMEQ continuum}
\label{app:harmonic-ex1}

First, we give a three-player 
polymatrix game example. 
In particular, the plots of Figure~\ref{fig:intro}
use an instantiation of this family.
Note that while the underlying operator $G$
is linear for this example 
(due to Proposition~\ref{prop:linear-polymatrix}),
the underlying game itself is 
non-constant-sum,
has a continuum of WMEQ points,
and the operator is non-monotone.
We give these details here:

\paragraph{Weighted polymatrix setup.}~ 

\noindent
Consider three players with mixed-strategy  
space $\calW=(\Delta_2)^3$;~i.e. each player has $2$ pure strategies.
For any $w_i=(p_i,1-p_i) \in \Delta_2$,
let $r_i=2p_i-1$. Then for positive weights
$m_1,m_2,m_3 > 0$ that are not all equal, 
define the losses
$\ell_1(w)=r_1(m_2r_2-m_3r_3)$,
$\ell_2(w)=r_2(m_3r_3-m_1r_1)$, and
$\ell_3(w)=r_3(m_1r_1-m_2r_2)$.
The joint loss operator $G=(G_i)_{i=1}^3$ is given by
\begin{equation*}
f(r)=
\begin{pmatrix}
m_2r_2-m_3r_3\\
m_3r_3-m_1r_1\\
m_1r_1-m_2r_2
\end{pmatrix} \;,
\quad\text{and}\quad
G_i(w)=
\begin{pmatrix}
f_i(r)\\
-f_i(r)
\end{pmatrix}.
\end{equation*}
Since each $r_i=w_{i,1}-w_{i,2}$, the 
operator $G$ is linear in $w$.

\paragraph{Continuum of WMEQ points.} ~

\noindent
Let $M=\Diag(m_1I_2,m_2I_2,m_3I_2)$ and
$m_{\min}=\min_i m_i$. 
It is straightforward to verify that 
the set $\WMEQ_M(G)$ is then given by
\begin{equation}
\textstyle
\WMEQ_M(G)
=
\{
 w\in\calW:
    r_i= (t/m_i)
    \ \text{for all $i,\;\; |t|\le m_{\min}$}
\} \;.
\end{equation}
Equivalently, in terms of 
the coordinates $p^\star_i \in [0, 1]$, 
each $w^\star = (w^\star_1, w^\star_2, w^\star_3) \in \WMEQ_M(G)$ 
is parameterized
by $p^\star_i=(1+t/m_i)/2$ for $|t|\leq m_{\min}$.
It is straightforward to verify the WMEQ identity for every point
of this segment by substituting into
$\langle MG(w),w-w^\star\rangle$.
Conversely, for any WMEQ point $w^\star$,
the expression $\langle MG(w),w-w^\star\rangle$
must vanish for every $w \in \calW$.
Since this expression is linear in the independently
varying coordinates $r_1, r_2, r_3$, the factor multiplying
each coordinate must be zero, which yields
$m_1 r^\star_1 = m_2 r^\star_2 = m_3 r^\star_3 = t$. 
Thus $r^\star_i = t/m_i$ and $|t| \le m_{\min}$, and so every 
WMEQ point lies on the segment. 
As $G$ is linear and clearly 
$\WMEQ_M(G) \cap \relint(\calW) \neq \emptyset$,
it then follows by Proposition~\ref{prop:ne-wmeq-linear}
that $\NE(G) = \WMEQ_M(G)$.

As an example, the left-subplot in Figure~\ref{fig:example}
plots the continuum of $\WMEQ_M(G)$
for the case
when $m_1 = 3, m_2 = 6, m_3 = 4$.

\paragraph{Nonmonotonicity and nonconstant sum.} ~

\noindent
While the operator $G$ is linear,
we show that $G$ is non-monotone,
and the underlying game is additionally
non-constant-sum. 
For this,  
let $w^0$ be the uniform strategy profile.
Then by definition of $G$ and $\ell_i$, 
we can compute for any other $w \in \calW$
\begin{align}
\langle G(w)-&G(w^0),w-w^0\rangle \nonumber \\
&=\ssum\nolimits_{i=1}^3\ell_i(w)
=(m_2-m_1)r_1r_2
 +(m_1-m_3)r_1r_3
 +(m_3-m_2)r_2r_3.
 \label{eq:e1-sum}
\end{align}
As the weights $\{m_i\}$
are not all equal,
at least one term in~\eqref{eq:e1-sum}
has non-zero coefficient $m_i - m_j$.
Supposing without loss of generality 
that $m_ 2 - m_1 > 0$,
then set $r_3 = 0$, 
and choose $w_1, w_2 \in\Delta_2$
such that $r_1 r_2 < 0$.
Then it follows from~\eqref{eq:e1-sum}
that $\langle G(w)  - G(w^0), w-w^0 \rangle < 0$,
and thus $G$ is non-monotone.
Moreover, at this profile, 
\eqref{eq:e1-sum} also implies that
the total loss $\sum_{i=1}^3 \ell_i(w)$ is strictly negative.
In contrast, when $w = w^0$,
the total loss is zero. Thus, 
the game is also non-constant-sum. 

\subsubsection{Example 2: non-linear operator with
    non-convex NE set.}
\label{app:harmonic-ex2}

\smallskip
Here, we give an example
of a harmonic game that
does not correspond to a linear polymatrix game.
In particular, the following example induces
a non-linear loss operator $G$ that 
has a continuum of WMEQ points,
a nonconvex set of NE points, 
is nonmonotone, corresponds
to a non-constant-sum game, 
and which additionally satisfies
the restricted regularity condition
of Assumption~\ref{ass:restricted-regularity}.
We give more details here:

\begin{figure}[tb!]
    \centering
    \includegraphics[width=0.8\linewidth]{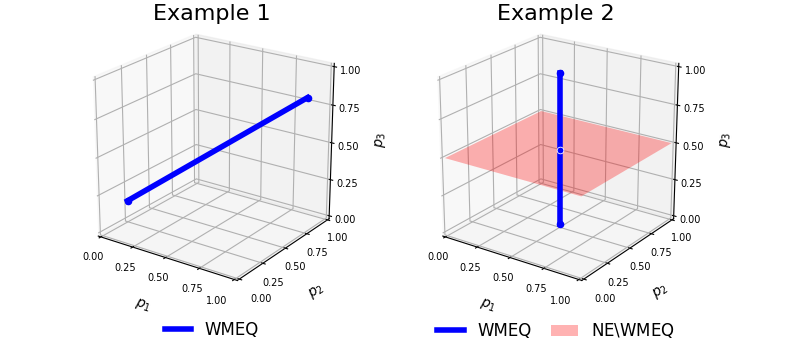}
    \caption{The sets $\WMEQ_M(G)$
    and $\NE(G) \setminus \WMEQ_M(G)$ 
    for the instantiation 
    of Example 1 from Section~\ref{app:harmonic-ex1}
    with $m_1 = 3$, $m_2 = 6$, $m_3 = 4$ (left),
    and Example 2 from Section~\ref{app:harmonic-ex2}.}
    \label{fig:example}
\end{figure}

\paragraph{Game setup.} ~ 

\noindent
Consider three players with 
mixed-strategy space $\calW=(\Delta_2)^3$.
Write $w_i= (w_{i, +}, w_{i, -}) = (p_i,1-p_i)$ and $r_i=2p_i-1$, 
and let $w(r)$ denote the strategy profile with 
centered coordinates $r\in[-1,1]^3$.
Define the losses $\ell_1(w)=r_1r_2r_3$,
$\ell_2(w)=-r_1r_2r_3$, and $\ell_3(w)=(w_{3, +} + w_{3, -}) r_1r_2$.

\paragraph{Nonlinear loss operator.} ~ 

\noindent
Given $\ell_1, \ell_2, \ell_3$, 
the joint loss operator $G$
as in~\eqref{eq:G-def} is given by
\begin{equation*}
    G(w)=
    \big(r_2r_3,\,-r_2r_3,\,-r_1r_3,\,r_1r_3,\,r_1r_2,\,r_1r_2
    \big) \;.
\end{equation*}
The cross-interactions between $r_i, r_j$
mean that $G$ is nonlinear on $\calW$.
For example, since $G_{1,+}(w)=(2p_2-1)(2p_3-1)$, we have 
$\partial^2G_{1,+}/(\partial p_2\,\partial p_3)=4$.
Thus, $G$ is not affine in $p$, and so 
there is no matrix $P$ such that $G(w)=Pw$ for $w \in \calW$.

\paragraph{Continuum of MEQ points.} ~ 

\noindent
Let $M=I_D$, where $D=6$. 
For any centered profiles $r,s$, 
we can compute
\begin{equation}
\langle G(w(r)),w(r)-w(s)\rangle
=
-s_1r_2r_3+s_2r_1r_3.
\label{eq:e2-0}
\end{equation}
Observe 
that~\eqref{eq:e2-0} vanishes for every $r\in[-1,1]^3$ if and only if
$s_1=s_2=0$. 
Hence, 
\begin{equation}
    \WMEQ_{I_D}(G)
    =
    \MEQ(G)
    =
    \{w(0,0,t):t\in[-1,1]\}.
\end{equation}
Equivalently, the MEQ set is the segment $p_1=p_2=1/2$,
$p_3\in[0,1]$. 
Thus, since $\WMEQ_M(G)$ contains
the uniform profile in $\relint(\calW)$
with $M = I_D$, the underlying 
game is harmonic due to variational
characterization of Proposition~\ref{prop:harmonic-mixed}.

\paragraph{Nonconvex Nash equilibrium set.} ~ 

\noindent
It is straightforward to verify that a profile 
corresponds to an NE if and only
if $r_2r_3=r_1r_3=0$.
Thus, including boundary profiles,
we have that 
\begin{align*}
\NE(G)
&=
\MEQ(G)
\cup
\{w(r_1,r_2,0):(r_1,r_2)\in[-1,1]^2\}, \\
\text{and thus}\;\;\;
\NE(G)\setminus\MEQ(G)
&=
\{w(r_1,r_2,0):
(r_1,r_2)\in[-1,1]^2\setminus\{(0,0)\}\}.
\end{align*}
In particular, there is a continuum of interior
but non-MEQ Nash
equilibria. 
Thus, the full set $\NE(G)$ 
is nonconvex: $w(0,0,1/2)$ and
$w(1/2,0,0)$ are both in $\NE(G)$, 
while their midpoint $w(1/4,0,1/4) \notin \NE(G)$.
In particular, the right-subplot in Figure~\ref{fig:example}
plots both the continuum of $\WMEQ_M(G) = \MEQ(G)$
in blue and the continuum of $\NE(G)\setminus \WMEQ_M(G)$ 
in red. The full set $\NE(G)$ is the union
of the two sets, which illustrates that $\NE(G)$
is non-convex.

\smallskip
\paragraph{Nonmonotonicity and nonconstant sum.} ~

\noindent
For $r=(1/2,0,1/2)$ and $s=(1/2,-1/2,0)$, 
we have using a direct calculation that 
\begin{equation*}
    \textstyle
\langle G(w(r))-G(w(s)),w(r)-w(s)\rangle=-\frac18<0.
\end{equation*}
Thus $G$ is not monotone. As
$\sum_{i=1}^3\ell_i(w)=r_1r_2$ varies with $w$, 
the game is also
not constant sum. 

\smallskip
\paragraph{Restricted regularity.} ~

\noindent
When the regularizer
$Q$ and its energy function $H = Q^*$
satisfy Assumptions~\ref{ass:legendre}
and~\ref{ass:lhs}, we
further show that $G$ satisfies
the regularity condition 
from Assumption~\ref{ass:restricted-regularity}.

For this, first consider the
standard (componentwise) 
negative Shannon entropy regularizer
$Q(w)=\sum_{i=1}^3\sum_{a\in\{+,-\}}w_{i,a}\log w_{i,a}$,
and let $H = Q^*$,
which satisfy~Assumptions~\ref{ass:legendre}
and~\ref{ass:lhs}, respectively. 
In reduced dual coordinates $b_i=(z_{i,+}-z_{i,-})/2$, 
its choice map is $r_i=\tanh b_i$, where $\tanh$ is the hyperbolic tangent.
Since player 3's loss block is normal to the simplex, 
the effective dual update space and reachable set are
\begin{equation*}
\calU=\{(u_1,-u_1,u_2,-u_2,0,0):u\in\R^2\},
\qquad
\calR(w_0)=\{w(r_1,r_2,c):(r_1,r_2)\in(-1,1)^2\} \;,
\end{equation*}
where $c=r_{3,0}=\tanh b_{3,0}$.
Identifying $\calU$ with $\mathbb{R}^2$, 
the effective stationarity map is then given by
\begin{equation*}
\pi_0(u)=
\begin{pmatrix}
c\tanh(b_{2,0}+u_2)\\
-c\tanh(b_{1,0}+u_1)
\end{pmatrix}.
\end{equation*}
If $c\neq0$, the unique zero is
$u^\star=(-b_{1,0},-b_{2,0})$.
Then we have 
$
J_{\pi_0}(u^\star)=
\begin{pmatrix}
0&c\\
-c&0
\end{pmatrix}
$
and thus 
$\det J_{\pi_0}(u^\star)=c^2>0$.
If $c=0$, the entire reachable set 
$\calR(w_0) \subset \NE(G)$,
and this case is excluded by Assumption~\ref{ass:restricted-regularity}.
Thus, for initializations 
with $c\neq 0$, Assumption~\ref{ass:restricted-regularity} holds. 

Note also  that the same argument applies to 
any other separable regularizer
$Q$  and energy function $H = Q^*$
satisfying Assumptions~\ref{ass:legendre}
and~\ref{ass:lhs}, respectively:
for this, we can replace $\tanh$ by the 
scalar choice maps $h_i$, whose derivatives 
are strictly positive,
which follows from the fact that
$\Null(\nabla^2H(z))=\calS$ 
(see Proposition~\ref{prop:second-order-duality} (iii)).
At a zero, the determinant  of $J_{\pi_0}$
then becomes
$c^2h_1'(b_1^\star)h_2'(b_2^\star)>0$ when $c\neq0$,
where $h_i(b_i^\star)=0$, 
and the case $c=0$ remains excluded
from the assumption.


\medskip
\section{Details on FTRL Dynamics}
\label{app:ftrl-prelims}

In this section, we give more details
on the global dynamics of FTRL variants
introduced in Section~\ref{sec:algorithms}.
First, we give in Section~\ref{app:ftrl-prelims:system}
a more complete derivation 
of the global dynamics from 
their playerwise definitions.
Then, we present in Section~\ref{app:ftrl-prelims:primal-dual}
several key primal-dual relationships
of the global regularizer and corresponding
energy function that are used throughout. 
This also includes the proof of
Proposition~\ref{prop:ftrl-primal-dual-invariant}.

\subsection{Playerwise and Global FTRL Dynamics}
\label{app:ftrl-prelims:system}

Recall from Section~\ref{sec:algorithms}
that each player $i \in [N]$
uses a strictly convex regularizer $Q_i : \calW_i \to \R$
and a fixed stepsize $\eta_i > 0$.
Recall also that we define
the global stepsize $\eta := \max_{i \in [N]} \eta_i > 0$.

\paragraph{Playerwise primal update.}
The standard playerwise FTRL update 
rule~\citep[Section 2.3]{shalev2012online},
along with its Optimistic, Extragradient, and Implicit variants,
all follow a similar primal update structure. 
First, we assume each player $i \in [N]$
has initial $w_{0, i} \in \relint(\calW_i)$.
Then for $t \ge 0$, they choose
$w_{t+1, i} \in \calW_i$ according to
\begin{equation}
	\textstyle
	w_{t+1, i} 
	= 
	\argmin_{w_i \in \calW_i} \;
	\{ 
		- \langle w_i, v_{t+1, i} \rangle + Q_i(w_i) 
	\} 
	\label{eq:ftrl-primal-player}
\end{equation}
where $v_{t+1, i} \in \R^{d_i}$ is  
a variant-specific \textit{dual iterate} (specified below)
that accumulates $\eta_i$-scaled operator vectors. 
Note here that the argmin in~\eqref{eq:ftrl-primal-player}
is unique due to strict convexity of $Q_i$.

\paragraph{Playerwise, variant-specific dual iterates.}
Recall that $w_t = (w_{t, 1}, \dots, w_{t, N}) \in \R^D$
is the concatenation of the primal iterates of all the players 
at time $t \ge 0$.  
Moreover, for each $i \in [N]$, 
each player has dual initialization $v_{0, i} \in \R^{d_i}$.
Then, the variant-specific
dual iterate $v_{t+1, i} \in \R^{d_i}$ is given by:
\begin{equation}
	\begin{aligned}
		\text{(FTRL)}&
		\quad
		v_{t+1, i}
		= 
		v_{t, i}
		- 
		\eta_i G_i(w_t)\,, \\
		\text{(IFTRL)}&
		\quad
		v_{t+1, i}
		= 
		v_{t, i}
		- 
		\eta_i G_i(w_{t+1})\,, \\
		\text{(OFTRL)}&
		\quad
		v_{t+1, i}
		= 
		v_{t, i}
		- 
		2 \eta_i G_i(w_t)
		+ \eta_i G_i(w_{t-1})\,, \\
		\text{(EGFTRL)}&
		\quad
		v_{t+1/2, i}
		= 
		v_{t, i}
		- 
	 	\eta_i G_i(w_{t})  \,,
	 	\;\;\;
	 	v_{t+1, i}
		= 
		v_{t, i}
		- 
	 	\eta_i G_i(w_{t+1/2})  \;.
	\end{aligned}
	\label{eq:playerwise-dual-unscaled}
\end{equation}
Here, we use the convention for
OFTRL that $w_{-1} = w_0$,
and for EGFTRL, 
we let $w_{t+1/2}$
denote the concatenation 
of the minimizers in~\eqref{eq:ftrl-primal-player}
with respect to 
$v_{t+1/2} = (v_{t+1/2, i})_{i\in[N]}$.
When instantiating~\eqref{eq:ftrl-primal-player}
with one of the dual iterate definitions 
in~\eqref{eq:playerwise-dual-unscaled},
we exactly obtain the standard form of 
the corresponding FTRL-type update rule
(see, e.g.,~\cite{syrgkanis2015fast} and
~\cite{legacci2024no}).

Then, for a fixed playerwise dual  sequence $\{v_{t, i}\}$,
we also use the global stepsize $\eta >0$ to define
the $\eta$-normalized 
dual iterates $\{z_{t, i}\}$.
Specifically, for all $i \in [N]$, 
let $z_{t, i} = \frac{\eta}{\eta_i} \cdot v_{t, i}$
for all $t \ge 0$.
It then follows that, for any of the dual sequences
in~\eqref{eq:playerwise-dual-unscaled},
the playerwise update
in~\eqref{eq:ftrl-primal-player} for $t \ge 0$
can be equivalently written in terms of $z_{t, i}$ as
\begin{equation}
	\textstyle
	w_{t+1, i} 
	= 
	\argmin_{w_i \in \calW_i} \;
	\{ 
		- \langle w_i, \frac{\eta_i}{\eta} z_{t+1, i} \rangle + Q_i(w_i) 
	\} 
	=
	\argmin_{w_i \in \calW_i} \;
	\{ 
		- \langle w_i, z_{t+1, i} \rangle + \frac{\eta}{\eta_i} Q_i(w_i) 
	\}  \;.
	\label{eq:ftrl-primal-player-scaled}
\end{equation}
Moreover, the variant-specific \textit{scaled} dual sequences
are explicitly given by:
\begin{equation}
	\begin{aligned}
		\text{(FTRL)}&
		\quad
		z_{t+1, i}
		= 
		z_{t, i}
		- 
		\eta G_i(w_{t})\,, \\
		\text{(IFTRL)}&
		\quad
		z_{t+1, i}
		= 
		z_{t, i}
		- 
		\eta G_i(w_{t+1})\,, \\
		\text{(OFTRL)}&
		\quad
		z_{t+1, i}
		= 
		z_{t, i}
		- 
		2 \eta G_i(w_t)
		+ \eta G_i(w_{t-1})\,, \\
		\text{(EGFTRL)}&
		\quad
		z_{t+1/2, i}
		= 
		z_{t, i}
		- 
	 	\eta G_i(w_{t})  \,,
	 	\;\;\;
	 	z_{t+1, i}
		= 
		z_{t, i}
		- 
	 	\eta G_i(w_{t+1/2})  \;.
	\end{aligned}
	\label{eq:playerwise-dual-scaled}
\end{equation}

\paragraph{Global primal update.}
Now consider the collection of any of the scaled
playerwise dual sequences of~\eqref{eq:playerwise-dual-scaled},
and let $z_t = (z_{t, 1}, \dots, z_{t,N}) \in \R^D$ for 
all $t \ge 0$. Moreover, recall from~\eqref{eq:global-reg}
that we define the global regularizer $Q: \calW \to \R$
by
\begin{equation*}
	\textstyle
	Q(w) 
	= 
	\sum_{i=1}^N
	\frac{\eta}{\eta_i} Q_i(w_i)
	\;\;\;\text{for all $w = (w_1, \dots, w_N) \in \calW$} \;.
\end{equation*}
Now, we use the fact that $\calW = \calW_1 \times \dots \times \calW_N$
is a product space and 
that $Q$ is separable with respect to the player coordinates $ i \in [N]$.
Because of this, for $w_t = (w_{t, 1}, \dots, w_{t, N}) \in \R^D$ 
where $w_{t, i}$ is defined in~\eqref{eq:ftrl-primal-player-scaled},
it holds that 
\begin{equation}
	\textstyle
	w_{t} 
	= 
	\argmin_{w \in \calW} \;
	\{ 
		- \langle w, z_{t} \rangle + Q(w) 
	\} \;,
	\label{eq:ftrl-global-primal-app}
\end{equation}
which is exactly the global primal update rule 
from~\eqref{eq:ftrl-primal}.

\paragraph{Global dual updates.}
Fix any of the playerwise dual sequences 
in~\eqref{eq:playerwise-dual-scaled},
and let $z_t = (z_{t, i})_{i=1}^N \in \R^D$
for all $t \ge 0$.
For the sequence $\{w_t\}$ 
using~\eqref{eq:ftrl-global-primal-app},
it is then straightforward to see 
that we then exactly recover the global  dual
sequences of~\eqref{eq:ftrl},
\eqref{eq:iftrl},~\eqref{eq:oftrl},
and~\eqref{eq:egftrl} from Section~\ref{sec:algorithms}.

\subsection{Primal-Dual Relationships
Between Regularizer and Energy}
\label{app:ftrl-prelims:primal-dual}

We now present several key properties of
the global regularizer $Q: \calW \to \R$
and its corresponding
energy function (convex conjugate) $H = Q^*: \R^D \to \R$.
Recall that $H$ is given by 
\begin{equation}
	\textstyle
	H(z) := 
	Q^*(z)
	= 
	\max_{w \in \calW} \,
	\{
		\langle w,  z \rangle - Q(w)
	\}
	\;\;\;\text{for all $z \in \R^D$}\;.
	\label{eq:conjugate}
\end{equation}
For further background
and preliminaries on conjugate functions,
see, e.g.,~\citet[Section 3.3]{boyd2004convex}.
Moreover, recall from Section~\ref{sec:algorithms:assumptions}
that we consider the following assumption:

\asslegendre*

We remark that Assumption~\ref{ass:legendre}
is the standard notion of Legendre-type 
functions (see, e.g.~\cite{rockafellar1997convex}, Chapter 26), 
but defined with respect to
the relative interior of $\calW$
(which is lower-dimensional).
Nevertheless, Assumption~\ref{ass:legendre}
implies several standard primal-dual relationships
between $Q$ and $H$,
which we state in the following sequence of propositions.
We first make the following technical remark
regarding differentiability and our notation:

\paragraph{Remark on differentiability.}
Throughout, derivatives of the regularizer $Q$
are taken only on $\relint(\calW)$
and relative to its affine hull,
whose direction space is $\calS^\bot$.
Specifically, $\nabla Q(w) \in \calS^\bot$
is the unique vector satisfying
$\frac{d}{dt} Q(w + tv)|_{t=0} = \langle \nabla Q(w), v \rangle$
for every $v \in \calS^\bot$.
Similarly, for each $w \in \relint(\calW)$,
the Hessian $\nabla^2 Q(w): \calS^\bot \to \calS^\bot$
is the derivative of this relative gradient 
along the same directions. 
The steepness condition of 
Assumption~\ref{ass:legendre} (ii)
should thus be interpreted relative to the affine hull,
meaning the norm of the relative gradient $\nabla Q(w)$
goes to infinity as $w$ approaches
the relative boundary $\rbd(\calW)$ of $\calW$.

Moreover, note that relative gradients suffice to define
the Bregman divergence $\D_Q(w', w)$ for $w' \in \calW$ and
$w \in \relint(\calW)$: since $w' - w \in \calS^\bot$ by 
Lemma~\ref{lem:normal-diff}, the normal component 
of an ambient gradient, which lies in $\calS$, contributes zero 
to the inner product in the definition of $\D_Q(w', w)$.
Also, when using a smooth ambient extension of $Q$,
note that its ambient gradient can be projected onto $\calS^\bot$
to obtain the relative gradient. 
For the corresponding energy function
$H = Q^*:\R^D\to\R$, 
derivatives are ordinary ambient derivatives.

Throughout, we regard $Q$ as a proper, lower semicontinuous, 
convex function on $\calW$,
allowing the value $+\infty$ at boundary points,
and extended to $+\infty$ outside
$\calW$.
We will also assume throughout that 
Assumptions~\ref{ass:reg-sc} (strong convexity)
and~\ref{ass:lhs} (local Hessian stability) hold
for $Q$ and $H$.

\subsubsection{Relative duality and Proof of 
Proposition~\ref{prop:ftrl-primal-dual-invariant}.}

We start with the following property
which gives a useful characterization of $\nabla H$.
This should be thought of as the 
basic property of Legendre-type functions
(e.g., as in~\cite{rockafellar1997convex}, Theorem 26.5),
but adapted to hold in a relative sense: 

\smallskip
\begin{restatable}{prop}{proprelativeduality}
	\label{prop:relative-duality}
	For every $z \in \R^D$,
	the maximizer of $H(z)$ in~\eqref{eq:conjugate}
	is unique and lies in $\relint(\calW)$. 
	It holds specifically that 
	$
    \nabla H(z)
    = 
    \argmax_{w \in \calW}
    \{ \langle w, z\rangle - Q(w) \}  \in \relint(\calW)
	$.
\end{restatable}

\begin{proof}
	Fix $z \in \R^D$.
	By compactness of $\calW$ and lower semicontinuity
	of $Q$, 
	the maximum in the definition of
	$H(z)$ from~\eqref{eq:conjugate} is finite
	and attained. 
	By the relative steepness property of
	Assumption~\ref{ass:legendre} (ii), 
	every maximizer lies in~$\relint(\calW)$.
	By strict convexity of $Q$, the 
	maximizer is also unique. 
	This proves the first claim of the proposition.

	For the second claim, 
	let $w \in \relint(\calW)$ denote the maximizer
	defining $H(z)$ in~\eqref{eq:conjugate}.
	Then for every $h \in \R^D$ and $t \in \R$,
    we have that 
    \begin{align*}
        H(z+ th) 
        = \max\nolimits_{v \in \calW}\, 
        \{\langle v,  z + th\rangle - Q(v)\} 
        &\ge 
        \langle w,  z + th \rangle - Q(w) \\
        &= 
        \langle w,  z \rangle - Q(w) + t \langle w, h \rangle 
        = H(z) + t \langle w, h \rangle\,,
    \end{align*}
    which gives the identity
	$H(z+th) \ge H(z) + t\langle w, h\rangle$.
	Then taking the left-hand derivative when $t \uparrow 0$ 
    and the right-hand derivative when $t \downarrow 0$ 
    on both sides of the identity, 
    and using the differentiability 
	of $H$ (which holds by Assumption~\ref{ass:lhs}), 
    we obtain $\langle \nabla H(z), h \rangle = \langle w, h \rangle$
	for all $h \in \R^D$. 
	Thus, $\nabla H(z) = w$. 
\end{proof}

\smallskip
Note  that as a consequence of Proposition~\ref{prop:relative-duality},
we obtain a short proof of the
primal-dual invariants of the global FTRL-type 
dynamics from Proposition~\ref{prop:ftrl-primal-dual-invariant},
restated here:

\propftrlprimaldual*

\begin{proof}
	Observe, by the global primal update rule
	of~\eqref{eq:ftrl-primal}, that
	\begin{equation*}
		\textstyle
		w_t = 
		\argmin_{w \in \calW}
		\{
			- \langle w, z_t \rangle
			+ Q(w) 
		\} 
		= 
		\argmax_{w \in \calW}
		\{
			\langle w, z_t \rangle
			- Q(w)
		\} \;.
	\end{equation*}
	Then by Proposition~\ref{prop:relative-duality},
	we have $w_t = \nabla H(z_t) \in \relint(\calW)$.
\end{proof}

\smallskip
\subsubsection{First-Order and Second-Order Primal-Dual Relationships}  

We now state several further primal-dual relationships
between $Q$ and $H$ that hold 
under Assumption~\ref{ass:legendre}.
Note that these relationships 
can be thought of as the usual
consequences of Legendre duality,
but working over the quotient space $\R^D / \calS$. 

\medskip
\begin{restatable}{prop}{propfirstorderduality}
	\label{prop:first-order-duality}
	For any $z \in \R^D$ and $w \in \relint(\calW)$, 
	$w = \nabla H(z)$
	$\iff$ $z - \nabla Q(w) \in \calS$.
\end{restatable}

\begin{proof}
Fix $z\in\R^D$ and $w\in\relint(\calW)$. 
For convenience, we will write 
$F_z(w'):=Q(w')-\langle w',z\rangle$
for any $w' \in \calW$. 
By Proposition~\ref{prop:relative-duality}, 
recall that 
$w = \nabla H(z)$ if and only if 
$w = \argmin_{w' \in \calW}\{F_z(w')\}$.
We will then prove both implications of 
the proposition statement separately. 

For the forward implication, suppose that
$w=\nabla H(z).$ 
Then $w$ minimizes $F_z$ over $\calW$.
Since $w\in\relint(\calW)$, for every $h\in\calS^\bot$
there exists $\eps>0$ such that $w+th\in\calW$ 
for all $|t|<\eps$. 
Hence, the function $\phi: \R \to \R$ given by
$\phi(t):=F_z(w+th)$ has a 
local minimum at the interior point $t=0 \in (- \eps, \eps)$. 
Therefore, for every $h\in\calS^\bot$, we have 
$0=\phi'(0)=\langle \nabla F_z(w),h\rangle=\langle \nabla Q(w)-z,h\rangle$.
Note here that we have used the fact that the minimizer of 
$\phi$ is interior. 
Thus, we have shown that 
$z - \nabla Q(w) \in(\calS^\bot)^\bot=\calS.$

Conversely, suppose that $z-\nabla Q(w)\in\calS.$  
It follows that $\langle \nabla Q(w)-z,h\rangle=0$ 
for all $h\in\calS^\bot$. 
For any $w'\in\calW$, we have
$w'-w\in\calS^\bot$ by Lemma~\ref{lem:normal-diff}. 
Since $F_z$ is convex, we have for any $w' \in \calW$, 
\begin{equation*}
    F_z(w')
\ge
F_z(w)
+
\langle \nabla F_z(w),w'-w\rangle.
\end{equation*}
Since $\nabla F_z(w)=\nabla Q(w)-z \in \calS$ by assumption 
and $w'-w\in\calS^\bot$, we obtain 
$\langle \nabla F_z(w),w'-w\rangle=0$.
Hence $F_z(w')\ge F_z(w)$ for every $w'\in\calW$. 
Thus $w$ minimizes $F_z$ over $\calW$. 
By Proposition~\ref{prop:relative-duality}, the minimizer 
is unique and equals $\nabla H(z)$. Therefore
$w=\nabla H(z)$, which concludes the proof.
\end{proof}

\smallskip
As an immediate consequence of
Proposition~\ref{prop:first-order-duality},
we also have that $\nabla H$ is 
invariant to shifts under directions in $\calS$. Formally:
\begin{restatable}{prop}{propfirstorderinvariance}
	\label{prop:shift-invariance}
	For all $z \in \R^D$ and $s \in \calS$, 
	$\nabla H(z+ s) = \nabla H(z)$.
\end{restatable}

\begin{proof}
Let $z \in \R^D$ and $s \in \calS$.
Define $w = \nabla H(z)$. 
By Proposition~\ref{prop:relative-duality}, 
we have $w \in \relint(\calW)$. Then by the 
forward implication of Proposition~\ref{prop:first-order-duality}, 
it follows that $z - \nabla Q(w) \in \calS$. 
Since $\calS$ is a linear subspace and $s \in \calS$, 
we also have that $z + s - \nabla Q(w) \in \calS$. 
Then by the backward implication of 
Proposition~\ref{prop:first-order-duality},
we have $w = \nabla H(z+s)$.
Thus $\nabla H(z) = w = \nabla H(z+s)$. 
\end{proof}

\medskip
As a direct corollary
of Proposition~\ref{prop:shift-invariance}, 
we also have the following property:

\begin{restatable}
    {prop}{propinvariancecor}
    \label{prop:invariance-corr}
    For all $z \in \R^D$
    and any $v \in \R^D$,
    $\nabla H(z + v) = 
    \nabla H(z+ \Pi_{\calS^\bot} v)$.
\end{restatable}

\medskip
Next, we establish the following
second-order primal-dual relationships:

\medskip
\begin{restatable}{prop}{propsecondorderduality}
	\label{prop:second-order-duality}
	For all $z \in \R^D$ and 
	$w = \nabla H(z) \in \relint(\calW)$, 
	\begin{enumerate}[
		label={(\roman*)},
		leftmargin=2.5em,
	]
	\item
	$\nabla^2 H(z + s) = \nabla^2 H(z)$
	for all $s \in \calS$.

	\item
	$\nabla^2 H(z) \nabla^2 Q(w) u = u$
	for all $u \in \calS^\bot$. 

	\item
	$\Null(\nabla^2 H(z)) = \calS$
	and $\Range(\nabla^2 H(z)) = \calS^\bot$.
	\end{enumerate}
\end{restatable}

\smallskip
\begin{proof}
	For (i), fix $s \in \calS$.
	Then differentiating
	the identity $\nabla H(z+s) = \nabla H(z)$ 
    from Proposition~\ref{prop:shift-invariance}
	with respect to $z$ gives
	$\nabla^2 H(z+s) = \nabla^2 H(z)$. 

	For (ii), since $\nabla Q(v) - \nabla Q(v) = 0 \in \calS$ 
    for any $v \in \relint(\calW)$, the backward implication of 
    Proposition~\ref{prop:first-order-duality} 
    with $z = \nabla Q(v)$ 
    gives $\nabla H(\nabla Q(v)) = v$ for any $v \in \relint(\calW)$. 
    Now fix $u \in \calS^\bot$. Since $w \in \relint(\calW)$, 
    then $w + tu \in \relint(\calW)$ for all sufficiently small $t$.
    Therefore, the previous identity gives 
    $\nabla H(\nabla Q(w + tu)) = w + tu$.
    Using the chain rule to differentiate with respect 
    to $t$ at $t = 0$ then yields: 
    	\begin{equation}
		\nabla^2 H(\nabla Q(w)) \nabla^2 Q(w)u = u \;.
		\label{eq:dhelper01}
	\end{equation}
    Since $w = \nabla H(z) \in \relint(\calW)$, it follows from Proposition~\ref{prop:first-order-duality} 
    that $\nabla Q(w) -z = - (z- \nabla Q(w)) \in \calS$. 
    Then using (i) yields 
    $\nabla^2 H(\nabla Q(w)) = 
    \nabla^2 H(z + \nabla Q(w) - z) = \nabla^2 H(z)$.
    Finally, substituting this identity into~\eqref{eq:dhelper01} 
    yields claim (ii).

	For claim (iii), we first prove that $\Null(\nabla^2 H(z)) = \calS$
    by double inclusion.
    For the forward direction, if $s \in \calS$, 
    then $\nabla H(z + t s) = \nabla H(z)$ 
    for any $t \in \R$
    by the shift invariance property of 
    Proposition~\ref{prop:shift-invariance}. 
    Differentiating this identity with respect to $t$ at $t = 0$ gives 
    $\nabla^2 H(z) s = 0$, meaning $s \in \Null(\nabla^2 H(z))$.
    
    For the second direction, suppose now that 
    $v \in \Null(\nabla^2 H(z))$.
    Since $\calS$ and $\calS^\bot$ are orthogonal subspaces,
    we can decompose $v = s + u$ with $s \in \calS$ 
    and $u \in \calS^\bot$. 
    Since $s \in \calS$, we have from the first inclusion 
    that $s \in \Null(\nabla^2 H(z))$, and hence $\nabla^2 H(z) s = 0$.    
    Moreover, since $v \in \Null(\nabla^2 H(z))$, we also have 
    $\nabla^2 H(z) v = 0$.
    Thus, $\nabla^2 H(z) u = \nabla^2 H(z) v - \nabla^2 H(z) s =  0$. 
    Now by part (ii), since $u \in \calS^{\bot}$, we have 
    $\nabla^2 H(z) \nabla^2 Q(w) u = u$. Then it follows that
    \begin{equation*}
        \|u\|^2 
        = \langle u, u \rangle 
        = \langle u, \nabla^2 H(z) \nabla^2 Q(w) u \rangle 
        = \langle \nabla^2 H(z) u,  \nabla^2 Q(w) u \rangle  = 0\,,
    \end{equation*}
    where in the third equality we used the symmetry of $\nabla^2 H(z)$ 
    and in the last equality we used the fact that $\nabla^2 H(z) u =0$.
    Thus we have shown that $u = 0$, 
    which by the decomposition of $v = s + u$ implies that
    $v = s \in \calS$. 
    Combining both inclusions gives $\Null(\nabla^2 H(z)) = \calS$
    and establishes the first part of claim (iii). 
    
    Finally, as $\nabla^2 H(z)$ is symmetric, 
    we have $\Range(\nabla^2 H(z)) = \Null(\nabla^2 H(z))^\bot = \calS^\bot$,
    which establishes the second part of claim (iii).
\end{proof}

\smallskip
\subsubsection{Properties of Bregman Divergences}

Recall that the Bregman divergence
$\D_{Q}: \calW \times \relint(\calW) \to \R$ 
is given by
\begin{equation}
	\D_Q(\wt w, w) = Q(\wt w) - Q(w) - \langle \nabla Q(w), \wt w - w \rangle \;.
	\label{eq:dr-app}
\end{equation}
Here, note that $\D_Q(\wt w, w) = + \infty$
whenever $Q(\wt w) =+  \infty$.

Then we further have the following
useful primal-dual characterization of 
the Bregman divergence for~$Q$ satisfying
Assumption~\ref{ass:legendre}:

\begin{restatable}
	{prop}{propbregmanpd}
	\label{prop:bregman-pd}
	Fix $z, z' \in \R^D$, and 
	let $w = \nabla H(z), w' = \nabla H(z') \in \relint(\calW)$.
	Then for $\wt w \in \calW$
	with $Q(\wt w) < \infty$, the following 
    properties hold:
	\begin{enumerate}[
		label=({\roman*}), 
		leftmargin=3em,
	]
	\item
		$\D_Q(\wt w, w)
		= 
		Q(\wt w) + H(z) - \langle z, \wt w \rangle$
	\item
		$\D_Q(\wt w, w') - \D_Q(\wt w, w)
		=
		H(z') - H(z) - \langle \wt w, z'- z \rangle 
		$.
	\end{enumerate}

\end{restatable}

\begin{proof}
	For the first claim (i), since $w = \nabla H(z)$,
    we have by the Fenchel-Young equality 
    (see, e.g.,~\citet[Section 2.7]{shalev2012online})
    that $Q(w) + H(z) = \langle z, w\rangle$
    (indeed, this follows from 
    Proposition~\ref{prop:relative-duality},
    which establishes $w$ as the 
    maximizer in the definition
    of $H(z)$ from~\eqref{eq:conjugate}, 
    i.e., 
    $w 
	= \argmax_{w' \in \calW} \{ \langle w', z\rangle - Q(w') \}$).
	Then substituting this identity into~\eqref{eq:dr-app} 
    and simplifying yields
	\begin{align*}
		\D_Q(\wt w, w) 
		&= 
		Q(\wt w) - Q(w) - \langle \nabla Q(w), \wt w - w \rangle\\
		&=Q(\wt w) + H(z) - \langle z, w \rangle 
		-  \langle \nabla Q(w), \wt w - w \rangle \\
        &=Q(\wt w) + H(z) - \langle z, w \rangle 
		+ \langle z - \nabla Q(w), \wt w - w \rangle - \langle z, \wt w - w \rangle \\
        &=Q(\wt w) + H(z) - \langle z, \wt w \rangle 
		+ \langle z - \nabla Q(w), \wt w - w \rangle \\
		&= 
		Q(\wt w) + H(z) - \langle z, \wt w \rangle\,,
	\end{align*}
    where the last equality follows from observing
    that $\langle z - \nabla Q(w), \wt w - w \rangle =0$ 
    since $z - \nabla Q(w) \in \calS$ by 
    Proposition~\ref{prop:first-order-duality} 
    and $\wt w - w \in \calS^\bot$ by Lemma~\ref{lem:normal-diff}. 
	This yields claim (i).

	For claim (ii), we  apply claim (i)
	with  the pairs $(w, z)$ and $(w', z')$,
	and take their difference. 
\end{proof}

\medskip
Finally, we establish the following 
boundary divergence property 
for the Bregman divergence $\D_Q$:

\smallskip
\begin{restatable}
[Boundary blowup of Bregman divergence]
{prop}{propbregmanboundary}
\label{prop:bregman-boundary-blowup}
Fix $\wt w \in \relint(\calW)$.
Then 
$\D_Q(\wt w, w) \to +\infty$
as $w \in \relint(\calW) \to \rbd(\calW)$. 
Moreover, for finite $\alpha \ge 0$, 
$\calC = \{w \in \relint(\calW) : \D_Q(\wt w, w) \le \alpha\} 
\subset \relint(\calW)$
is a compact subset of $\relint(\calW)$.
\end{restatable}

\begin{proof}
	We choose some $r > 0$ 
	such that the relative closed Euclidean ball
	$\calB = \{\wt w +h:h\in\calS^\bot,\ \|h\|_2\le r\}$ 
	lies in $\relint(\calW)$.
	By continuity, $Q$ is bounded above 
	on $\calB$, and thus let $c > 0$ denote such an upper bound. 
	Now for $w \in \relint(\calW)$,
	let $z =  \nabla Q(w) \in \calS^\bot$ 
    (recall here that $\nabla Q(w)$ is defined as the relative gradient, 
    i.e., the gradient of $Q$ along the affine hull of $\calW$,
    and hence by definition $\nabla Q(w) \in \calS^{\bot}$). 
	By Proposition~\ref{prop:first-order-duality}, 
    since $z - \nabla Q(w) = 0 \in \calS$,  we
	have $w = \nabla H(z)$. 

    If $z \neq 0$, since $\wt w +r \frac{z}{\|z\|_2} \in \calB$ as $z \in \calS^\bot$, we have from the definition of 
    $H(z)$ in~\eqref{eq:conjugate} that
    \begin{equation*}
        \textstyle
        H(z) = 
        \max_{v \in \calW} \,\{\langle v,  z \rangle - Q(v)\}
        \geq \big\langle \wt w +r \frac{z}{\|z\|_2},  z \big\rangle 
        - Q\Big(\wt w + r \frac{z}{\|z\|_2}\Big) 
        \ge \langle \wt w, z \rangle + r \|z\|_2 - c\,.
    \end{equation*}
    This same inequality also holds for $z = 0$,
    as again by \eqref{eq:conjugate}, and 
    since $\wt w \in \calW$ and $Q$ is upper bounded by~$c$,
    we have $H(0) \geq - Q(\wt w) \geq -c$, 
	Then, combining part (i) of Proposition~\ref{prop:bregman-pd} 
    and the inequality above, we obtain
	\begin{equation}
		\D_Q(\wt w, w) = 
		Q(\wt w) + H(z) - \langle z, \wt w \rangle \ge Q(\wt w) + r \|z\|_2 - c = Q(\wt w) + r \|\nabla Q(w)\|_2 - c.
		\label{eq:drbu-01}
	\end{equation}
	By the relative steepness condition of
	Assumption~\ref{ass:legendre}, part (ii),
	it follows that the right hand side of~\eqref{eq:drbu-01}
	tends to $+\infty$ as $w \to \rbd(\calW)$. 

    Since $\calW$ is compact and $\calC \subset \relint(\calW) \subset \calW$,
    it is sufficient to prove that $\calC$ is closed in $\calW$ 
    to prove that $\calC$ is also a compact subset of $\relint(\calW)$.
    For this, let $(w_n)$ be a sequence in $\calC$ such that 
    $w_n \to \overline{w} \in \calW$. 
    In what follows, we will prove that $\overline{w} \in \calC$.
    First, we show that $\overline{w} \in \relint(\calW)$. 
    For this, since $\calW = \relint(\calW) \cup \rbd(\calW)$ 
    and the union is disjoint, it suffices to show that 
    $\overline{w} \notin \rbd(\calW)$. 
    Suppose by contradiction that $\overline{w} \in \rbd(\calW)$. 
    Since $w_n \in \calC$, we have $\D_Q(\wt w, w_n) \le \alpha$ 
    for every $n$. But since $w_n \in \relint(\calW)$ and 
    $w_n \to \overline{w} \in \rbd(\calW)$, it follows from the 
    boundary-blowup result proved in the first part of the
    proof that $\D_Q(\wt w, w_n) \to + \infty$, 
    which contradicts the bound $\D_Q(\wt w, w_n) \le \alpha$. 
    Therefore, $\overline{w} \notin \rbd(\calW)$ and hence 
    $\overline{w} \in \relint(\calW)$. Now, since 
    $w_n \to \overline{w}$ and $\D_Q(\wt w, \cdot)$ is continuous
    on $\relint(\calW)$, we have that
    $\D_Q(\wt w, \overline{w}) = \lim_{n \to + \infty} \D_Q(\wt w, w_n)$. 
    Given that $\D_Q(\wt w, w_n) \le \alpha$ for all $n$, it follows
    that $\D_Q(\wt w, \overline{w}) \le \alpha$ 
    Together with $\overline{w} \in \relint(\calW)$, this shows 
    that $\overline w \in \calC$. Thus, the set $\calC$ is a closed subset 
    of the compact set $\mathcal W$, and hence $\calC$ is compact. 
    Since by definition  $\calC \subset \relint(\calW)$, 
    then $\calC$ is a compact subset 
    of $\relint(\calW)$. 
\end{proof}

\medskip
\subsection{Examples of Regularizers}
\label{app:ftrl-prelims:examples}

In this section, we give 
details on instantiations
of regularizers $Q: \calW \to \R$
that satisfy Assumptions~\ref{ass:legendre},
\ref{ass:reg-sc}, and~\ref{ass:lhs}.
In particular, we verify the properties 
for the canonical examples
of negative Shannon entropy~(\ref{app:negative-shannon}),
negative Tsallis entropy~(\ref{app:tsallis})
of order $q \in (0, 1)$,
and for the log-barrier regularizer~(\ref{app:log-barrier}).
Throughout, 
we verify the properties of 
$Q$ and $H$ at a single component,
which suffices due to separability.
For readability, we will write 
(by slight abuse of notation)
$Q$ to denote a single component $Q_i$,
and $H$ for a single component $H_i$,
and we will assume that $\calW = \Delta_D$.

\subsubsection{Negative Shannon Entropy}
\label{app:negative-shannon}

For the negative Shannon entropy regularizer, 
we have that 
\begin{equation*}
    \textstyle 
    Q(w) =  \sum_{i=1}^D w_i \log w_i
    \;\;\text{for all $w \in \calW$}
\end{equation*}
It is well known that the corresponding
$H =Q^*: \R^D \to \R$ and
$\nabla H: \R^D \to \relint(\calW)$
are given by
\begin{equation*} 
    \textstyle
    H(z) = \log(\sum_{i=1}^D \exp(z_i)) \,,
    \;\;
    (\nabla H(z))_i
    = \frac{\exp(z_i)}{\sum_{j=1}^D \exp(z_j)}
    \;\;\;\text{for all $z \in \R^D$}.
\end{equation*}

\paragraph{Relative Legendre and Strong Convexity.}
It is well-known
for negative Shannon entropy that
$Q$ satisfies Assumption~\ref{ass:legendre}
(see, e.g.,~\citet[Section 4.1]{xiao2022convergence}). 
Moreover, by Pinsker's inequality, 
we have that $Q$ is 1-strongly-convex 
with respect to the $\ell_1$ norm $\|\cdot\|_1$,
and thus satisfies Assumption~\ref{ass:reg-sc}.

\smallskip
\paragraph{Local hessian stability.}
The energy $H$ additionally satisfies 
the~\eqref{eq:lhs} property of 
Assumption~\ref{ass:lhs} with $\omega = 2$,
as proven in~\cite[Proposition 4.3]{lazarsfeld2026}.

\subsubsection{Negative Tsallis Entropy}
\label{app:tsallis}

Fix  $q \in (0, 1)$. 
Then the (normalized) negative $q$-Tsallis entropy 
is given by
\begin{equation*}
    \textstyle
    Q_q(w)=\frac{1-\sum_{j=1}^D w_j^q}{q(1-q)}
    \;\;\text{for all $w\in\calW$},
\end{equation*}
with $0^q=0$. The corresponding energy $H_q=Q_q^*$
and gradient map are given by
\begin{equation*}
    \textstyle
    H_q(z)=\lambda_q(z)+\frac1q\sum_{j=1}^D w_j(z)^q
    -\frac1{q(1-q)},
    \;\;\;
    (\nabla H_q(z))_j=w_j(z)
    =\bigl((1-q)(\lambda_q(z)-z_j)\bigr)^{-1/(1-q)},
\end{equation*}
where $\lambda_q(z)>\max_j z_j$ is uniquely determined by
$\sum_j w_j(z)=1$. These expressions follow directly from
the optimality conditions defining $H_q$. The implicit
function theorem implies that $\lambda_q$ and $\nabla H_q$
are smooth on $\R^d$.

\paragraph{Relatively Legendre.}
On $\relint(\calW)$, note that $Q_q$ is smooth. 
Its natural extension to the positive orthant has 
ambient derivatives
\begin{equation*}
    \textstyle
    \partial_j Q_q(w)=-\frac{w_j^{q-1}}{1-q},
    \qquad
    \nabla^2 Q_q(w)=\operatorname{Diag}(w_j^{q-2}).
\end{equation*}
The ambient Hessian is positive definite, so $Q_q$ is strictly
convex on the simplex. As $w$ approaches the 
relative boundary, a smallest
coordinate tends to zero, and its gradient component tends
to $-\infty$. Meanwhile, some coordinate is at least $1/d$,
so its gradient component remains bounded. Thus the
projection of the ambient gradient onto the simplex tangent space
diverges in norm, establishing
the relative steepness condition of Assumption~\ref{ass:legendre}.

\paragraph{Strong Convexity.}
Since $q<1$ and $w_j\leq1$, we have by Cauchy-Schwarz 
\begin{equation*}
    \textstyle
    \langle v,\nabla^2 Q_q(w)v\rangle
    =\sum_{j=1}^D w_j^{q-2}v_j^2
    \geq\sum_{j=1}^D\frac{v_j^2}{w_j}
    \geq\|v\|_1^2
    \qquad
    \text{for all $w\in\relint(\calW)$ and $v\in\R^d$}.
\end{equation*}
Integrating along line segments and using continuity of
$Q_q$ on $\calW$ yields
$D_{Q_q}(w',w)\geq\frac12\|w'-w\|_1^2$ for every
$w'\in\calW$ and $w\in\relint(\calW)$.
Thus, $Q_q$ is $1$-strongly convex with respect to
$\|\cdot\|_1$ and satisfies Assumption~\ref{ass:reg-sc}.

\paragraph{Local Hessian Stability.}
The energy $H_q$ satisfies the~\eqref{eq:lhs} property of
Assumption~\ref{ass:lhs} with 
parameter $\omega_q=2(2-q)$.
To see this, set $a_j(z)=w_j(z)^{2-q}$ and
$A(z)=\sum_j a_j(z)$. Differentiating the optimality
conditions gives
\begin{equation*}
    \textstyle
    \nabla^2 H_q(z)
    =\operatorname{Diag}(a(z))
    -\frac{a(z)a(z)^\top}{A(z)}.
\end{equation*}
Fix $h=z'-z$ and $v\in\R^d$. Along $z_s=z+sh$
for $s \in [0, 1]$,
write $w_j=w_j(z_s)$, $a_j=a_j(z_s)$, and let
$\overline h_s=\sum_j a_jh_j/A$ and
$\overline v_s=\sum_j a_jv_j/A$.
Then $\dot w_j=a_j(h_j-\overline h_s)$, and
$f_v(s)=\langle v,\nabla^2H_q(z_s)v\rangle$ satisfies
\begin{equation*}
    \textstyle
    f_v(s)=\sum_j a_j(v_j-\overline v_s)^2,
    \qquad
    f_v'(s)=(2-q)\sum_j
    a_jw_j^{1-q}(h_j-\overline h_s)(v_j-\overline v_s)^2.
\end{equation*}
Since $w_j^{1-q}\leq1$ and
$|h_j-\overline h_s|\leq2\|h\|_\infty$, we obtain
$|f_v'(s)|\le 2(2-q) \|h\|_\infty f_v(s)$.
Applying Gr\"onwall's inequality and using 
the fact that $v$ was arbitrary then yields
\begin{equation*}
    \textstyle
    \exp(-\omega_q\|z-z'\|_\infty) \cdot \nabla^2H_q(z')
    \preceq\nabla^2H_q(z)
    \preceq 
    \exp(\omega_q\|z-z'\|_\infty)  \cdot \nabla^2H_q(z')\;,
\end{equation*}
where $\omega_{q} = 2(2-q) \ge 1$.

\subsubsection{Log-Barrier Regularizer}
\label{app:log-barrier}

Let $\calW=\Delta_D$. The log-barrier regularizer is given by
\begin{equation*}
    \textstyle
    Q_{\logb}(w)=-\sum_{j=1}^D\log w_j
    \;\;\text{for all $w\in\relint(\calW)$},
\end{equation*}
and $Q_{\logb}(w)=+\infty$ on the relative boundary.
Here, the regularizer assumptions are understood with this
extended-valued convention.
The corresponding energy $H_{\logb}=Q_{\logb}^*$
and gradient map are:
\begin{equation*}
    \textstyle
    H_{\logb}(z)
    =\lambda(z)-D-\sum_{j=1}^D\log(\lambda(z)-z_j),
    \;\;\;
    (\nabla H_{\logb}(z))_j=w_j(z)
    =\frac1{\lambda(z)-z_j},
\end{equation*}
where $\lambda(z)>\max_j z_j$ is uniquely determined by
$\sum_j w_j(z)=1$. These expressions follow directly from
the optimality conditions defining $H_{\logb}$.
The implicit function theorem implies that $\lambda$ and
$\nabla H_{\logb}$ are smooth on~$\R^d$.

\paragraph{Relatively Legendre.}
On $\relint(\calW)$, the regularizer is smooth, and its
natural extension to the positive orthant has ambient derivatives
\begin{equation*}
    \textstyle
    \partial_j Q_{\logb}(w)=-\frac1{w_j},
    \qquad
    \nabla^2 Q_{\logb}(w)
    =\operatorname{Diag}(w_j^{-2}).
\end{equation*}
The ambient Hessian is positive definite, so $Q_{\logb}$ is
strictly convex on the simplex. As $w$ approaches the relative boundary,
a smallest coordinate tends to zero, and its gradient
component tends to $-\infty$. Meanwhile, some coordinate
is at least $1/d$, so its gradient component remains bounded.
Thus the projection of the ambient gradient onto the simplex tangent
space diverges, establishing Assumption~\ref{ass:legendre}.

\paragraph{Strong Convexity.}
Since $w_j\leq1$, we have by Cauchy--Schwarz
\begin{equation*}
    \textstyle
    \langle v,\nabla^2 Q_{\logb}(w)v\rangle
    =\sum_{j=1}^D\frac{v_j^2}{w_j^2}
    \geq\sum_{j=1}^D\frac{v_j^2}{w_j}
    \geq\|v\|_1^2
    \qquad
    \text{for all $w\in\relint(\calW)$ and $v\in\R^d$}.
\end{equation*}
Integrating along line segments yields
$D_{Q_{\logb}}(w',w)\geq\frac12\|w'-w\|_1^2$
for every $w,w'\in\relint(\calW)$.
For $w'$ on the relative boundary, the inequality holds
in the extended-valued sense because
$D_{Q_{\logb}}(w',w)=+\infty$.
Thus, $Q_{\logb}$ is $1$-strongly convex with respect
to $\|\cdot\|_1$ and satisfies Assumption~\ref{ass:reg-sc}.

\paragraph{Local Hessian Stability.}
The energy $H_{\logb}$ satisfies the~\eqref{eq:lhs}
property of Assumption~\ref{ass:lhs} with parameter
$\omega_{\logb}=4$.
To see this, set $a_j(z)=w_j(z)^2$ and
$A(z)=\sum_j a_j(z)$. Differentiating the optimality
conditions gives
\begin{equation*}
    \textstyle
    \nabla^2 H_{\logb}(z)
    =\operatorname{Diag}(a(z))
    -\frac{a(z)a(z)^\top}{A(z)}.
\end{equation*}
Fix $h=z'-z$ and $v\in\R^d$. Along $z_s=z+sh$ for
$s\in[0,1]$, write $w_j=w_j(z_s)$, $a_j=a_j(z_s)$,
$A=A(z_s)$, and let
$\overline h_s=\sum_j a_jh_j/A$ and
$\overline v_s=\sum_j a_jv_j/A$.
Then $\dot w_j=a_j(h_j-\overline h_s)$, and
$f_v(s)=\langle v,\nabla^2H_{\logb}(z_s)v\rangle$
satisfies
\begin{equation*}
    \textstyle
    f_v(s)=\sum_j a_j(v_j-\overline v_s)^2,
    \qquad
    f_v'(s)=2\sum_j
    a_jw_j(h_j-\overline h_s)(v_j-\overline v_s)^2.
\end{equation*}
Since $w_j\leq1$ and
$|h_j-\overline h_s|\leq2\|h\|_\infty$, we obtain
$|f_v'(s)|\leq4\|h\|_\infty f_v(s)$.
Applying Gr\"onwall's inequality and using 
the fact that $v$ was arbitrary then yields
\begin{equation*}
    \textstyle
    \exp(-\omega_{\logb}\|z-z'\|_\infty)
    \cdot\nabla^2H_{\logb}(z')
    \preceq\nabla^2H_{\logb}(z)
    \preceq
    \exp(\omega_{\logb}\|z-z'\|_\infty)
    \cdot\nabla^2H_{\logb}(z')\;,
\end{equation*}
where $\omega_{\logb}=4\geq1$.


\medskip
\section{Details on 
Reduction to Symmetrized System}
\label{app:reduction}

This section provides more details
on the $M$-symmetrized system that
was introduced in Section~\ref{sec:main-proof-overview},
and that is studied to prove our main results in
Theorem~\ref{thm:main-last-iterate} and 
Theorem~\ref{thm:ftrl-diverge-main}.
In particular,
the proofs of Proposition~\ref{prop:symmetrized-meq}
and~\ref{prop:reduction-algorithmic}
are given below
in Sections~\ref{app:reduction:structural}
and~\ref{app:reduction:algorithmic}, respectively.
We further summarize all properties
of the symmetrized system
in Setting~\ref{setting:weighted},
which will be referenced in the proofs
of our later results.

\paragraph{Overview of underlying setting.}
Throughout this section,
we assume an operator
$G: \calW \to \R^D$
such that~$\WMEQ_M(G) \cap \relint(\calW)\neq \emptyset$
for a weights matrix
$M = \Diag(m_{1} I_{d_1}, \dots, m_N I_{d_N}) \in \R^{D \times D}$,
where each $m_i > 0$.
We assume that $G$ satisfies Assumption~\ref{ass:operator}
with Lipschitz constant $L \ge 1$
and boundedness constant $B \ge 1$.
We also assume
that $Q: \calW \to \R$
is a regularizer
satisfying Assumption~\ref{ass:legendre}
and Assumption~\ref{ass:reg-sc},
and its corresponding
energy function $H = Q^*:\R^D \to \R$
satisfies Assumption~\ref{ass:lhs}.
We further assume that
Assumption~\ref{ass:restricted-regularity} holds.

\subsection{Structural Components and Properties}
\label{app:reduction:structural}

\paragraph{Scaled Primal and Dual Spaces.}
Recall that we define the scaled primal domain 
$\calX := M^{1/2} \calW \subseteq \R^D$.
For each $w \in \calW$ and $z \in \R^D$,
we associate $x = M^{1/2} w$ and $y = M^{1/2} z$,
from which it follows that
$w = M^{-1/2} x$ and $z = M^{-1/2} y$.
Recalling that $\calW = \prod_{i=1}^N \calW_i$,
where each $\calW_i = \Delta_{d_i}$,
observe then that 
$\calX = \prod_{i=1}^N \calX_i$
where each $\calX_i = \sqrt{m_i} \Delta_{d_i}$. 
Recalling the definitions of  
$\calS = \prod_{i=1}^N \Span(\1_{d_i}) \subseteq \R^D$
and its orthogonal complement $\calS^\bot \subseteq \R^D$
from~\eqref{eq:calS-def}, we have
the following result, 
the proof of which is similar to that of 
Lemma~\ref{lem:normal-diff}:

\smallskip
\begin{restatable}
	{lem}{lemweightednormaldiff}
	\label{lem:normal-diff-weighted}
    For any $x,  x' \in \calX$, $x-x' \in \calS^\bot$.  
\end{restatable}

In words, as each component
of $\calX_i$ is a positive constant scaling 
of $\calW_i$, 
$\calS$ and $\calS^\bot$ remain 
the normal and tangent spaces of $\calX$,
respectively.

\smallskip
\paragraph{Symmetrized operator.}
Recall that the $M$-symmetrized operator
$V: \calX \to \R^D$ is given by
\begin{equation}
	V(x) = M^{1/2} G(M^{-1/2} x) 
	\;\;\;\text{for all $x \in \calX$} \;.
	\label{eq:V-app}
\end{equation}
By a slight abuse of notation,
we let $\NE(V)\subset \calX$ denote the set
of points $x' \in \calX$ satisfying
\begin{equation}
	\langle V(x'), x-x' \rangle \ge 0
	\;\;\;\text{for all $x \in \calX$} \;.
	\label{eq:V-NE}
	\tag{NE}
\end{equation}

As introduced in Section~\ref{sec:main-proof-overview},
the key structural property of 
$V$ is that, for all 
$w^\star \in \WMEQ_M(G)$,
the scaled variable $x^\star = M^{1/2}w^\star \in \calX$
satisfies the \textit{unweighted}
Minty variational equality with respect to $V$:

\smallskip
\propmeqsymmetrized*

\begin{proof}
    Let $w^\star \in \WMEQ_M(G)$ and $x \in \calX$.  
    Define then $x^\star = M^{1/2} w^\star \in \calX$ 
    and $w = M^{-1/2} x$. By definition of~$V$ in  \eqref{eq:V-app},
    we have
    \begin{align*}
    \langle V(x), x-x^\star \rangle
    &= \langle M^{1/2} G(M^{-1/2}x), x-x^\star \rangle \\
    &= \langle M^{-1/2} M G(M^{-1/2}x), x-x^\star \rangle\\
    &= \langle MG(M^{-1/2}x), M^{-1/2}(x-x^\star)\rangle 
    = 
    \langle MG(w), w-w^\star \rangle 
    = 0\,, 
    \end{align*}
    where the third equality uses the symmetry of $M$, 
    and the last equality follows from $w^\star \in \WMEQ_M(G)$.
\end{proof}

\smallskip
Recall that the set of all solutions $x^\star$ 
of~\eqref{eq:minty-equality} is denoted as: 
\begin{equation}
\MEQ(V) = \{x^\star \in \calX : \langle V(x), x-x^\star \rangle = 0
	\;\text{for all $x \in \calX$}
\}
\subseteq \calX\,.
\end{equation}
 
We also have as an immediate corollary of 
Proposition~\ref{prop:symmetrized-meq} and \eqref{eq:V-NE} the
following relationships:

\begin{restatable}{cor}{corrweightedmeqne}
	\label{cor:meq-ne-weighted}
	The following hold:
	(i)
	$w^\star \in \WMEQ_M(G)$
	if and only if $M^{1/2} w^\star \in \MEQ(V)$
	and (ii)
	$w^\star \in \NE(G)$
	if and only if $M^{1/2} w^\star \in \NE(V)$.
\end{restatable}

\paragraph{Lipschitzness and boundedness of V.}
Under Assumption~\ref{ass:operator} for the operator $G$,
it follows that the operator $V$ is Lipschitz and
bounded, with parameters 
depending on $L, B$ and $m_{\max}, m_{\min}$.
To state this more concretely, recall that we let 
$\kappa_{M} := m_{\max}/ m_{\min}$.
Then, we have the following properties of $V$:

\begin{restatable}{prop}{propVlipbounded}	
	\label{prop:V-lipschitz-bounded}
	The operator $V$ is $\wh L$-Lipschitz
	and $\wh B$ bounded,
	for $\wh L = \sqrt{\kappa_M} L$
	and $\wh B = \sqrt{m_{\max}}B$.
\end{restatable}

\begin{proof}
First, using the block-diagonal structure of
$M = \Diag(m_1 I_{d_1}, \dots, m_N I_{d_N})$, 
observe that 
$\|M^{1/2} v\|_{*} \le \sqrt{m_{\max}}\|v\|_{*}$
and 
$\|M^{-1/2} v\| \le \tfrac{1}{\sqrt{m_{\min}}} \|v\|$
for all $v \in \R^D$.
Then by definition of $V(x) = M^{1/2} G(M^{-1/2}x)$
and using the $L$-Lipschitzness of 
$G$ from Assumption~\ref{ass:operator},
we have for all $x, x' \in \calX$:
\begin{equation*}
	\textstyle
	\|V(x')-V(x)\|_{*}
	= 
	\|M^{1/2}(G(M^{-1/2}x') - G(M^{-1/2}x))\|_{*}
	\le 
	\frac{\sqrt{m_{\max}}}{\sqrt{m_{\min}}} L \|x-x'\| \;.
\end{equation*}
Similarly, by $B$-boundedness of $G$,
we have for all $x \in \calX,$ 
\begin{equation*}
	\|V(x)\|_{*} 
	= 
	\|M^{1/2} G(M^{-1/2}x)\|_{*}
	\le 
	\sqrt{m_{\max}} B \;,
\end{equation*}
which concludes the proof.
\end{proof}

\smallskip
\subsection{Algorithmic Components and Properties}
\label{app:reduction:algorithmic}

We now expand on the
algorithmic components of the symmetrized
system. 
In particular, we develop
the proof of Proposition~\ref{prop:reduction-algorithmic}, 
restated here:

\smallskip
\propsymmetrizeddynamics*

We will prove the proposition
for the regularizer $R: \calX \to \R$
defined as follows:
first, recall from Section~\ref{sec:main-results}
that we defined the $M$-symmetrized
regularizer $Q_M : \calW \to \R$ by
\begin{equation}
	\textstyle
	Q_M(w)
	= 
	\sum_{i=1}^N 
	m_i \cdot \big( \frac{\eta}{\eta_i} Q_i(w_i) \big)
	\;\;\;\text{for all $w \in \calW$}\;.
	\label{eq:QM-def}
\end{equation}
Then we define $R$ as the separable regularizer given by 
\begin{equation}
	\textstyle
	R(x) := Q_M(M^{-1/2} x)
	= 
	\sum_{i=1}^N m_i \cdot 
	\big( \frac{\eta}{\eta_i} Q_i(\frac{x_i}{\sqrt{m_i}}) \big)
	\;\;\text{for all $x \in \calX$} \;.
	\label{eq:R-def}
\end{equation}
Recall that $F=R^*: \R^D \to \R$ is then the convex
conjugate of $R$.
Importantly, recall that $Q$ satisfies Assumption
\ref{ass:legendre} with respect to $\calW$.
Then, due to the positive weights in $M$
and the invertibility of the 
change of coordinates  defining
 $\calX = M^{1/2}\calW$,
it follows that $R$ 
is twice continuously
differentiable on $\relint(\calX)$,
and that the relative steepness
of $\nabla R$ also holds
with respect to $\rbd(\calX)$.
Thus, $R$ satisfies the
properties of Assumption~\ref{ass:legendre} with
respect to $\calX$,
and the primal-dual relationships
of Section~\ref{app:ftrl-prelims:primal-dual}
then identically apply to the pair $R$ and $F$.

We then proceed with the proof
of Proposition~\ref{prop:reduction-algorithmic}.

\begin{proof}
	The proof proceeds in two steps: 
	first, we establish a
	relationship between the
	original and symmetrized
	regularizer and energy gradients.
	Then, we apply the relationship
	to the iterates of OFTRL and EGFTRL
	under the original system,
	and we show that the update
	structure is exactly preserved for
	the symmetrized system.

	\noindent
	\textit{1. Relationship between gradient maps.}
	First, note that by separability 
	of $Q_M$ (and 
	using the definition of $Q$ 
	from~\eqref{eq:global-reg}),
	we have
	that $\nabla Q_M(w) = M \nabla Q(w)$
	for all $w \in\relint(\calW)$.
	Then for $x \in \relint(\calX)$,
	let $w = M^{-1/2} x \in \calW$.
	Using the fact that 
	$R(x) = \Q_M(M^{-1/2} x)$,
	it then follows by the chain rule 
	and the symmetry of $M$ that
	\begin{equation}
		\nabla R(x) 
		= 
		(M^{-1/2})^\top 
		\nabla Q_M(M^{-1/2}x)
		= 
		M^{-1/2} M \nabla Q(w)
		= M^{1/2} \nabla Q(w)
		\;.
		\label{eq:grad-R-property}
	\end{equation}
	Now, for fixed $y \in \R^D$,
	set $z = M^{-1/2} y \in \R^D$.
	Then, let $w = \nabla H(z) \in \relint(\calW)$,
	and $x = M^{1/2} w \in \relint(\calX)$.
	We will show that
	$x = \nabla F(y)$.
	For this, first recall
	that due to Proposition~\ref{prop:first-order-duality}
	applied to the pair $(w, z)$,
	we have that
	$z = \nabla Q(w) + s$
	for some $s \in \calS$.
	It follows using~\eqref{eq:grad-R-property}
	that 
	\begin{equation*}
		M^{1/2}z = M^{1/2} \nabla Q(w) + M^{1/2}s\,,\quad  
		y = \nabla R(x) + s'\,, 
		\label{eq:}
	\end{equation*}
	where $s' = M^{1/2}s \in\calS$.
	However, by 
	Proposition~\ref{prop:first-order-duality}
	applied with $R$ and $F$,
	this implies the identity 
	$x = \nabla F(y)\in \relint(\calX)$.
	Thus for every 
	$z \in \R^D$ and 
	$w = \nabla H(z) \in \relint(\calW)$, 
	it holds for 
	$y= M^{1/2}z \in \R^D$
	and $x = M^{1/2}w \in \relint(\calX)$
	that 
	\begin{equation}
		\nabla F(y)
		=
		x
		= 
		M^{1/2} w 
		= M^{1/2} \nabla H(M^{-1/2} y) \;.
		\label{eq:dual-identity-weighted}
	\end{equation}

	\noindent
	\textit{2. Preservation of FTRL-type update structure.}
	Using the relationship of~\eqref{eq:dual-identity-weighted},
	we then have the following equivalence 
	in the primal and dual update structure of 
	FTRL-type dynamics for the symmetrized system.
	For generality, let $z', z \in \R^D$, 
	$w' = \nabla H(z') \in \relint(\calW)$
	and $w = \nabla H(z) \in \relint(\calW)$.
	Let $y' = M^{1/2}z'$ and $y = M^{1/2}z$,
	and let $x'= M^{1/2}w'\in \relint(\calX)$ and
	$x = M^{1/2}w \in \relint(\calX)$.
	Suppose further that for some $\eta > 0$,
	we have $z' = z - \eta G(w)$.
	Then multiplying by $M^{1/2}$, it follows that:
	\begin{equation}
		y' = M^{1/2} z'
		= 
		M^{1/2} z - \eta M^{1/2} G(M^{-1/2} x) 
		= 
		y - \eta V(x) \;,
		\label{eq:dual-structure-reduction}
	\end{equation}
	where we use the definition of $V$
	from~\eqref{eq:V-def}.
	As $z'$ and $w'$ satisfy
	$w' = \nabla H(z')$,
	then applying~\eqref{eq:dual-identity-weighted}
	to $y'$, we further obtain 
	$x' = \nabla F(y') \in \relint(\calX)$.

	For the iterates $\{(w_t, z_t)\}$
	of~\eqref{eq:oftrl} and~\eqref{eq:egftrl}
	under the original system, 
	this immediately implies that, 
	for $x_t = M^{1/2} w_t$
	and $y_t = M^{1/2} z_t$, 
	we have the following update structure
	in the symmetrized coordinates
	for all $t \ge 0$:
	\begin{equation}
		\begin{aligned}
			\text{for OFTRL}:&
			\quad
			y_{t+1} 
			= y_t - 2 \eta V(x_t) + \eta V(x_{t-1})
			\,,\; x_{t+1} = \nabla F(y_{t+1})\,, \\[0.25em]
			\text{for EGFTRL}:&
			\quad
			\begin{cases}
				y_{t+1/2} = y_t - \eta V(x_t) \, \,,\;\;
				x_{t+1/2}	 = \nabla F(y_{t+1/2})\,, \\
				y_{t+1} = y_t - \eta V(x_{t+1/2}) \,,\;\; 
				x_{t+1}	 = \nabla F(y_{t+1}) \;.
			\end{cases}
		\end{aligned}
		\label{eq:sym-dynamics-0}
	\end{equation}	
	This completes the proof of the proposition.
\end{proof}

\medskip
\paragraph{Local Hessian stability.}
As the original energy function $H = Q^*$
satisfies the Assumption~\ref{ass:lhs} with
constant $\omega > 0$, 
we further have a similar property for the
symmetrized energy function:

\begin{restatable}
	[LHS property for $F$]
	{prop}{proplhssym}
	\label{prop:lhs-sym}
	For every $y, y' \in \R^D$, we have
	\begin{equation}
		\exp(-\wh \omega  \|y-y'\|_{*})\cdot
		\nabla^2 F(y') 
		\preceq
		 \nabla^2 F(y)
		\preceq
		 \exp(\wh \omega \|y-y'\|_{*}) \cdot
		\nabla^2 F(y') 
		\;.
		\label{eq:lhs-sym}
		\tag{LHS}
	\end{equation}
	where $\wh \omega=\omega/\sqrt{m_{\min}}$.
\end{restatable}

\begin{proof}
	First, recall from~\eqref{eq:dual-identity-weighted}
	that for $y \in \R^{D}$ and $z = M^{-1/2} y$, we have
	$\nabla F(y) = M^{1/2} \nabla H(M^{-1/2} y)$.
	Differentiating then yields
	\begin{equation*}
		\nabla^2 F(y)
		= 
		M^{1/2} \nabla^2 H(M^{-1/2}y)M^{-1/2}
		= 
		\nabla^2 H(M^{-1/2} y) 
        =\nabla^2 H(z)\;.
	\end{equation*}
	Here, the final equality holds due
	to the separability of $H = Q^*$,
	which means that
	the block diagonal Hessian $\nabla^2 H$ 
	commutes with the block-diagonal matrix $M$. 
	Now fix also $y' \in \R^D$ 
	and let $z' =  M^{-1/2} y'$. 
	By applying Assumption~\ref{ass:lhs} for $H$, 
	we then have for a constant $ \omega \ge 1$ that
	\begin{equation}
		\exp(-\omega  \|z-z'\|_{*})\cdot
		\nabla^2 H(z') 
		\preceq
		 \nabla^2 H(z)
		\preceq
		 \exp( \omega \|z-z'\|_{*}) \cdot
		\nabla^2 H(z')  \;.
		\label{eq:lhs-0}
	\end{equation}
	Moreover, observe that 
	$\omega \|z-z'\|_* = \omega 
	\|M^{-1/2}(y-y')\|_{*} \leq \frac{\omega}{\sqrt{m_{\min}}} \|y - y'\|_{*}
	$.
	Substituting this inequality into~\eqref{eq:lhs-0}
	then yields the statement of the proposition. 
\end{proof}

\paragraph{Proof of Proposition~\ref{prop:bregman-shifted-energy}}
For the energy function $F = R^*$, 
recall that we define for any $x^\star \in \relint(\calX)$
the \textit{shifted energy function} $F_{x^\star}: \R^D \to \R$
by 
\begin{equation}
	F_{x^\star}(y)
	= 
	F(y) - \langle x^\star, y \rangle
	\;\;\;\text{for all $y \in\R^D$} \;.
	\label{eq:shifted-energy-app}
\end{equation}
We then prove the following relationship between
the Bregman divergence $\D_R$
and $F_{x^\star}$,
as introduced in Section~\ref{sec:proof-overview:energy}:

\smallskip
\propshiftedenergy*

\begin{proof}
    As $F$ and $R$ satisfy Assumption~\ref{ass:legendre}, we can apply Proposition~\ref{prop:bregman-pd} (ii)
    for the pairs $(x', y')$ and
    $(x,y)$, and use the definition of $F_{x^\star}$
    from~\eqref{eq:shifted-energy-app}.
\end{proof}

\medskip
\paragraph{Regularity condition under
symmetrized system.}
Finally, when Assumption~\ref{ass:restricted-regularity}
holds for the original system
using $G$, $Q$, and $H$,
then the analogous condition
also holds for the symmetrized system:

\begin{restatable}
	[Regularity under symmetrized system.]
	{prop}{propregularitysym}
	\label{prop:regularity-sym}
	Fix $z_0\in \R^D$ and
	let $w_0=\nabla H(z_0)$, and let
	$y_0=M^{1/2}z_0$ and $x_0=M^{1/2}w_0$. 
	If Assumption~\ref{ass:restricted-regularity}
	holds for the original system with
	$(G,Q, H)$ at $w_0$,
	then it also holds for the symmetrized system at 
	$x_0$ with $(V, R,F)$.
\end{restatable}

\begin{proof}
	First, we extend the definitions 
	of the effective dual space and the map $\pi_0$
	introduced in~\ref{sec:ftrl:geometry}
	to the symmetrized system.
	In particular, to distinguish between the 
	original and symmetrized systems, 
	we define
	$\calU_G := \Pi_{\calS\bot}\Span(G(\calW))$
	and 
	$\calU_V := \Pi_{\calS\bot}\Span(V(\calX))$.
	Since we have $V(M^{1/2}w)=M^{1/2}G(w)$,
	and because 
	$M^{1/2}$ commutes with 
	the projection $\Pi_{\calS^\bot}$, 
	we have that $\calU_V = M^{1/2}\calU_G$.
	Thus $M^{1/2}$ restricts to a linear isomorphism from
	$\mathcal U_G$ onto $\mathcal U_V$.

	Now, we further distinguish between the 
	two respective restricted stationarity maps, 
	where we let 
	$\pi^G_0(v) := \Pi_{\calS^\bot}G(\nabla H(z_0+v))$
	for $v \in \calU_G$
	and 
	$\pi^V_0(u) := \Pi_{\calS^\bot} V(\nabla F(y_0 +u))$
	for $u \in\calU_V$.
	Using the relationship $y_0=M^{1/2}z_0$
	and the identity of~\eqref{eq:dual-identity-weighted},
	we then have 
	\begin{equation}
	 	\nabla F(y_0+ u) = M^{1/2}\nabla H(z_0 + M^{-1/2}u) \;.
	\end{equation}
	It then follows that
	\begin{equation}
		\pi^V_0(u)
		=
		\Pi_{\calS^\bot}
		M^{1/2}G(\nabla H(z_0+M^{-1/2}u))\\
		=
		M^{1/2}\pi_0^G(M^{-1/2}u),
	\end{equation}
	which implies that 
	$u$ is a zero of $\pi^V_0$
	if and only if $v = M^{-1/2} u$ is a 
	zero of $\pi^G_0$.
	Thus, differentiating on the effective dual spaces 
	gives the relationship
	$J_{\pi^V_0}(u) = M^{1/2} J_{\pi^G_0}(v) M^{-1/2}$.
	Now at every such zero, recall that 
	Assumption~\ref{ass:restricted-regularity} 
	implies that 
	the map $J_{\pi^G_0}(v):\calU_G \to \calU_G$
	is nonsingular.
	Thus, since all $m_i > 0$ and
    $M^{1/2}:\calU_G\to\calU_V$ is an
	isomorphism, it follows that
	$J_{\pi^V_0}(u):\calU_V  \to \calU_V$
	is also nonsingular.

    Finally, write $\calR_G(w_0)$ and
    $\calR_V(x_0)$ 
    to denote the respective reachable primal sets
    of the two systems. 
    Then by Corollary~\ref{cor:meq-ne-weighted},
    it follows immediately that 
    $\calR_G(w_0) \not \subseteq \NE(G)$
    if and only if 
    $\calR_V(x_0) \not \subseteq \NE(V)$.
    Thus, when Assumption~\ref{ass:restricted-regularity}
    holds for the original system,
    then the two conditions
    similarly hold for the symmetrized system.
\end{proof}

\subsection{Summary of Symmetrized Setting}
\label{subsec:summary-sym-setting}

Here, we give a summary of the key properties
of the symmetrized system
that will be used in the remainder
of the analysis:

\begin{restatable}
    [Summary of symmetrized system]
    {setting}{settingsymmetrized}
    \label{setting:weighted}
    Assume the operator
    $G : \calW \to \R^D$ 
    satisfies
    $\WMEQ_M(G)\cap \relint(\calW) \neq \emptyset$
    for a weights matrix $M$,
    and that $G$
    satisfies Assumption~\ref{ass:operator}.
    Let $Q: \calW \to \R$
    and $H = Q^* : \R^D \to \R$
    satisfy Assumptions~\ref{ass:legendre},
    \ref{ass:reg-sc}, and~\ref{ass:lhs}.
    Further let Assumption~\ref{ass:restricted-regularity} hold.

    Then the following properties hold
    for the symmetrized system:
    \begin{itemize}
    	\item
    	$V: \calX \to \R^D$ 
    	is $\wh L$-Lipschitz
    	and $\wh B$-bounded
    	with $\wh L = \sqrt{\kappa_M} L$
    	and $\wh B = \sqrt{m_{\max}}B$.

    	\item
    	$w^\star \in \WMEQ_M(G)$ 
    	if and only if $x^\star \in \MEQ(V)$.

    	\item
    	$R: \calX \to \R$ and $F  =R^*: \R^D \to \R$
    	satisfy Assumption~\ref{ass:legendre}
    	with respect to $\relint(\calX)$.

    	\item
    	$F$ satisfies~\eqref{eq:lhs-sym}
    	with parameter $\wh \omega = 
    	\omega/\sqrt{m_{\min}}$.

    	\item
    	For the iterates $\{(w_t,z_t)\}$
    	of~\eqref{eq:ftrl},~\eqref{eq:iftrl},
    	\eqref{eq:oftrl}, or~\eqref{eq:egftrl}
    	with stepsize $\eta > 0$,
    	letting $x_t = M^{1/2}w_t$
    	and $y_t = M^{1/2}z_t$ for all $t \ge 0$,
    	it holds correspondingly that:
    	\begin{align}
			y_{t+1}
			&= y_t - \eta \, V(x_t) 
			\label{eq:ftrl-sym}
			\tag{FTRL}\\
			y_{t+1}
			&= y_t - \eta \, V(x_{t+1})
			\label{eq:iftrl-sym}
			\tag{IFTRL} \\
			y_{t+1} 
			&= y_t - 2\eta \, V(x_t) + \eta \, V(x_{t-1})  
			\label{eq:oftrl-sym}
			\tag{OFTRL} 
			\\
			y_{t+1/2}
			&= 
			y_{t} - \eta \, V(x_t) 
			\,,\quad
			y_{t+1} 
			= y_t - \eta \, V(x_{t+1/2}) \;.
			\label{eq:egftrl-sym}
	    \tag{EGFTRL} 
    	\end{align}
    	Moreover, for each variant,
    	it holds for all $t \ge 0$ that
    	$\|y_{t+1} - y_t\|_{*} \le 3\eta \hat B$.
    	\item
    	Under Assumption~\ref{ass:restricted-regularity}
    	for the original system with
    	$(G, Q, H)$
    	the corresponding
    	properties from the assumption
    	also hold for 
    	the symmetrized system
    	with $(V, R, F)$.
    \end{itemize}
    Given the equivalence between the original
    and symmetrized systems, 
    throughout, 
    we will freely apply preliminaries proved
    for the original system
    (e.g., from Section~\ref{app:ftrl-prelims:primal-dual}).
\end{restatable}


\section{Details on Effective Dual Space and Reachable Primal Sets}
\label{app:geometry-primal-dual}

In this section,
we give more details
on properties on the
reachable primal set,
effective dual space, and 
the stationarity map 
introduced in Section~\ref{sec:ftrl:geometry}.
Specifically, given the reduction
to the symmetrized system
introduced in
Section~\ref{app:reduction}
(as summarized in Section~\ref{subsec:summary-sym-setting}, Setting~\ref{setting:weighted}),
we consider here the analogous
objects for the symmetrized system.
We begin by recalling these definitions:

\paragraph{Recap of Definitions.}
Let $\calY = \Span(V(\calX))$
and $\calU = \Pi_{\calS^\bot}\calY \subseteq \calS^\bot$
where $\calS^\bot$ is defined in \eqref{eq:calS-def}.
For a fixed primal and dual initialization
$x_0 = \nabla F(y_0) \in \relint(\calX)$, 
recall that we define
\begin{equation*}
	\calR(x_0) 
	= 
	\{\nabla F(y_0 + y) \;:\; y \in \calY\}  \;.
\end{equation*}
The effective stationarity map 
$\pi_0: \calU \to \calU$ is given by
\begin{equation}
\label{eq:recall-pi0-def}
	\pi_0(u) 
	= 
	\Pi_{\calS^\bot} V(\nabla F(y_0 + u)) 
	\;\;\text{for all $u \in \calU$} \;.
\end{equation}

\paragraph{Key 
property of effective
stationarity map.} 
To begin, we note the following
important relationship 
between zeros of $\pi_0$ 
and the set $\NE(V)$:

\begin{restatable}
    {prop}{proppi0zerosnash}
    \label{prop:pi0-zeros-ne}
    Fix $y_0 \in\R^D$ and 
    $x_0 = \nabla F(y_0) \in \relint(\calX)$
    and consider the map
    $\pi_0: \calU \to \calU$ as defined in \eqref{eq:recall-pi0-def}. 
    Then for all $u \in \calU$,
    it holds that 
    $\pi_0(u) = 0$
    if and only if 
    $\nabla F(y_0 + u) \in \NE(V)$.
\end{restatable}

\begin{proof}
    Fix $u \in \calU$,
    and denote as shorthand $x_0(u) := \nabla F(y_0 + u)
    \in \relint(\calX)$. 
    For the forward direction
    suppose $\pi_0(u) = 0$. 
    Then by definition,
    we must have $V(x_0(u))\in \calS$.
    By Proposition~\ref{prop:ne-vs-normal} (i)
    (applied to the symmetrized system),
    this implies that $x_0(u) \in \NE(V)$.
    For the other direction,
    suppose $x_0(u) \in \NE(V)$. 
    Then since also 
    $x_0(u) \in \relint(\calX)$,
    it follows from 
    Proposition~\ref{prop:ne-vs-normal} (ii)
    that $V(x_0(u)) \in\calS$.
    Thus by definition,
    $\pi_0(u) = \Pi_{\calS^\bot} V(x_0(u))
    = 0$.
\end{proof}

\medskip
Next, we give in 
Section~\ref{app:reachable-props}
several properties
of the reachable primal  
sets $\calR(x_0)$.
Then, in 
Section~\ref{app:regularity},
we derive several 
consequences of
the regularity
property of Assumption~\ref{ass:restricted-regularity}.
These properties and consequences
are used to later establish our equilibrium
selection result 
and operator domination inequality.

\subsection{Properties
of Reachable Primal Set}
\label{app:reachable-props}

\medskip
\paragraph{Unique 
reachable MEQ point.} 

We show that, among all MEQ points, at most one is reachable by 
the dynamics from any given initialization. 
Note that this property holds without 
Assumption~\ref{ass:restricted-regularity}. 

\smallskip
\begin{restatable}[Unique reachable MEQ point]
{prop}{propmeqreachable set}
	\label{prop:unique-reachable-meq}
	Fix $y_0 \in \R^D$ and
	$x_0 = \nabla F(y_0) \in \relint(\calX)$.
	Then
	\begin{equation*}
	\big|\MEQ(V) \cap \calR(x_0)\big| \le 1 \;.
	\end{equation*}
\end{restatable}

\noindent 
Intuitively, for two reachable MEQ points, 
the MEQ condition forces their difference to be orthogonal 
to the corresponding reachable dual displacement, 
whereas strict monotonicity of the gradient map 
forces their inner product to be strictly positive 
whenever the displacement is nonzero. Hence the displacement 
must vanish, and the two points must coincide. 
We now provide a formal proof. 

\begin{proof}
	Let $x^\star_1, x^\star_2 \in \MEQ(V) \cap \calR(x_0)$. 
	We prove that $x^\star_1 = x^\star_2$.

    \smallskip
    \noindent
	\textit{Step 1. Consequence of membership 
	in MEQ set.} ~

    \noindent
	First, since $x^\star_1 \in \MEQ(V)$
	and $x^\star_2 \in \MEQ(V)$, 
    $0 = \langle V(x), x - x^\star_1\rangle = \langle V(x), x - x^\star_2 \rangle$
    for all $x \in \calX$.
	Subtracting the two equalities implies that
	\begin{equation*}
		\langle V(x), x^\star_1 - x^\star_2 \rangle = 0\,,\quad
		\forall x \in \calX \;.
	\end{equation*}
	Thus $x^\star_1 - x^\star_2$ is orthogonal to 
	$V(x)$ for all $x \in \calX$,
	meaning
	$x^\star_1 - x^\star_2 \in \calY^\bot = \Span(V(\calX))^\bot$.

	\smallskip
    \noindent
	\textit{Step 2. Consequence of membership 
	in reachable primal set.}  

    \noindent
	Now as $x^\star_1, x^\star_2 \in \calR(x_0)$, this implies
	there exist $y_1, y_2 \in \calY$ such that
	$x^\star_1 = \nabla F(y_0 + y_1)$ 
	and $x^\star_2 = \nabla F(y_0 + y_2)$.
	Define $u_1 := \Pi_{\calS^\bot} y_1 \in \calU$
	and $u_2 := \Pi_{\calS^\bot}y_2 \in \calU$. 
	Now by the shift-invariance property of
	$\nabla F$ from Proposition~\ref{prop:invariance-corr}, 
    it follows that
	$x^\star_1 = \nabla F(y_0 + u_1)$
	and 
	$x^\star_2 = \nabla F(y_0 + u_2)$.

	\smallskip
    \noindent
	\textit{Step 3. Combining the pieces.} ~

	As $x^\star_1 - x^\star_2 \in \calY^\bot$
	and $y_1 - y_2 \in \calY$, 
	it follows that $\langle x^\star_1 - x^\star_2, y_1 - y_2 \rangle = 0$.
	Moreover, we can express $y_1-y_2$ via the 
	orthogonal decomposition 
	\begin{equation*}
		y_1 - y_2 
		= 
		\Pi_{\calS^\bot} (y_1 -y_2) + \Pi_{\calS}(y_1-y_2)
		= 
		u_1 - u_2 + \Pi_{\calS}(y_1 -y_2) \;,
	\end{equation*}
	where the final equality uses the definitions of 
	$u_1, u_2 \in \calU \subseteq \calS^\bot$.
	Combining these two facts yields
	\begin{align}
		0 
		&= 
			\langle x^\star_1 - x^\star_2, y_1 - y_2 \rangle
			\nonumber \\
		&= 
			\langle x^\star_1 - x^\star_2, u_1 - u_2 \rangle
			+ 
			\langle x^\star_1 - x^\star_2, \Pi_{\calS}(y_1 - y_2)\rangle
			\nonumber \\
		&= 
			\langle 
			 \nabla F(y_0 + u_1) - \nabla F(y_0 + u_2),
			u_1 - u_2 
			\rangle \,,
			\label{eq:meqrx-01}
	\end{align}
	where the final equality comes from the fact
	that $x^\star_1, x^\star_2 \in \calX$
	and thus $x^\star_1 - x^\star_2 \in \calS^\bot$,
	and also using the characterization of
	$x^\star_1, x^\star_2$ from Step 2. 
    We now prove that $u \mapsto \nabla F(y_0 + u)$ is strictly monotone on $\calS^{\bot}$. 
    Let $d = u_1 - u_2$, 
    and observe by the fundamental theorem of calculus
    that 
    \begin{equation*}
    \langle
    \nabla F(y_0+u_1)-\nabla F(y_0+u_2),d
    \rangle
    =
    \int_{0}^1 \,
    \langle
    d,\nabla^2F(y_0+u_2+sd)d
    \rangle\,ds.
    \end{equation*}
    By convexity of $F$, and using 
    Proposition~\ref{prop:second-order-duality} (iii),
    the Hessian $\nabla^2 F$
    in the integral is 
    positive definite on $\calS^\bot$.
    As $u_1, u_2 \in \calU \subseteq \calS^\bot$,
     $d = u_1 - u_2 \in \calS^\bot$.
    Thus, by positive definiteness, 
    the integral must be strictly
    positive for $d \neq 0$. 
    Thus,~\eqref{eq:meqrx-01}
    subsequently implies that
    we must have $u_1 = u_2$, 
    and therefore 
     $x^\star_1 = 
     \nabla F(y_0 + u_1) = \nabla F(y_0 + u_2) 
     = x^\star_2$,
	which concludes the proof.
\end{proof}

\smallskip
\paragraph{Useful characterization
of reachable primal set.} We further give the following
useful characterization
of the reachable primal set $\calR(x_0)$:

\smallskip
\begin{restatable}
    {prop}{propreachablehelper}
    \label{prop:reachable-helper}
    Let $y_0 \in \R^D$
    and define $x_0 = \nabla F(y_0) \in \relint(\calX)$.
    Then for all $x \in \relint(\calX)$, 
    \begin{equation*}
        x \in \calR(x_0)
        \;\iff\;
        \Pi_{\calS^\bot}(\nabla R(x) - y_0) \in \calU \;.
    \end{equation*}
\end{restatable}

\begin{proof}
    Recall that $\calY = \Span(V(\calX))$, 
    $\calU = \Pi_{\calS^\bot} \calY$,
    and
    $\calR(x_0) = \{\nabla F(y_0 + v): 
    v\in \calY\}$.
    Now by the primal-dual 
    relationship of 
    Proposition~\ref{prop:first-order-duality},
    recall that
    for $x \in \relint(\calX)$
    and $y \in\R^D$
    we have $x = \nabla F(y)$
    if and only if 
    $y - \nabla R(x) \in \calS$.
    Then fixing such an $x \in \relint(\calX)$,
    we then have the 
    following string of equivalences:
    \begin{align*}
        x \in \calR(x_0)
        &\iff 
        \;\exists v \in \calY : 
        x = \nabla F(y_0 +v)\\
        &\iff 
        \;\exists v \in \calY : 
        y_0 + v - \nabla R(x) \in \calS \\
        &\iff
         \;\exists v \in \calY : 
        \Pi_{\calS^\bot} (y_0 + v - \nabla R(x))
        = 0 \\
        &\iff
        \;\exists v \in \calY : 
        \Pi_{\calS^\bot}(\nabla R(x) - y_0)
        = \Pi_{\calS^\bot} v \\
        &\iff
        \Pi_{\calS^\bot} (\nabla R(x) - y_0) 
        \in \calU \;,
    \end{align*}
    where the final equivalence
    uses the definition of 
    $\calU = \Pi_{\calS^\bot} \calY$. 
\end{proof}

\subsection{Regularity of Effective Stationarity Map}
\label{app:regularity}

\medskip
In this section, 
we derive an important consequence
of the restricted regularity condition
of Assumption~\ref{ass:restricted-regularity}.
In particular, we show in the
following proposition that
\textit{regular} zeros of the stationarity
map $\pi_0$ must also be MEQ points. 
As a result, under 
the conditions of Assumption~\ref{ass:restricted-regularity} (when applied to
the symmetrized system),
every reachable zero
of the map $\pi_0$ 
corresponds to an MEQ point
(see Corollary~\ref{cor:reachable-meq}).
We begin with the following proposition:

\medskip
\begin{restatable}[Regular reachable zeros are MEQ points]
{prop}{proprestricteddegeneracy}
	\label{prop:restricted-degeneracy}
	Let $\MEQ(V) \neq \emptyset$.
	Fix initial $y_0 \in \R^D$
	and $x_0 = \nabla F(y_0) \in \relint(\calX)$. 
	For any $y \in \calY$, let 
	$u = \Pi_{\calS^\bot} y \in \calU$. 
	If $\pi_0(u) = 0$
	and $J_{\pi_0}(u): \calU \to \calU$ is non-singular,
	then $\nabla F(y_0 + u) \in \MEQ(V)$.
\end{restatable}

\begin{proof}
Let $y\in\calY$ and set 
$u=\Pi_{\calS^\bot}y\in\calU$ and 
$x_u:=\nabla F(y_0+u).$ 
Suppose that
$\pi_0(u)=0$ 
and that
$J_{\pi_0}(u):\calU\to\calU$ 
is non-singular.
We will prove that $x_u\in\MEQ(V)$.

\smallskip
\noindent
\textit{Proof strategy.}
Since $\MEQ(V)\neq\emptyset$, let $x^\star\in\MEQ(V)$. 
Then, for all $x \in \calX$, 
\begin{equation}
\langle V(x), x - x_u \rangle 
= 
\langle V(x), x - x^{\star} \rangle - 
\langle V(x), x_u - x^{\star} \rangle 
= - \langle V(x), x_u - x^{\star} \rangle\,, 
\label{eq:master-proof}
\end{equation}
where the last equality follows from the fact 
that $x^\star\in\MEQ(V)$. 
Now, since $V(x)\in \Span(V(\calX)) =\calY$, 
observe that, if we can establish $x_u-x^\star\in\calY^\bot$,
then it will follow that $\langle V(x), x_u - x^{\star} \rangle = 0$.
From~\eqref{eq:master-proof}, this will then
imply that $\langle V(x), x - x_u \rangle = 0$ for all $x \in \calX$,
which shows that $x_u\in\MEQ(V)$. 
To establish that $x_u-x^\star\in\calY^\bot$, 
we first show that $x_u-x^\star$ is orthogonal 
to $\Range(J_{\pi_0}(u))$. Non-singularity of 
$J_{\pi_0(u)}$ then implies that this range is 
all of $\calU$. Finally, the fact that 
$\calU=\Pi_{\calS^\bot}\calY$ allows us to lift
orthogonality from $\calU$ to $\calY$. 
We proceed to give the details of these steps:

\noindent
\textit{Step 1. Show that $x_u-x^\star$ is orthogonal
to the range of $J_{\pi_0}(u)$.}

Fix any $h\in\calU$ and consider the curve
\begin{equation*}
x(t):=\nabla F(y_0+u+th)\,, 
\quad t \in \R\,.
\end{equation*}
By Proposition~\ref{prop:relative-duality}, 
$x(t) \in \relint(\calX)$ for all $t$.  
Since $x^\star\in\MEQ(V)$, then
$\langle V(x(t)),x(t)-x^\star\rangle=0$ 
for all sufficiently small $t$. 
Differentiating at $t=0$ then gives: 
\begin{equation}
    0
    =
    \big\langle
    J_V(x_u)\dot x(0),x_u-x^\star
    \big\rangle
    +
    \big\langle
    V(x_u),\dot x(0)
    \big\rangle.
    \label{eq:differentiation}
\end{equation}
Moreover,
$\dot x(0)=\nabla^2F(y_0+u)h.$ 
Since $h\in\calU\subseteq\calS^\bot$, 
Proposition~\ref{prop:second-order-duality} (iii) gives
$\dot x(0) = \nabla^2F(y_0+u)h \in \Range(\nabla^2F(y_0+u)) = \calS^\bot.$ 
On the other hand, $\pi_0(u)=0$ means
$\Pi_{\calS^\bot}V(x_u)=0,$ 
and hence
$V(x_u)\in\calS.$ 
Therefore, 
$\langle V(x_u),\dot x(0)\rangle=0$, 
and thus by~\eqref{eq:differentiation} we have:
\begin{equation}
    \big\langle
    J_V(x_u)\nabla^2F(y_0+u)h,
    x_u-x^\star
    \big\rangle
    =0 \;.
    \label{eq:diff-simplified1}
\end{equation}
Since $x_u,x^\star\in\calX$, we have
$x_u-x^\star\in\calS^\bot$ by Lemma~\ref{lem:normal-diff}.  
Thus projecting the first argument onto $\calS^\bot$ 
does not change the inner product, and it follows from 
\eqref{eq:diff-simplified1} that:  
\begin{equation}
\label{eq:diff-simplified2}
    \langle \Pi_{\calS^\bot}
    J_V(x_u)\nabla^2F(y_0+u)h,x_u-x^\star\rangle
    = 
    \big\langle
    J_V(x_u)\nabla^2F(y_0+u)h,
    x_u-x^\star
    \big\rangle
    =0\,.
\end{equation}
Using the chain rule, we have 
$J_{\pi_0}(u)h = \Pi_{\calS^\bot} J_V(x_u)\nabla^2F(y_0+u)h$,
and plugging this expression into \eqref{eq:diff-simplified2} 
then gives for all $h \in \calU$ that
$\langle J_{\pi_0}(u)h,x_u-x^\star\rangle=0$.
Thus $x_u-x^\star$
is orthogonal to
$\Range(J_{\pi_0}(u))$.

\smallskip
\noindent
\textit{Step 2. Use non-singularity to prove 
that $x_u-x^\star\in\calU^\bot$.}

\noindent
Since $J_{\pi_0}(u):\calU\to\calU$ 
is non-singular, it is surjective. Therefore,
$\Range(J_{\pi_0}(u))=\calU.$ 
The conclusion of Step 1 thus gives
$x_u-x^\star\in\calU^\bot.$

\smallskip
\noindent
\textit{Step 3. Lift orthogonality from $\calU$ to $\calY$.}

\noindent
Recall that
$\calU=\Pi_{\calS^\bot}\calY.$  
Also,
$x_u-x^\star\in\calS^\bot.$ 
Hence, for every $y'\in\calY$,
\begin{equation}
    \langle y',x_u-x^\star\rangle
    =
    \langle
    \Pi_{\calS^\bot}y',
    x_u-x^\star
    \rangle
    =0,
\end{equation}
because $\Pi_{\calS^\bot}y'\in\calU$ and
$x_u-x^\star\in\calU^\bot$ as shown in step 2. 
Therefore,
$x_u-x^\star\in\calY^\bot.$ 
\end{proof}

\medskip
As an immediate 
corollary of Proposition~\ref{prop:pi0-zeros-ne}
and Proposition~\ref{prop:restricted-degeneracy}, we then have under 
Assumption~\ref{ass:restricted-regularity}
(and by
Proposition~\ref{prop:regularity-sym},
which extends 
these properties 
to the present, symmetrized
system)
the following characterization:

\smallskip
\begin{restatable}
    {cor}{correachablemeq}
    \label{cor:reachable-meq}
    Let $y_0 \in \R^D$ and set $x_0= \nabla F(y_0)$, 
    and suppose that 
    Assumption~\ref{ass:restricted-regularity} holds
    for this initialization.
    Then for every $x \in \calR(x_0) \subseteq \relint(\calX)$,
    if $x \in \NE(V)$ then $x \in \MEQ(V)$. 
\end{restatable}

\subsubsection{Regularity Under Linear Harmonic Operators}

When $V: \calX \to \R^D$
is the linear map $V(x) = Kx$ 
for some $K \in \R^{D \times D}$
satisfying the harmonic property
$\MEQ(V) \cap \relint(\calX) \neq \emptyset$,
then $V$ always satisfies 
the regularity condition
of Assumption~\ref{ass:restricted-regularity}
(for the symmetrized system).
We formally state and prove this here:

\medskip
\begin{restatable}
    [Regularity under linearity]
    {prop}{proplinearregular}
    \label{prop:linear-regular}
    Let $V: \calX \to \R^D$ 
    be given by the linear map $x \mapsto Kx$
    for some $K \in \R^{D\times D}$,
    and suppose that
    $\MEQ(V) \cap \relint(\calX) \neq \emptyset$.
    Let $R$ satisfy Assumption~\ref{ass:legendre}
    with respect to $\calX$
    and let $F = R^*$
    satisfy Assumption~\ref{ass:lhs}.
    Fix an initialization pair 
    $y_0 \in \R^D$ and $x_0 = \nabla F(y_0) \in \relint(\calX)$ such that
    $\calR(x_0) \not \subseteq \NE(V)$. 
    Then 
    the Jacobian $J_{\pi_0}(u) : \calU \to \calU$
    is nonsingular at every $u \in \calU$,
    and thus 
    Assumption~\ref{ass:restricted-regularity}
    holds whenever $\calR(x_0) \not \subseteq \NE(V)$.
\end{restatable}

\begin{proof}
    To show nonsingularity of $J_{\pi_0}$, we show that its nullspace is trivial. 
    We first identify the structure of the projected 
    linear operator underlying $J_{\pi_0}$. 
    Since $\pi_0(u) = \Pi_{\calS^\bot} K \nabla F(y_0 + u)$ 
    for all $u \in \calU$, its Jacobian is:  
    \begin{equation}
    \label{eq:chain-rule}
        J_{\pi_0}(u) =  \Pi_{\calS^\bot} K \nabla^2 F(y_0 + u)\,,
        \quad \forall u \in \calU\,.
    \end{equation}
    Fix an arbitrary $u \in \calU$ and define $P := \nabla^2 F(y_0 + u)$  
    and recall from Proposition~\ref{prop:second-order-duality} 
    that $\Range(P) = \calS^{\bot}$.  
    Define then the linear map $A:\calS^\bot\to\calS^\bot$ by
    $Av:=\Pi_{\calS^\bot}Kv$ for any $v \in \calS^{\bot}$. 
    It then follows that $APv$ is well-defined for 
    any $v \in \calU$, and given \eqref{eq:chain-rule}, 
    we have $J_{\pi_0}(u)v = APv$ for any $v \in \calU$. 

    Now, recall from Proposition~\ref{prop:second-order-duality} (iii)
    that $P \succeq 0$ with $\Null(P)=\calS$. 
    In particular, $P$ is positive definite on $\calS^\bot$. 
    We then show that $A$ is skew-symmetric (e.g., $A^{\top} = -A$), 
    that $\Range(A) = \calU$, and that
    $\Null(A) = \calU^{\bot} \cap \calS^{\bot}$.
    We prove each one of these claims sequentially before 
    using them to conclude the statement of the proposition.

    \smallskip
    \noindent
    \textit{Proof of skew symmetry of $A: \calS^{\bot} \to \calS^{\bot}$.} 
    
    \noindent
    Let $p\in \MEQ(V)\cap\relint(\calX)$. 
    Since $p\in \MEQ(V)$, we have for any $x \in \calX$
    that $\langle Kx, x-p \rangle = 0$. 
    Since we also have $p \in \relint(\calX)$, 
    for any $v \in \calS^{\bot}$ and for sufficiently small $t > 0$, 
    we have $p + tv \in \calX$ and $p-tv \in \calX$. 
    Applying the MEQ identity to both points then gives:
    \begin{align}
        0 &= \langle K(p+tv), p + tv - p\rangle = t \langle Kp, v \rangle  + t^2 \langle Kv, v \rangle\,,\\
        0 &= \langle K(p-tv), p - tv - p\rangle = - t \langle Kp, v \rangle  + t^2 \langle Kv, v \rangle\,. 
    \end{align}
    Adding and subtracting the two identities gives respectively: 
    \begin{equation*}
        \langle Kv, v \rangle = 0\,,\quad \langle Kp, v \rangle = 0\,, \quad 
        \forall v \in \calS^{\bot}\,. 
        \end{equation*}
    In particular, it follows from the second identity above that $Kp \in \calS$ and hence $\Pi_{\calS^\bot} Kp = 0$. This fact will be used in later steps. 
    Now, using the definition of $A$ together with the first identity, 
    we have for any $v \in \calS^{\bot}$
    \begin{equation}
        \langle Av , v \rangle =  
        \langle \Pi_{\calS^\bot} Kv , v \rangle =  
        \langle Kv , v \rangle = 0\,.
        \label{eq:ssid}
    \end{equation}
    From this identity, we deduce that $A$ is skew-symmetric. 
    To see this, take any $v,w \in \calS^{\bot}$.
    Then~\eqref{eq:ssid} gives $\langle A (v +w), v+w \rangle = 0$. 
    Expanding this identity and using the facts that 
    $\langle Av , v \rangle=0$ and $\langle Aw, w \rangle=0$
    then yields $\langle Av , w \rangle = - \langle v , Aw \rangle$.
    Thus $A^{\top} = -A$,  where $A^{\top}$ is the adjoint as an operator on $\calS^{\bot}$. 

    \smallskip
    \noindent
    \textit{Proof of $\Range(A) = \calU$.} 

    \noindent
    Recall that $\calU = \Span\{\Pi_{\calS^\bot} Kx: x \in \calX\}$.
    Since $\Pi_{\calS^\bot} Kp = 0$ as proved above, 
    we have for every $x \in \calX$ that 
    $\Pi_{\calS^\bot} Kx= \Pi_{\calS^\bot} K(x-p) = A(x-p)$. 
    Therefore $\calU = \Span \{A(x-p): x \in \calX \}$. 
    Moreover, since $\Span(\calX - p) = \calS^{\bot}$ by 
    Lemma~\ref{lem:normal-diff}, then using the linearity of $A$, we have: 
    \begin{equation*}
    \calU = \Span \{A(x-p): x \in \calX \} 
    = A(\Span(\calX - p)) = A(\calS^{\bot}) = \Range(A)\,,
    \end{equation*}
    where the last identity follows by definition of $A$. 

    \smallskip
    \noindent
    \textit{Proof of $\Null(A) = \calU^{\bot} \cap \calS^{\bot}$.} 

    \noindent
    Given that $\calU = \Range(A)$ from the previous step, 
    it is sufficient to prove that 
    $\Null(A) = \Range(A)^{\bot} \cap \calS^{\bot}$. 
    Note first that $\Null(A) \subset \calS^{\bot}$ 
    by definition of $A: \calS^{\bot} \to \calS^{\bot}$.
    We now prove the two inclusions to show the desired identity. 
    We use the skew symmetry of $A$ proved in the first point above. 
    If $x \in \Null(A) \subseteq \calS^{\bot}$, 
    then for every $y \in \calS^{\bot}$, 
    $\langle x, Ay\rangle =  \langle A^{\top} x, y\rangle 
    = - \langle A x, y\rangle = 0$, 
    where the last identity is due to $x \in \Null(A)$. 
    Hence we have proved that $x \in \Range(A)^{\bot}$. 
    Now, conversely, if $x \in \Range(A)^{\bot} \cap \calS^{\bot}$, 
    then for any $y \in \calS^{\bot}$, 
    $0 = \langle x, Ay\rangle 
    = \langle A^{\top} x, y\rangle 
    = - \langle Ax, y\rangle$,
    where the first identity follows from the
    fact that $x \in \Range(A)^{\bot}$. 
    Taking $y = Ax \in \calS^\bot$ gives $\|Ax\|^2 = 0$, 
    and therefore $Ax =0$. Hence $x \in \Null(A)$. 

    \smallskip
    \noindent
    \textit{Conclusion: $J_{\pi_0}(u)$ is non-singular.} 

    \noindent
    We are now ready to show that $\Null(J_{\pi_0}(u)) = \{0\}$.
    Suppose that $J_{\pi_0}(u) v = 0$ for some $v \in \calU$. 
    Since $J_{\pi_0}(u)v = APv$ for any $v \in \calU$, this gives $APv  = 0$,
    and hence $Pv \in \Null(A)$.
    But $\Null(A) = \calU^{\bot} \cap \calS^{\bot}$
    as shown above. 
    Thus, $Pv \in \calU^{\bot}$. Since we also have 
    $v \in \calU$, we finally obtain $\langle v, Pv\rangle = 0$.
    Now, since $P$ is positive definite on $\calS^{\bot}$ and 
    $v \in \calU \subset \calS^{\bot}$, we must have $v=0$. 
    It follows that $J_{\pi_0}(u):\calU\to\calU$ has a
    trivial nullspace, and it is thus injective and 
    hence nonsingular since $\calU$ is finite-dimensional.
    As $u \in \calU$ was arbitrary, 
    we thus conclude that
    Assumption~\ref{ass:restricted-regularity} holds. 
\end{proof}


\section{Intuition from Continuous Time}
\label{app:ct-intuition}

In this section, we give more details
on the continuous-time intuition
introduced in Section~\ref{sec:proof-overview:energy}.

\paragraph{Continuous-time FTRL flow.} 
First, recall that we can consider the 
FTRL flow:
\begin{equation}
	\dot y_t = - V(x_t)
	\quad\text{and}\quad
	x_{t} = \nabla F(y_t) \;.
	\tag{FTRL flow}
	\label{eq:ftrl-flow-app}
\end{equation}
Under~\eqref{eq:ftrl-flow-app}, we showed in~\eqref{eq:flow-conservation}
that 
\begin{equation}
    \textstyle
    \tfrac{d}{dt}
    \D_R(x^\star, x_t)
    = 
    \tfrac{d}{dt} F_{x^\star}(y_t)
    = 
    - \langle V(x_t), x_t - x^\star \rangle 
    = 
    0 
\end{equation}
for all points $x^\star \in \MEQ(V)$.

\medskip
\paragraph{Forward and backward
discretizations of FTRL flow.} 
Under the forward and backward discretizations
of~\eqref{eq:ftrl-flow}, we then obtain 
the following bounds on change in shifted energy,
as stated in~\eqref{eq:forward-backward-energy}:

\begin{restatable}[Change in shifted energy
for FTRL and IFTRL]
	{prop}{propftrliftrlenergy}
	\label{prop:ftrl-iftrl-energy}
	Consider the symmetrized system of Setting~\ref{setting:weighted}.
	Then along~\eqref{eq:ftrl} and~\eqref{eq:iftrl},
    and for any $x^\star \in \MEQ(V)$:
		\begin{equation*}
		F_{x^\star}(y_{t+1}) - F_{x^\star}(y_t)
		= 
		\begin{cases}
		\;\;\;\D_F(y_{t+1}, y_t) \;\ge\; 0 
		\quad\textnormal{for FTRL} \\
		-\D_F(y_t, y_{t+1}) \;\le\; 0 
		\quad\textnormal{for IFTRL} \;.
		\end{cases} 
		\end{equation*}	
\end{restatable}

\begin{proof}
Let $x^\star \in  \MEQ(V).$ Then, by definition of $\MEQ(V)$, we have: 
\begin{equation}
\label{eq:MEQ(V)-xstar}
\langle V(x), x- x^{\star} \rangle = 0\,, \quad \forall x \in \calX\,.
\end{equation}
By definition of the shifted energy function $F_{x^\star}$
(see \eqref{eq:shifted-energy}), we have for any $t \geq 0$, 
\begin{equation}
\label{eq:shifted-energy-diff}
F_{x^\star}(y_{t+1}) - F_{x^\star}(y_{t}) = F(y_{t+1}) - F(y_{t}) - \langle x^{\star}, y_{t+1} - y_t \rangle\,.
\end{equation}
We start by proving the identity for vanilla FTRL. 
By definition of $\D_F(y_{t+1}, y_t)$, 
we have for any $t \geq 0$, 
\begin{align}
\label{eq:bregman-forward}
\D_F(y_{t+1},y_t) &= F(y_{t+1}) - F(y_t) - 
\langle \nabla F(y_{t}), y_{t+1} - y_t \rangle \nonumber\\
&= F(y_{t+1}) - F(y_t) - \langle x_t, y_{t+1} - y_t \rangle\,, 
\end{align}
where the last identity follows from 
the primal-dual relationship $x_t = \nabla F(y_t).$
Rearranging~\eqref{eq:bregman-forward} 
and combining it with~\eqref{eq:shifted-energy-diff} yields: 
\begin{align}
F_{x^\star}(y_{t+1}) - F_{x^\star}(y_{t}) 
&=  \D_F(y_{t+1},y_t) + \langle x_t - x^{\star}, y_{t+1} - y_t \rangle\nonumber\\
&= \D_F(y_{t+1},y_t) - \eta \langle x_t - x^{\star}, V(x_t) \rangle \nonumber\\
&= \D_F(y_{t+1},y_t)\,,
\end{align}
where the penultimate identity follows from using 
the FTRL dual update rule $y_{t+1} = y_t - \eta V(x_t)$, 
and the last equality follows from~\eqref{eq:MEQ(V)-xstar},
since $x^\star \in  \MEQ(V)$ and $x_t \in \mathcal{X}$ for all $t \geq 0$. 
This concludes the proof of the first identity for FTRL.

As for IFTRL, the proof follows similar lines.
Again using the definition of the Bregman divergence, 
we have for any $t \geq 0$ that
\begin{align}
\label{eq:bregman-backward}
\D_F(y_t,y_{t+1}) &= F(y_{t}) - F(y_{t+1}) 
- \langle \nabla F(y_{t+1}), y_{t} - y_{t+1} \rangle \nonumber\\
&= F(y_{t}) - F(y_{t+1}) - \langle x_{t+1}, y_{t} - y_{t+1} \rangle\,, 
\end{align}
where the last identity follows from the primal-dual 
relationship $x_{t+1} = \nabla F(y_{t+1})$.
Rearranging and combining the above identity~\eqref{eq:bregman-backward} 
with~\eqref{eq:shifted-energy-diff} yields for any $t \geq 1$, 
\begin{align}
F_{x^\star}(y_{t+1}) - F_{x^\star}(y_{t}) 
&= - \D_F(y_{t},y_{t+1}) + 
\langle x^{\star} - x_{t+1}, y_{t} - y_{t+1} \rangle\nonumber\\
&= - \D_F(y_t,y_{t+1}) 
+ \eta \langle x^{\star} - x_{t+1}, V(x_{t+1}) \rangle \nonumber\\
&= - \D_F(y_{t},y_{t+1})\,,
\end{align}
where the penultimate identity follows from using 
the IFTRL dual update rule $y_{t+1} = y_t - \eta V(x_{t+1})$, 
and the last equality follows from \eqref{eq:MEQ(V)-xstar},
since $x^\star \in  \MEQ(V)$ and $x_{t+1} \in \mathcal{X}$ 
for all $t \geq 0$. 
This concludes the proof of the second identity for IFTRL.
\end{proof}

\section{Properties of Energy Function Hessian}
\label{app:error-machinery}

In this section, we develop the machinery that will be used
in the proof of Lemma~\ref{lem:energy-dissipation-main},
which establishes one-step bounds on the dissipation
of the shifted energy function along the OFTRL and EGFTRL iterates.
We establish this machinery 
for the symmetrized system with operator $V: \calX \to \R^D$,
regularizer $R$, and energy function $F = R^*$
(as summarized in Proposition~\ref{prop:reduction-algorithmic}). 
In particular, most of the bounds derived in this section
stem from the local Hessian stability property~\eqref{eq:lhs}
for the energy Hessian~$\nabla^2 F$.

We start with a brief recap 
of the assumptions satisfied under the symmetrized system:

\paragraph{Summary of global assumptions for
symmetrized system.} 
Recall from~\ref{setting:weighted} that the operator
$V:\calX \to \R^D$, regularizer $R: \calX \to \R$, 
and energy $F: \R^D \to \R$ satisfy the following properties:
\begin{itemize}[
    leftmargin=1.5em,
]
    \item 
    $V$ is $\wh L$-Lipschitz and $\wh B$-Bounded. 
    \item 
    $R$ is relatively Legendre and
    1-strongly-convex with respect to $\|\cdot\|$.
    \item 
    $F=R^*$ is 1-smooth with respect 
    to $\|\cdot\|_{*}$ and 
    satisfies~\eqref{eq:lhs} with constant $\wh\omega > 0$.
\end{itemize}

\subsection{Properties of Local Semi-norms}
\label{app:error-machinery:seminorms}

We start by recalling the notation 
for the \textit{local semi-norms}
introduced in Section~\ref{sec:main-proof-overview}.
Specifically, for all $y \in \R^D$,
recall that $\nabla^2 F(y) \succeq 0$ 
is positive semi-definite,
but with the non-trivial nullspace 
$\Null(\nabla^2 F(y)) = \calS$
(Proposition~\ref{prop:second-order-duality}).
Thus, $\nabla^2 F(y)$ can be used to define a local
degenerate inner-product, which in turn induces a local semi-norm as follows:
\begin{equation}
    \langle u, v \rangle_{y}
    := \langle u, \nabla^2 F(y) v\rangle 
    \;\;\text{and}\;\;
    \|v\|_{y} := \sqrt{\langle v, \nabla^2 F(y) v \rangle}
    \;\;\text{for all $u, v \in \R^D$}\;.
    \label{eq:local-norm}
\end{equation}
Importantly, note that
$\|v\|_{y} = 0$ if and only if 
$v \in \calS = \Null(\nabla^2 F(y))$
(which holds even if $v \neq 0$).

When restricted to the subspace $\calS^\bot$,
we also define the dual semi-norm $\|\cdot\|_{y, *}$.
In particular, for $x = \nabla F(y)$, 
recall that we have the primal-dual relationship
between the operators $\nabla^2 F$ and $\nabla^2 R$ given 
by $(\nabla^2 F(y))^{-1} = \nabla^2 R(x)$, when the operators
are restricted to the subspace $\calS^\bot$,
e.g., when applied to vectors $v \in \calS^\bot$
(Proposition~\ref{prop:second-order-duality}).
Then for $y \in \R^D$, we define
\begin{equation}
    \|u\|_{y, *} 
    := 
    \sqrt{\langle u, (\nabla^2 F(y))^{-1} u\rangle}
    \;\;\text{for $u \in \calS^\bot$}.
    \label{eq:dual-local-norm}
\end{equation}
For $y \in \R^D$, the semi-norm $\|\cdot\|_{y}$ 
satisfies the Cauchy-Schwarz inequality:
\begin{equation}
    |\langle v, v' \rangle_{y}|
    \le \|v\|_{y} \cdot \|v'\|_{y}
    \;\;\text{for $v, v' \in \R^D$}\;.
    \label{eq:local-cs}
\end{equation}
Moreover, the pair
$(\|\cdot\|_{y}, \|\cdot\|_{y, *})$
satisfies the generalized Cauchy-Schwarz inequality:
\begin{equation}
    |\langle u, v \rangle |
    \le
    \|u\|_{y, *} \cdot \|v\|_{y}
    \;\;\text{for $u \in \calS^\bot$
    and $v \in \R^D$}\;.
    \label{eq:local-gcs}
\end{equation}
Finally, fixing $y \in \R^D$
and letting $P_{y} = \nabla^2 F(y)$
and $P^{-1}_y = (\nabla^2 F(y))^{-1}$
(restricted to $\calS^\bot$), 
we further have by definition
the following equivalent and
useful characterizations:
\begin{equation}
    \|v\|_y = \|P^{1/2}_{y} v\|_2
    \;\;\text{and}\;\;
    \|u\|_{y, *} = \|P^{-1/2}_{y} u\|_2
    \;\;\text{for all $v \in \R^D$
    and $u \in \calS^\bot$} \;.
    \label{eq:local-norm-useful}
\end{equation}
See~\cite{rudin1991functional}
for further background
and preliminaries 
on the degenerate inner products
and semi-norms induced by 
positive semi-definite matrices.

\subsection{Relationship between Local and Global (Semi)-Norms}

The following lemma relates the above local norms to 
particular global norms assuming strong convexity of
the regularizer $R$ with respect to the global norm.
\begin{restatable}{lem}{lemnorm-globalvslocal}
	\label{lem:norm-global-vs-local}
    Assume that $R$ is 1-strongly-convex
with respect to a particular norm $\|\cdot \|$ (equivalently, that
$F$ is 1-smooth with respect to its dual norm $\|\cdot \|_{*}$).
	Then, for all $y \in \R^D$ the following global relationships hold.
	\begin{enumerate}[
		label={(\roman*)},
	]
	\item
	$\|v\|_y \le \|v\|_*$ for all $v \in \R^D$. 
	\item
	$\|u\| \le \|u\|_{y, *}$ for all $u \in \calS^\bot$.
	\end{enumerate}
\end{restatable}

\begin{proof}
    Fix $y \in \R^D$ 
    and $x = \nabla F(y) \in \relint(\calX)$. 
    First, by 1-strong-convexity
    of $R$ with respect to the norm $\|\cdot\|$,
    we have the second-order characterization:
    $\langle u, \nabla^2 R(x) u \rangle
    \ge \|u\|^2$ for all $u \in \calS^\bot$
    (see, e.g.,~\cite{boyd2004convex}, Section 9.1).
    As $\nabla^2 R(x) = (\nabla^2 F(y))^{-1}$
    when restricted to the subspace $\calS^\bot$ 
    it then follows for all $u \in \calS^\bot$ that
    \begin{equation*}
        \|u\|^2
        \le
        \langle 
        u, \nabla^2 R(x) u
        \rangle
        = 
        \langle u, (\nabla^2 F(y))^{-1} u\rangle
        = \|u\|^2_{y, *} \;.
    \end{equation*}
    Taking square roots yields claim (ii) of the lemma.
    For claim (i),
    we similarly have
    by $1$-smoothness of $F = R^*$ 
    with respect to the dual norm $\|\cdot\|_{*}$ that
    $
        \|v\|^2_{y} =
        \langle v, \nabla^2 F(y) v \rangle
        \le 
        \|v\|^2_{*}
    $
    for all $v \in \R^D$.
    Taking square roots then yields claim (i)
    and completes the proof.
\end{proof}

\subsection{Local Semi-Norm Transfer Inequalities}

Under the local Hessian stability property of Assumption~\ref{ass:lhs},
the following inequalities are immediate:

\begin{restatable}{lem}{lemlocalnormtransfer}
	\label{lem:local-transfer}
	Let $y, y' \in \R^D$ satisfy
	$\|y' - y\|_{*} \le \alpha$
	for $\alpha > 0$. Then 
	for every $v \in \R^D$ it holds that
	\begin{equation*}
		\exp(- \wh \omega \alpha)
		\cdot \|v\|^2_{y'}
		\le 
		\|v\|_y^2 
		\le
		\exp(\wh \omega \alpha)
		\|v\|^2_{y'} \;.
	\end{equation*}
\end{restatable}

\subsection{Two-Sided Bounds on Bregman Divergences}

In this section, we derive bounds on 
the Bregman divergence $\D_F(\cdot, \cdot)$
induced by the energy Hessian 
that hold under the local semi-norm.

\paragraph{Bregman divergence in integral form.}
First, recall that the Bregman divergence $\D_F(y', y)$
admits the following integral formulation
as a result of Taylor's theorem with integral remainder:
\begin{restatable}
    {prop}{propbregmanintegral}
    \label{prop:bregman-integral}
    Let $y, y' \in\R^D$, and 
    let $y_s := y + s(y'-y)$ for all $s \in [0, 1]$. 
    Then the Bregman divergence $\D_F(y', y)$ is given by
    \begin{equation*}
        \D_F(y', y)
		  = 
		\int_{0}^1 
		(1-s) 
		\big\langle
			y'-y, \nabla^2 F(y_s) (y'-y)
		\big\rangle \, ds \;.
    \end{equation*}
\end{restatable}

This result is standard, and thus we omit a proof (see, e.g.,~\cite{nocedal2006numerical}, Theorem 2.1).

Under the \eqref{eq:lhs} property, 
the Bregman divergence  $\D_F(y',y)$ 
is locally equivalent to the squared local semi-norm $\|y'-y\|_{y}$. 
More precisely, whenever $y'$ is sufficiently close to $y$ 
in the global dual norm $\|\cdot\|_{*}$, the two quantities
are comparable up to multiplicative factors that
depend explicitly on $\|y'-y\|_{*}$. 
The following lemma makes this relationship precise.

\begin{restatable}{lem}{lembregmanlocalbound}
	\label{lem:bregman-local-bound}
	Let $y, y' \in \R^D$ be such that
	$\|y' - y\|_{*} \le \alpha$ for
    $\alpha > 0$.
	Then it holds that:
	\begin{equation*}
    \textstyle
	\big(
	\frac{\exp(-\wh \omega \alpha) +  \wh \omega\alpha - 1}{(\wh \omega\alpha)^2}
	\big)
	\cdot 
	\|y' - y\|^2_{y}
	\le
	\D_F(y', y)
	\le
	\big(
	\frac{\exp(\wh \omega\alpha) - \wh \omega\alpha - 1}{(\wh \omega\alpha)^2}
	\big)
	\cdot \|y' - y\|^2_{y} \;.
	\end{equation*}
\end{restatable}

\noindent
Note that this is the generalization of~\cite[Lemma E.9]{lazarsfeld2026},
and we include a full proof here for completeness:

\begin{proof}
	We first prove the upper bound.
	For this, let $y_s := y + s (y' - y)$ for $s \in [0, 1]$,
	and observe for all such~$s$ that 
	$\|y_s - y \|_{*}
		= 
		s \|y' - y\|_{*}
		\le
		s \alpha 
		\le \alpha$,
	where we use the assumption 
    that $\|y' - y\|_{*} \le \alpha$. 
	Then using the integral representation 
	of Bregman divergence (Proposition~\ref{prop:bregman-integral})
	and the upper bound part of Lemma~\ref{lem:local-transfer}, we obtain
	\begin{align*}
		\D_F(y', y)
		&= 
		\int_{0}^1 
		(1-s) 
		\big\langle
			y'-y, \nabla^2 F(y_s) (y'-y)
		\big\rangle \, ds \\
        		&= 
		\int_{0}^1 
		(1-s) 
		\|y' - y\|_{y_s}^2 \, ds \\
		&\le 
		\int_{0}^1 
		(1-s) 
		\Big( 
		\exp(\wh \omega\alpha s)
        \|y' - y\|_y^2 
		\Big) \, ds \\
		&= 
		\Big(\int_{0}^1(1-s) \exp(\wh \omega\alpha s) \, ds \Big)
		\cdot 
		\|y' - y\|^2_{y} \;.
	\end{align*}
	Integration by parts yields 
	$\int_{0}^1(1-s) \exp(\wh \omega\alpha s) ds = 
	\frac{\exp(\wh \omega\alpha) - \wh \omega\alpha -1}{(\wh \omega\alpha)^2}$, 
	which yields the upper bound. 

	\noindent
	For the lower bound, we use the lower bound part of Lemma~\ref{lem:local-transfer} similarly to obtain
	\begin{equation*}
		\D_F(y', y)
		\ge
		\Big(
		\int_{0}^1 (1-s) \exp(-\wh \omega\alpha s) \, ds
		\Big)
		\|y' - y\|^2_y \;.
	\end{equation*}
	Then integration by parts again yields
	$\int_{0}^1(1-s) \exp(-\wh \omega\alpha s) ds = 
	\frac{\exp(- \wh \omega\alpha) + \wh \omega\alpha -1}{(\wh \omega\alpha)^2}$, 
	which yields the lower bound of the lemma and
	concludes the proof. 
\end{proof}

\medskip
\subsection{Dual-to-Primal Local Semi-Norm Inequalities}

We next establish a relationship between
differences of primal iterates $x', x \in \calX$
under the local dual seminorm $\|\cdot\|_{y, *}$,
and differences of the corresponding 
dual iterates $y', y \in \R^D$ 
under the seminorm $\|\cdot \|_{y}$.
Formally, we prove the following:

\begin{restatable}{lem}{lemdualtoprimallocal}
	\label{lem:dual-to-primal-local}
	Let $y, y' \in \R^D$ such that $\|y'-y\|_{*} \le \alpha$
	for $\alpha > 0$.
    Let $x = \nabla F(y)$
	and $x' = \nabla F(y')$. Then it holds that:
	\begin{equation*}
    \textstyle
	\|x'- x\|_{y, *} \le 
	\big(\frac{\exp(\wh \omega\alpha)-1}{\wh \omega\alpha}\big)
	\cdot \|y' - y\|_{y}\;.
	\end{equation*}
\end{restatable}

\medskip
\begin{proof}
The proof proceeds in several steps:

\noindent
\textit{1. Shorthand notation.}
For readability, we begin by introducing 
several pieces of shorthand notation.
First, we drop the subscript $y$ to denote $P := \nabla^2 F(y)$. 
Recall that when restricted to $\calS^\bot$
the Hessian $P$ is positive-definite, and 
thus over this subspace we will denote the
energy Hessian inverse by $P^{-1} := (\nabla^2 F(y))^{-1}$.
For $s \in [0, 1]$, let $y_s := y + s (y' -y)$, 
and further define 
$P_s := \nabla^2 F(y_s)$,
and the normalized $B_s := P^{-1/2} P_s P^{-1/2}$.
As with $P^{-1}$, 
we only use $B_s$ when explicitly 
working over $\calS^\bot$.

\smallskip
\noindent
\textit{2. Bound on $\|B_s\|_2$.}
Recall that for all $s \in [0, 1]$ 
that $\|y_s - y\|_{*} = s\|y'-y\|_{*} \le \alpha s$.
Thus, the~\eqref{eq:lhs} property with respect to 
the pair $(y_s,y)$ can be expressed in our shorthand notation as
\begin{equation*}
	\exp(-\wh \omega\alpha s) P 
	\preceq
	P_s 
	\preceq \exp(\wh \omega\alpha s) P \;.
\end{equation*}
Then by congruence (using the fact that $P^{-1/2}$ is symmetric), 
it follows for all $s \in [0,1]$ that
\begin{equation*}
	\exp(-\wh \omega\alpha s) I \preceq B_s \preceq \exp(\wh \omega\alpha s) I \;.
\end{equation*}
By definition of the spectral norm, this implies for
all $s \in [0, 1]$ that $\|B_s\|_2 \le \exp(\wh \omega\alpha s)$.

\smallskip
\noindent
\textit{3. Expand primal difference in integral form.}
Recall by assumption that $x = \nabla F(y)$
and $x' = \nabla F(y')$. It follows by the fundamental
theorem of calculus (see, e.g.,
\cite{nocedal2006numerical}, Theorem 2.1)
that we can write 
\begin{equation}
	x' - x 
	= 
	\nabla F(y') - \nabla F(y)
	= 
	\int_{0}^1 \nabla^2 F(y_s) (y'- y) \, ds 
	=
	\int_{0}^1 P_s (y'- y) \, ds \;.
	\label{eq:grad-diff-expression}
\end{equation}
As $x' - x \in \calS^\bot$ by Lemma~\ref{lem:normal-diff-weighted},
we can take its dual local seminorm. 
Using the triangle inequality with respect to 
the integral then yields 
\begin{equation}
	\|x' - x\|_{y, *}
	= 
	\Big\| 
	\int_{0}^1 P_s (y'- y) \, ds
	\Big\|_{y, *}
	\le 
	\int_{0}^1 
	\| 
	P_s (y'- y)
	\|_{y, *} \, ds \;.
	\label{eq:ldp01}
\end{equation}
Now by definition of $\|\cdot \|_{y, *}$
and $B_s = P^{-1/2} P_s P^{-1/2}$,
we have for all $s \in [0, 1]$ that
\begin{align}
	\|P_s(y'-y)\|_{y, *}
	= 
	\|P^{-1/2} P_s(y'-y)\|_2
	&= 
	\| B_s P^{1/2} (y'-y) \|_2 
		\nonumber \\
 	&\le \| B_s\|_2 \cdot \|P^{1/2} (y'-y)\|_2 
 		\nonumber \\
 	&\le \exp(\wh \omega\alpha s) \cdot  \|P^{1/2} (y'-y)\|_2 
 		\nonumber \\
 	&=  \exp(\wh \omega\alpha s) \cdot \|y'-y\|_y  \;.
 	\label{eq:ldp02}
\end{align}
Here, the first inequality follows by 
definition of the spectral norm $\|B_s\|_2$,
the second inequality follows from the bound in
Step 2, and the final equality comes from the
definition of $\|\cdot\|_y$.

Substituting~\eqref{eq:ldp02} into~\eqref{eq:ldp01}
and integrating then yields
\begin{equation*}
	\|x'-x\|_{y, *}
	\le 
	\Big(
		\int_{0}^1 \exp(\wh \omega\alpha s) \, ds
	\Big)
	\cdot 
	\|y' -y \|_{y}
	= 
	\Big(\frac{\exp(\wh \omega\alpha) - 1}{\wh \omega\alpha}\Big)
	\cdot 
	\|y' - y\|_y \;,
\end{equation*}
which completes the proof.
\end{proof}

\medskip
\subsection{Upper Bound on Hessian Remainder Term}

In this section, we derive bounds on the magnitude
of the following \textit{Hessian remainder term}:

\paragraph{Hessian remainder term.}
For $y', y \in \R^D$, let $\G_F(y', y) \in \R^D$
be the vector
\begin{equation}
    \G_F(y', y) 
    = 
    \nabla F(y') - \nabla F(y)
    - \nabla^2 F(y)(y'-y)
    \;.
    \label{eq:hessian-remainder}
\end{equation}
The vector $\G_F(y', y)$
can be thought of as a higher-order
analogue of the Bregman divergence $\D_F(y', y)$.
Specifically, $\G_F(y', y)$ is the 
first-order error vector associated
with approximating the difference of gradients
$\nabla F(y') - \nabla F(y)$
by the matrix-vector product  $\nabla^2 F(y)(y'-y)$. 

Under the~\eqref{eq:lhs} property,
we can then bound
the norm of $\G_F(y', y)$ under
the semi-norm $\|\cdot\|_{y, *}$.
Formally, we prove the following:

\smallskip
\begin{restatable}{lem}{lemhessianremainder}
	\label{lem:hessian-remainder-bound}
	Let $y, y' \in \R^D$ such that
	$\|y'-y\|_{*} \le \alpha$ for $\alpha > 0$.
	Let $\G_F(y', y)$
	be the Hessian remainder term as defined in~\eqref{eq:hessian-remainder}.
    Then it holds that 
	$\G_F(y', y) \in \calS^\bot$ and that
	\begin{equation*}
		\textstyle
		\|\G_F(y', y)\|_{y, *} 
		\le 
		\big(
			\frac{\exp(\wh \omega\alpha) - 1 - \wh \omega\alpha}{\wh \omega\alpha}
		\big)
		\cdot 
		\|y' - y\|_{y} \;.
	\end{equation*}
\end{restatable}

\medskip
\begin{proof}
	The proof proceeds in several steps
	and reuses several calculations 
	from the proof of Lemma~\ref{lem:dual-to-primal-local}.
	We start by re-introducing several pieces
	of notation from that proof:

    \smallskip
    \noindent
    \textit{1. Shorthand notation.}
	We let $P := \nabla^2 F(y)$, and let
	$P^{-1} = (\nabla^2 F(y))^{-1}$ when restricted 
	to $\calS^\bot$. 
	For all $s \in [0, 1]$, let 
	$y_s := y + s(y'-y)$, 
	let $P_s := \nabla^2 F(y_s)$,
	and let $B_s := P^{-1/2} P_s P^{-1/2}$.

    \smallskip
    \noindent
	\textit{2. Properties of $B_s$ from 
	Lemma~\ref{lem:dual-to-primal-local}.}
	From Step 2 of the proof of Lemma~\ref{lem:dual-to-primal-local},
	recall that we established for all $s \in [0, 1]$ that
	$\exp(-\wh \omega\alpha s) I \preceq B_s \preceq \exp(\wh \omega\alpha s) I$.
	This further implies for $B_s - I$ that:
	\begin{equation*}
		(\exp(-\wh \omega\alpha s) -1) I \preceq B_s - I \preceq (\exp(\wh \omega\alpha s) - 1)I  \;.
	\end{equation*}
	As $\exp(-x) - 1 \ge 1-\exp(x)$ for all $x \in \R$
    (by the AM-GM inequality), we further 
	have that $B_s - I \succeq -(\exp(\wh \omega \alpha s) -1)I$.
	Thus, it follows from the definition of spectral norm that 
	$\|B_s - I\|_{2} \le \exp(\wh \omega\alpha s) - 1$ for all $s \in [0, 1]$.

    \smallskip
    \noindent
	\textit{3. Expand $\G_F$ in integral form.}
    Then, using the definition of $\G_F(y', y)$
    from~\eqref{eq:hessian-remainder},
    together with the expression for $\nabla F(y') - \nabla F(y)$
	from~\eqref{eq:grad-diff-expression} in the proof
	of Lemma~\ref{lem:dual-to-primal-local}, it follows that 
	\begin{equation}
		\G_F(y', y)
		= 
		\int_{0}^1 (P_s - P)(y'-y) ds  \;.
		\label{eq:hessian-remainder-integral}
	\end{equation}
	Now recall from Proposition~\ref{prop:second-order-duality} 
    (applied to $\nabla^2 F$) that 
    $\Range(\nabla^2 F(y'')) = \calS^\bot$
    for all $y'' \in \R^{D}$.
	As $\calS^\bot$ is a linear subspace, 
	then also $(P_s - P)(y'-y) \in \calS^\bot$
	for all $s \in [0, 1]$, and thus
	$\G_F(y', y) \in \calS^\bot$. 
	This establishes the first claim of the lemma. 

    \smallskip
    \noindent
	\textit{4. Bound on $\G_F$ in dual seminorm.}
	Using expression~\eqref{eq:hessian-remainder-integral} and noting that we have shown that $\G_F(y',y) \in \calS^\bot$,
	we can take the dual semi-norm $\|\cdot\|_{y,*}$ and apply
	the triangle inequality to write:
	\begin{equation}
		\|\G_F(y', y)\|_{y, *}
		\le 
		\int_{0}^{1} \|(P_s - P)(y'-y)\|_{y, *} ds \;.
		\label{eq:gf00}
	\end{equation}
	Then for each $s \in [0,1]$, we can use the
	definition of $B_s = P^{-1/2} P_s P^{-1/2}$
	to further write
	\begin{align}
		\|(P_s - P)(y'-y)\|_{y, *}
		&=
		\|P^{-1/2} (P_s - P)(y'-y)\|_2 
			\nonumber \\
		&= \|(B_s P^{1/2} - P^{1/2})(y'-y)\|_{2} 
			\nonumber  \\
		&= \|(B_s - I) P^{1/2}(y'-y)\|_{2} 
			\nonumber \\
		&\le \|B_s - I \|_2 \cdot \|P^{1/2}(y'-y)\|_{2}
		\le (\exp(\wh \omega\alpha s)-1) \cdot \|y'-y\|_{y}
		\label{eq:gf01} \;.
	\end{align}
	Here, the final inequality uses the 
	bound on $\|B_s - I \|_2$ from Step 2
	and the definition of $\|\cdot \|_y$.
	
	Substituting~\eqref{eq:gf01} into~\eqref{eq:gf00}
	and integrating then yields
	\begin{equation*}
		\|\G_F(y', y)\|_{y, *}	
		\le 
		\Big(\int_{0}^1 (\exp(\wh \omega\alpha s)-1) ds \Big) 
		\cdot \|y'-y\|_{y} 
		= 
		\Big(\frac{\exp(\wh \omega\alpha) - 1 - \wh \omega\alpha}{\wh \omega\alpha}\Big)
		\cdot \|y'- y\|_{y} \;,
	\end{equation*}
	which concludes the proof. 
\end{proof}

\medskip
\subsection{Upper Bound on Operator Differences}

Under $\wh L$-Lipschitzness of $V$,
1-strong-convexity of $R$, 
and using the dual-to-primal local semi-norm
bound of Lemma~\ref{lem:dual-to-primal-local},
we can derive the following key lemma that bounds
the magnitude of operator differences $\|V(x') - V(x)\|_y$ 
in terms of the magnitude of 
dual increments $\|y' - y\|_{y}$. 
Formally:

\smallskip
\begin{restatable}{lem}{lemopdiffdualdiff}
	\label{lem:op-to-dual-diff}
	Fix $y, y' \in \R^D$
	such that $\|y' - y\|_* \le \alpha$
	for $\alpha > 0$,
	and let $x = \nabla F(y) \in \relint(\calX)$
	and $x' = \nabla F(y') \in \relint(\calX)$.
	Then it holds that
	\begin{equation*}
        \textstyle
		\|V(x') - V(x)\|_y
		\le
		\wh L \cdot 
		\big(
			\frac{\exp(\wh \omega\alpha) - 1}{\wh \omega\alpha}
		\big)
		\cdot 
		\|y' - y \|_y \;.
	\end{equation*}
\end{restatable}

\smallskip
\begin{proof}
	Observe the following sequence of inequalities:
	\begin{align*}
		\|V(x')- V(x)\|_{y}
		&\le 
		\|V(x') - V(x)\|_{*}
			&&\text{(Lemma~\ref{lem:norm-global-vs-local}, part (i))} \\
		&\le
		\wh L \cdot \|x'-x\| 
			&&\text{(Lipschitzness of $V$)} \\
		&\le
		\wh L \cdot \|x' - x\|_{y, *}
			&&\text{(Lemma~\ref{lem:norm-global-vs-local}, part (ii))} \\
		&\le 
		\wh L \cdot 
		\big(
			\tfrac{\exp(\wh \omega\alpha) - 1}{\wh \omega\alpha}
		\big)
		\cdot 
		\|y' - y \|_y  
			&&\text{(Lemma~\ref{lem:dual-to-primal-local})} 
			\;.
	\end{align*}
	Note that the third inequality holds due to the fact that 
	$x, x' \in \calX$, which means
	$x'-x \in \calS^\bot$ and allows for 
	applying Lemma~\ref{lem:norm-global-vs-local}, part (ii).
	This completes the proof.
\end{proof}

\section{Proof of Lemma~\ref{lem:energy-dissipation-main}: Energy Dissipation}
\label{app:energy-unified}

In this section, we develop the proof of 
Lemma~\ref{lem:energy-dissipation-main}.
This lemma establishes tight bounds on 
the one-step change in (shifted) energy dissipation
for OFTRL and EGFTRL
under the symmetrized system of 
Proposition~\ref{prop:reduction-algorithmic}
(as summarized in Setting~\ref{setting:weighted}). 
We restate the lemma here:

\smallskip
\lemenergydissipationmain*

Before giving the proof of the lemma,
we mention several further remarks: 

\smallskip
\begin{restatable}[Assumptions and generality]
	{rem}{remdissipationlemassumptions}
	\label{rem:dissipation-assumptions}
	Note that the energy dissipation 
	bounds of Lemma~\ref{lem:energy-dissipation-main} 
	\textit{do not} require the game regularity condition
	of Assumption~\ref{ass:restricted-regularity}.
    Moreover, the bounds hold simultaneously
    for all points $x^\star \in \MEQ(V)$. 
\end{restatable}

\bigskip
\begin{restatable}[Constant stepsize]
	{rem}{remenergystepsize}
	\label{rem:stepsize-energy-dissipation}
	The tight dissipation bounds of 
	Lemma~\ref{lem:energy-dissipation-main}
	hold for a sufficiently small \textit{constant} 
	stepsize. In general, our proof does not attempt 
	to optimize the final stepsize constraint
	(and the resulting absolute constant dependencies
	in the lemma statement),
	and thus the result likely also holds under
	larger constant stepsize settings.
	Note that within the stepsize constraint, 
	the inverse dependencies on the
	Lipschitz and boundedness parameters
	of the operator are standard. 
	On the other hand, the constraint
	also depends inversely on the 
	ratio $\kappa_M = \tfrac{m_{\max}}{m_{\min}}$
	for the weights matrix $M$.
	A similar constraint is required for the
	asymptotic result of~\cite{legacci2024no},
	and we leave the design and analysis
	of adaptive stepsize methods 
	as an interesting direction for future work. 
\end{restatable}

\subsubsection*{Preliminaries and Recap of Setting}
\label{app:energy:prelims}

For convenience, we recall here some properties of
the symmetrized system and dynamics. 
Specifically, for a weights matrix
$M = \Diag(m_1 I_{d_1}, \dots, m_N I_{d_N}) \succ 0$,
we consider the $M$-symmetrized
dynamics from Proposition~\ref{prop:reduction-algorithmic}
using operator $V$, regularizer $R$, and energy $F = R^*$.
Recall that we define $\kappa_M = m_{\max}/m_{\min}$.

\paragraph{Definition of constants.}
Recall that the Lipschitz and boundedness constants 
for the operator $G$ (from Assumption~\ref{ass:operator})
are given by $L \ge 1$ and $B \ge 1$.
Moreover, the~\eqref{eq:lhs} constant
for the energy $H$ in Assumption~\ref{ass:lhs} 
is given by $\omega \ge 1$.
Then for the symmetrized system, we define the constants:
\begin{equation}
    \textstyle
    \wh L := \sqrt{\kappa_M} L \,,
    \;\;\;
    \wh B := \sqrt{m_{\max}} B \,,
    \;\;\;
    \wh \omega := \frac{1}{\sqrt{m_{\min}}} \omega  \,,
    \;\;\;
    \wh \tau := 3 \wh B \;.
    \label{eq:constant-helpers}
\end{equation}

\paragraph{Properties of symmetrized system.}
For the symmetrized system with operator $V$ 
regularizer $R$, recall that $V$ 
is $\wh L$-Lipschitz and $\wh B$-bounded (Setting~\ref{setting:weighted}):
\begin{equation}
	\|V(x') - V(x)\|_{*} \le 
	L \sqrt{\kappa_M} \cdot \|x' - x\|   \;,
	\;\;\;
	\|V(x) \|_{*} 
	\le \sqrt{m_{\max}} B 
	\;\;\text{for all $x, x' \in \calX$} \;.
	\label{eq:V-lip-bounded}
\end{equation}
Moreover, the energy function $F = R^*$ 
satisfies the~\eqref{eq:lhs}
with parameter $\wh \omega$.

\paragraph{Dual update rules.}
For convenience, we also restate here the dual update rules of
OFTRL and EGFTRL of the symmetrized system.
For all $t \ge 0$, we have:
\begin{align}
	&\quad y_{t+1}
	= 
	y_t - \eta 2 V(x_t) + \eta V(x_{t-1}) 
	\;,\quad
	x_{t+1} = \nabla F(y_{t+1}) 
	\label{eq:oftrl-helper}
	\tag{OFTRL} \\[0.25em]
	&\begin{cases}
	y_{t+1/2}	= y_t - \eta V(x_t) \;, 
	\quad &x_{t+1/2} = \nabla F(y_{t+1/2}) \\
	y_{t+1}	= y_t - \eta V(x_{t+1/2}) \;, 
	\quad &x_{t+1} = \nabla F(y_{t+1}) 
	\end{cases}
	\label{eq:egftrl-helper}
	\tag{EGFTRL}
\end{align}
Recall from Setting~\ref{setting:weighted} that
for both~\eqref{eq:oftrl-helper} and~\eqref{eq:egftrl-helper},
we have $\|y_{t+1} - y_t\|_{*} \le \eta \wh \tau$
for all $t \ge 0$
(and also $\|y_{t+1/2} - y_t\|_{*} \le \eta \wh \tau$
for~\eqref{eq:egftrl-helper} for all $t \ge 0$). 

\paragraph{Stepsize constraint.}
We assume the following bound 
on the stepsize $\eta$: 
\begin{equation}
	0 < \eta \le (880 \wh L \wh B \wh \omega)^{-1}
	= (880 L B \omega \kappa_M)^{-1} \;.
	\label{eq:eta-helper}
\end{equation}
Observe that~\eqref{eq:eta-helper}
corresponds to the condition 
from the statement of Lemma~\ref{lem:energy-dissipation-main}.

\subsection{Proof of Lemma~\ref{lem:energy-dissipation-main}}
\label{app:energy:main-lemma-proof}

We now give the proof 
of Lemma~\ref{lem:energy-dissipation-main}, 
which builds off the high-level overview
from Section~\ref{sec:proof-overview:energy}.

\paragraph{Scalar functions.}
For convenience, we first define the scalar functions 
$\phi(\cdot)$,
$\mu(\cdot)$, $\nu(\cdot)$, 
and $\gamma(\cdot)$:
\begin{equation}
    \textstyle
    \phi(\alpha)
	= \exp(\wh \omega \alpha) \,,
    \;
	\mu(\alpha)
	= 
	\frac{\exp(\wh\omega \alpha)-(\wh \omega \alpha) - 1}{(\wh \omega \alpha)^2} \,, 
	\;
	\nu(\alpha)
	= 
	\frac{\exp(\wh \omega \alpha) - 1 - \wh \omega \alpha}{\wh \omega \alpha} \,,
	\;
	\gamma(\alpha)
	= 
	\frac{\exp(\wh\omega \alpha) -1}{\wh \omega \alpha} \,.
    \label{eq:scalar-helpers}
\end{equation}
Each of these scalar functions corresponds 
to a multiplicative factor arising from
the bounds developed in Section~\ref{app:error-machinery},
which hold for the energy $F$
satisfying~\eqref{eq:lhs} with parameter $\wh \omega$.
In particular, each function in~\eqref{eq:scalar-helpers}
is non-negative
and increasing in $\alpha$.
Throughout the proof of the lemma, 
we will usually apply the functions of~\eqref{eq:scalar-helpers}
with $\alpha = \eta \wh \tau$, 
which stems from the upper bound 
$\|y_{t+1} - y_t\|_{*} \le \eta \wh \tau$
that holds for the~\eqref{eq:oftrl-helper} 
and~\eqref{eq:egftrl-helper} iterates.

\medskip
\paragraph{1. Coupling with implicit step.} ~ 

\noindent
The first step of the proof involves setting up
a general coupling between a sequence
of FTRL-type iterates and 
an \textit{implicit} update step:

\smallskip
\propenergycouplingmeta*

To streamline the presentation, we give the
proof of the proposition further below in 
Section~\ref{app:energy:coupling-proof}.
It thus remains to instantiate
Proposition~\ref{prop:meta-coupling}
with the~\eqref{eq:oftrl-helper} and~\eqref{eq:egftrl-helper} iterates
and to control at each step
the Bregman divergence $\D_{F}(y_t, y_{t+1})$
and the coupling error term
$C_t := \langle x_{t+1} - x^\star, y_{t+1} - \wh y_{t+1} \rangle$.
We describe these steps below.

\medskip
\paragraph{2. Bounds on Bregman divergence terms.} ~ 

\noindent
Suppose that $\|y_{t+1} - y_t \|_{*} \le \alpha$
for some $\alpha >0$. 
First, we apply the two-sided bounds on 
the Bregman divergence
$\D_F$ from Lemma~\ref{lem:bregman-local-bound}.
Together, with an application of
Lemma~\ref{lem:local-transfer}, which transfers
the local seminorm to the iterate $y_t$, 
we obtain the inequalities:
\begin{equation}
	\mu(-\alpha) \phi(-\alpha) \cdot 
	\|y_{t+1} - y_t\|^2_{y_t}
	\;\le\;
	\D_F(y_t, y_{t+1})
	\;\le\;
	\mu(\alpha) \phi(\alpha) 
	\cdot \|y_{t+1} - y_t\|^2_{y_t} \;,
	\label{eq:energy-bregman-term}
\end{equation}
where $\mu(\cdot)$ and $\phi(\cdot)$
are as defined in~\eqref{eq:scalar-helpers}.
To instantiate~\eqref{eq:energy-bregman-term} 
with the~\eqref{eq:oftrl-helper} and~\eqref{eq:egftrl-helper}
iterates, we thus require two-sided
bounds on the magnitude of increments 
$\|y_{t+1} -y_t\|_{y_t}$ under the local semi-norm.
For this, we prove, under the constraint on 
$\eta$ in~\eqref{eq:eta-helper},
the following inequalities 
that relate the increments $y_{t+1} -y_t$
to the operator vector $V(x_t)$ in local semi-norm:

\smallskip
\noindent
\textbf{Proposition~\ref{prop:oftrl-dual-increment-op}.}\,
\textit{For~\eqref{eq:oftrl-helper}}:
$\frac{6}{7} \eta  \|V(x_t)\|_{y_t} 
	\le \|y_{t+1} - y_t \|_{y_t}
	\le \frac{6}{5} \eta \|V(x_t)\|_{y_t}$\;
\textit{for all $t \ge 1$.}

\noindent
\textbf{Proposition~\ref{prop:egftrl-dual-increment-op}.}\,
\textit{For~\eqref{eq:egftrl-helper}}:
$\frac{5}{6} \eta  \|V(x_t)\|_{y_t} 
	\le \|y_{t+1} - y_t \|_{y_t}
	\le \frac{7}{6} \eta \|V(x_t)\|_{y_t}$\;
\textit{for all $t \ge 1$.}

\medskip
We give the proofs of these propositions
in Sections~\ref{app:energy:oftrl:dual-increment}
and~\ref{app:energy:egftrl:dual-increment}, respectively.
Then, recall that 
$\|y_{t+1} - y_t\|_{*} \le \eta \wh \tau$
for all $t \ge 1$ under both~\eqref{eq:oftrl-helper}
and~\eqref{eq:egftrl-helper}
(where $\wh \tau = 3\wh B$, as in~\eqref{eq:constant-helpers}),
meaning that we can set $\alpha = \eta \wh \tau$.
Substituting these statements into~\eqref{eq:energy-bregman-term} 
yields the following bounds:
\begin{itemize}[
	leftmargin=1.5em,
	itemsep=0.5em,
]
	\item
	For~\eqref{eq:oftrl-helper}, 
	we obtain for all $t \ge 1$:
	\begin{equation}
		\textstyle
		\frac{36}{49}
		\mu(-\eta \wh \tau)
		\phi(-\eta \wh \tau)
		\cdot  \eta^2
		\|V(x_t)\|^2_{y_t} 
		\le
		\D_F(y_t, y_{t+1})
		\le
		\frac{36}{25}
		\mu(\eta \wh \tau)
		\phi(\eta \wh \tau)
		\cdot \eta^2 
		\|V(x_t)\|^2_{y_t} \;.
		\label{eq:oftrl-bregman-term} 
	\end{equation}

	\item
	For~\eqref{eq:egftrl-helper}, 
	we obtain for all $t \ge 1$:
	\begin{equation}
		\textstyle
		\frac{25}{36}
		\mu(-\eta \wh \tau)
		\phi(-\eta \wh \tau)
		\cdot  \eta^2
		\|V(x_t)\|^2_{y_t} 
		\le
		\D_F(y_t, y_{t+1})
		\le
		\frac{49}{36}
		\mu(\eta \wh \tau)
		\phi(\eta \wh \tau)
		\cdot \eta^2 
		\|V(x_t)\|^2_{y_t} \;.
		\label{eq:egftrl-bregman-term}
	\end{equation}
\end{itemize}

Now, the product of scalar functions
$\mu(-\eta \wh \tau)\cdot\phi(- \eta \wh \tau)$
is decreasing with $\eta$ and 
can be made a sufficiently large positive constant 
bounded above by $1/2$ when
$\eta$ is sufficiently small.
Similarly, $\mu(\eta \wh \tau)\cdot\phi(\eta \wh \tau)$
is increasing with $\eta$ and be made 
arbitrarily close to $1/2$ from above
for sufficiently small $\eta$. 
In particular, for $\eta$ as in~\eqref{eq:eta-helper},
it is straightforward to check that
\begin{equation}
    \textstyle
    \mu(- \eta \wh \tau) \cdot \phi(- \eta \wh \tau) 
    \ge \frac{49}{100}
    \;\;\text{and}\;\;
    \mu(\eta \wh \tau) \cdot \phi(\eta \wh\tau) 
    \le 
    \frac{51}{100} \;.
    \label{eq:df-constants}
\end{equation}

\smallskip
\paragraph{3. Bounds on coupling error terms.} ~

\noindent
Recall that $C_t = \langle x_{t+1} - x^\star, y_{t+1} - \wh y_{t+1} \rangle$
denotes the coupling error term at each step $t \ge 1$.
For the iterates of~\eqref{eq:oftrl-helper} and
\eqref{eq:egftrl-helper}, we obtain the following bounds
for $\eta$ as in~\eqref{eq:eta-helper}:

\smallskip
\noindent
\textbf{Proposition~\ref{prop:oftrl-dual-ct-bound}}\;
\textit{For~\eqref{eq:oftrl-helper}, for all $x^\star \in \MEQ(V)$ and $t \ge 1$:}
\begin{equation}
	\textstyle
	|C_t|
	\le 
	\frac{24}{5}
	\cdot \nu(\eta \wh \tau) \cdot \eta^2 \|V(x_t)\|^2_{y_t}
	+ 
	110 \wh L \cdot \eta^3 \|V(x_t)\|^2_{y_t} \;.
\end{equation}

\smallskip
\noindent
\textbf{Proposition~\ref{prop:egftrl-dual-ct-bound}}\;
\textit{For~\eqref{eq:egftrl-helper}, 
for all $x^\star \in \MEQ(V)$ and $t \ge 0$}:
\begin{equation}
	\textstyle
	|C_t|
	\le 
	9 \wh L \cdot \eta^3 \|V(x_t)\|^2_{y_t} \;.
\end{equation}

We give the proofs of Propositions~\ref{prop:oftrl-dual-ct-bound}
and~\ref{prop:egftrl-dual-ct-bound}
in Sections~\ref{app:energy:oftrl:ct}
and~\ref{app:energy:egftrl:ct}, respectively.
Similar to the prior step, we note that 
$\nu(\eta \wh \tau)$ is non-negative and strictly increasing in $\eta$
and can be made an arbitrarily small positive
constant for $\eta$ sufficiently small.
Moreover, the $O(\eta^3)$ terms in 
Propositions~\ref{prop:oftrl-dual-ct-bound}
and~\ref{prop:egftrl-dual-ct-bound} can also be made
small positive constants for $\eta$ sufficiently small.
In particular, for $\eta$ as in~\eqref{eq:eta-helper},
we have the bounds:
\begin{equation}
    \textstyle
    \nu(\eta \wh \tau) \le \frac{1}{500} 
    \;\;\text{and}\;\;
    \eta \cdot 110 \wh L \le \frac{1}{8} \;.
    \label{eq:ct-constants}
\end{equation}

\paragraph{Combining the pieces: upper bound.}
To derive the upper bound of Lemma~\ref{lem:energy-dissipation-main},
fix $x^\star \in \MEQ(V)$ and $t \ge 1$.
Then we combine the lower bounds on $\D_F(y_t, y_{t+1})$ from
\eqref{eq:oftrl-bregman-term} and~\eqref{eq:egftrl-bregman-term}
and the upper bounds on $|C_t|$ from 
Proposition~\ref{prop:oftrl-dual-ct-bound}
and~\ref{prop:egftrl-dual-ct-bound}. 
Taking the worse bound of the two instances,
this yields for both~\eqref{eq:oftrl-helper} 
and~\eqref{eq:egftrl-helper}:
\begin{align*}
    F_{x^\star}(y_{t+1}) 
    - 
    F_{x^\star}(y_t)
    &= 
    -\D_F(y_{t},y_{t+1})
    + C_t
    \nonumber \\
    &\le 
    - \big(\tfrac{25}{36} \mu(-\eta \wh \tau) \phi(-\eta \wh \tau)
        -  \tfrac{24}{5} \nu(\eta \wh \tau)\big)
    \cdot \eta^2 \|V(x_t)\|^2_{y_t}
    + 110 \wh L \cdot \eta^3 \|V(x_t)\|^2_{y_t} \\
    &\le 
    - \tfrac{1}{4} \cdot \eta^2 \|V(x_t)\|^2_{y_t}
    + \tfrac{1}{8} \cdot \eta^2 \|V(x_t)\|^2_{y_t} \\
    &= 
    - \tfrac{1}{8} \cdot \eta^2 \|V(x_t)\|^2_{y_t} \;.
\end{align*}
Here, the penultimate inequality comes from substituting
the lower bound on $\mu(- \eta \wh \tau) \cdot \phi(- \eta \tau)$
from~\eqref{eq:df-constants} and 
the upper bounds of~\eqref{eq:ct-constants},
both of which hold for $\eta$ as in~\eqref{eq:eta-helper}.
Thus, for such sufficiently small constant $\eta$, 
the negative second order $\D_{F}(y_t, y_{t+1})$ term 
dominates the coupling error and establishes
a dissipative one-step upper bound on $F_{x^\star}$. 

\smallskip
\paragraph{Combining the pieces: lower bound.}
For the lower bound of Lemma~\ref{lem:energy-dissipation-main},
we now use the upper bounds on the $\D_F(y_t, y_{t+1})$ terms
from~\eqref{eq:oftrl-bregman-term} and~\eqref{eq:egftrl-bregman-term}.
Then for any $x^\star \in \MEQ(V)$ 
and $t \ge 1$, we similarly find for both~\eqref{eq:oftrl-helper}
and~\eqref{eq:egftrl-helper} that:
\begin{align*}
    F_{x^\star}(y_{t+1}) 
    - 
    F_{x^\star}(y_t)
    &= 
    -\D_F(y_{t},y_{t+1})
    + C_t \\
    &\ge
    - \D_F(y_t, y_{t+1}) 
    - |C_t| \\
    &\ge
    - \big(\tfrac{36}{25} \mu(\eta \wh \tau) \phi(\eta \wh \tau)
        + \tfrac{24}{5}\nu(\eta \wh \tau) \big) \cdot 
        \eta^2 \|V(x_t)\|^2_{y_t}
    - 110 \wh L \cdot \eta^3 \|V(x_t)\|^2_{y_t} \\
    &\ge
    - \tfrac{3}{4} \cdot \eta^2 \|V(x_t)\|^2_{y_t}
    - \tfrac{1}{8} \cdot \eta^2 \|V(x_t)\|^2_{y_t} \\ 
    &= 
    - \tfrac{7}{8} \cdot \eta^2 \|V(x_t)\|^2_{y_t} \;.
\end{align*}
Here, the penultimate inequality now follows from 
the upper bound on $\mu(\eta \wh \tau) \cdot \phi(\eta \wh \tau)$
from~\eqref{eq:df-constants} 
(and again using~\eqref{eq:ct-constants}), 
which holds for $\eta$ as in~\eqref{eq:eta-helper}.
This completes the proof of the lemma.
\hfill $\blacksquare$

\subsubsection{Proof of Proposition~\ref{prop:meta-coupling}}
\label{app:energy:coupling-proof}

Here we give the proof of the 
implicit-step coupling of
Proposition~\ref{prop:meta-coupling}:

\begin{proof}
The proof relies on two steps: (1) expressing
the difference in shifted energy in terms 
of the (unshifted) Bregman divergence $\D_F(y_t, y_{t+1})$
and (2) coupling the true iterate $y_{t+1}$ 
with the implicit-step  iterate $\wh y_{t+1}$. 
We proceed with these two steps:

\smallskip
\noindent
\textit{Step (1).} 
Fix $x^\star \in \MEQ(V)$ and $t \ge 1$, 
and observe for any $t \ge 1$ that
\begin{align}
	F_{x^\star}(y_{t+1}) - F_{x^\star}(y_t)
	&= 
		F(y_{t+1}) - F(y_t) - \langle x^{\star}, y_{t+1} - y_t \rangle
		\nonumber \\
	&= 
		\langle \nabla F(y_{t+1}), y_{t+1} - y_{t} \rangle
		- \D_F(y_t, y_{t+1}) 
		-  \langle x^{\star}, y_{t+1} - y_t \rangle 
		\nonumber \\
	&= 
		\langle x_{t+1}, y_{t+1} - y_t \rangle
		- \D_F(y_t, y_{t+1}) 
		-  \langle x^{\star}, y_{t+1} - y_t \rangle 
		\nonumber \\
	&= 
		- \D_F(y_t, y_{t+1}) + 
		\langle x_{t+1} - x^\star, y_{t+1} - y_t \rangle \;.
    \label{eq:one-step-energy-shift1}
\end{align}
Here, the second equality comes
from the definition of $\D_F(y_t, y_{t+1})$
and the third equality comes from
the assumption that $x_{t+1} = \nabla F(y_{t+1})$
for all $t \ge 0$.

\smallskip
\noindent
\textit{Step (2).}
To further simplify~\eqref{eq:one-step-energy-shift1},
we use the implicit update $\wh y_{t+1} = y_t - \eta V(x_{t+1})$
to write:
\begin{equation}
y_{t+1} - y_t 
= \wh y_{t+1} - y_t + y_{t+1} - \wh y_{t+1}
= - \eta V(x_{t+1})  + y_{t+1} - \wh y_{t+1} \;.
\end{equation}
Then it follows that
\begin{align}
	\langle x_{t+1} - x^\star, y_{t+1} - y_t \rangle 
	&= 
	- \eta \langle x_{t+1} - x^\star, V(x_{t+1}) \rangle 
	+ \langle x_{t+1} - x^\star, y_{t+1} - \wh y_{t+1} \rangle 
	\nonumber \\
	&= 
	\langle x_{t+1} - x^\star, y_{t+1} - \wh y_{t+1} \rangle \,.
	\label{eq:uee02}
\end{align}
Here the last equality follows 
due to the fact that $x^\star \in \MEQ(V)$,
and thus $\langle x_{t+1}- x^\star, V(x_{t+1})\rangle = 0$.
Substituting~\eqref{eq:uee02} 
into~\eqref{eq:one-step-energy-shift1} 
then yields the desired statement.
\end{proof}

\subsection{Instantiating Implicit Step Coupling for OFTRL}
\label{app:energy:oftrl}

In this section, we provide details on the instantiation
of Lemma~\ref{lem:energy-dissipation-main} 
along~\eqref{eq:oftrl-helper}, as outlined in 
the proof in Section~\ref{app:energy:main-lemma-proof}.
This requires proving two main propositions,
which we state here.

First, the following proposition relates 
the dual increments $y_{t+1} - y_t$
to the operator value $V(x_t)$:

\smallskip
\begin{restatable}
	{prop}{propoftrldualincrement}
	\label{prop:oftrl-dual-increment-op}
	For~\eqref{eq:oftrl-helper}
	with $\eta > 0$ as in~\eqref{eq:eta-helper}, 
	it holds for all $t \ge 1$:
	\begin{equation*}
		\textstyle
		\frac{6}{7} \eta  \|V(x_t)\|_{y_t} 
		\le \|y_{t+1} - y_t \|_{y_t}
		\le \frac{6}{5} \eta \|V(x_t)\|_{y_t} \;.
	\end{equation*}
\end{restatable}

Second, we establish bounds on the coupling
error terms
$C_t := \langle x_{t+1} -x^\star, y_{t+1} - \wh y_{t+1}\rangle$
that hold for all $x^\star \in \MEQ(V)$.
Formally we prove the following:

\smallskip
\begin{restatable}
	{prop}{propoftrlctbound}
	\label{prop:oftrl-dual-ct-bound}
	For~\eqref{eq:oftrl-helper}
	with $\eta > 0$ as in~\eqref{eq:eta-helper}, 
	it holds for all $t \ge 1$
    and $x^\star \in \MEQ(V)$:
	\begin{equation*}
        \textstyle
        |C_t|
        \le 
        \frac{24}{5} \nu(\eta \wh\tau) 
        \cdot \eta^2 \|V(x_t)\|^2_{y_t}
        + 
        110 \wh L \cdot \eta^3 \|V(x_t)\|^2_{y_t} \;.
	\end{equation*}
\end{restatable}

The proofs of Proposition~\ref{prop:oftrl-dual-increment-op}
and~Proposition~\ref{prop:oftrl-dual-ct-bound}
are given below in Section~\ref{app:energy:oftrl:dual-increment}
and Section~\ref{app:energy:oftrl:ct}, respectively.
Both proofs additionally use the following helper lemma, which
establishes the \textit{stability}
of the OFTRL dual iterates in local seminorm.
Formally:

\smallskip
\begin{restatable}[OFTRL dual stability]{lem}{lemoftrlstability}
	\label{lem:oftrl-dual-stability}
	Under~\eqref{eq:oftrl-helper} with $\eta$ 
    satisfying~\eqref{eq:eta-helper},
    for all $t \ge 2$:
	\begin{equation*}
		\textstyle
		\frac{1}{2} \cdot 
		\|y_{t-1} - y_{t-2}\|_{y_t} 
		\;\le\;
		\|y_t - y_{t-1} \|_{y_t}
		\;\le\;
		3 \cdot \|y_{t+1} - y_t\|_{y_t} \;.
	\end{equation*}
\end{restatable}
We will grant this result as true for now
and defer its proof to 
Section~\ref{app:energy:oftrl:stability-proof}.

\subsubsection{Proof of Proposition~\ref{prop:oftrl-dual-increment-op}}
\label{app:energy:oftrl:dual-increment}

This section gives the proof of 
Proposition~\ref{prop:oftrl-dual-increment-op}, 
which relates the magnitude of increments
$y_{t+1} - y_t$ to the magnitude of the
operator vector $V(x_t)$ under the local
semi-norm $\|\cdot\|_{y_t}$. 

\begin{proof}
    First, recall the 
    scalar function $\gamma(\cdot)$ from~\eqref{eq:scalar-helpers}
    and the constant $\wh \tau$ from~\eqref{eq:constant-helpers}.
    In particular, recall for~\eqref{eq:oftrl-helper}
    that $\|y_{t}-y_{t-1}\|_* \le \eta \wh \tau$ 
    for all $t \ge 2$.
    Moreover, it is straightforward to check that 
    for $\eta$ as in~\eqref{eq:eta-helper},
    we have $3 \eta \wh L \gamma(\eta \wh \tau) \le \frac{1}{6}$. 
	We then prove the lower and upper inequalities
	of the proposition separately:
    
	\paragraph{Upper inequality.}
	By definition of~\eqref{eq:oftrl-helper} and using the 
	triangle inequality with respect to the local norm $\|\cdot\|_{y_t}$, we obtain the bound
	\begin{alignat*}{2}
		\|y_{t+1} - y_t \|_{y_t}
		&\le
		\eta \|V(x_t)\|_{y_t}
		+ 
		\eta \|V(x_t) - V(x_{t-1})\|_{y_t}  \\
		&\le
		\eta \|V(x_t)\|_{y_t} + 
		\eta \wh L \gamma(\eta \wh \tau) \cdot \| y_t - y_{t-1}\|_{y_t}
		&&\quad\text{(Lemma~\ref{lem:op-to-dual-diff})} \\
		&\le 
		\eta \|V(x_t)\|_{y_t} + 
		3 \eta \wh L \gamma(\eta \wh \tau) \cdot \| y_{t+1} - y_t \|_{y_t} 
		&&\quad\text{(Lemma~\ref{lem:oftrl-dual-stability})} \;.
	\end{alignat*}
	Rearranging and using the fact that 
	$3 \eta \wh L \gamma(\eta \wh \tau)  < 1/6$, we conclude that
	\begin{equation*}
        \textstyle
		\|y_{t+1} - y_t \|_{y_t}
		\le 
		\frac{1}{1-3\eta \wh L \gamma(\eta \wh \tau)} \cdot 
		\eta \|V(x_t)\|_{y_t} 
		\le
		\frac{6}{5} \cdot \eta \|V(x_t)\|_{y_t} \;.
	\end{equation*}

	\paragraph{Lower inequality.}
	For the lower inequality, we again use the definition
	of~\eqref{eq:oftrl-helper} and the triangle inequality to write
	\begin{alignat*}{2}
		\eta \|V(x_t)\|_{y_t} 
		&\le 
		\|y_t - y_{t+1}\|_{y_t}
		+ 
		\eta \|V(x_{t}) - V(x_{t-1})\|_{y_t} \\
		&\le 
		\|y_t - y_{t+1}\|_{y_t}
		+ 
		\wh L \gamma(\eta \wh \tau) \eta \cdot \|y_t - y_{t-1}\|_{y_t} 
		&\quad\text{(Lemma~\ref{lem:op-to-dual-diff})} \\
		&\le 
		\|y_t - y_{t+1}\|_{y_t}
		+ 
		3 \wh L \gamma(\eta \wh \tau) \eta \cdot  \|y_{t+1} - y_{t}\|_{y_t} 
		&\quad\text{(Lemma~\ref{lem:oftrl-dual-stability})} \\
		&\le
		\tfrac{7}{6} \cdot \|y_{t+1} - y_t\|_{y_t}
		&\quad\text{($3 \eta \wh L \gamma(\eta \wh \tau) \le 1/6$)} \;.
	\end{alignat*}
	Rearranging then yields the desired lower inequality and 
    concludes the proof.
\end{proof}

\subsubsection{Proof of Proposition~\ref{prop:oftrl-dual-ct-bound}}
\label{app:energy:oftrl:ct}

This section gives the proof of 
Proposition~\ref{prop:oftrl-dual-ct-bound},
which establishes a bound on the 
coupling error terms $C_t$ 
from Proposition~\ref{prop:meta-coupling}
along the~\eqref{eq:oftrl-helper} iterates:

\begin{proof}
The proof proceeds in two steps:
first, deriving an exact expansion of the
terms $C_t$ into four constituent \textit{error terms}, 
and second, deriving bounds on the absolute magnitudes of
the error terms.

\paragraph{1. Expanding the inner product term.} ~ 

First, recall that $\wh y_{t+1} = y_t - \eta V(x_{t+1})$. 
Then under~\eqref{eq:oftrl-helper}, it follows that
we can express the difference $y_{t+1} - \wh y_{t+1}$ as
\begin{align*}
	y_{t+1} - \wh y_{t+1}
	&= 
	y_t - 2\eta V(x_t) + \eta V(x_{t-1})
	- (y_t - \eta V(x_{t+1})) \\
	&= 
	\eta (V(x_{t+1}) - 2 V(x_t)
    +V(x_{t-1}))  \;.
\end{align*}
Now fix $x^\star$ satisfying~\eqref{eq:minty-equality}.
Then we can expand $C_t$ as follows:
\begin{align}
	C_t := 
	\langle 
	x_{t+1} - x^\star, y_{t+1} - \wh y_{t+1}
	\rangle
	&= 
	\eta \big\langle 
	x_{t+1} - x^\star, 
	V(x_{t+1}) - 2 V(x_t) + V(x_{t-1})
	\big\rangle \nonumber \\
	&= 
	- 2
	\eta \big \langle
		x_{t+1} - x^\star, V(x_t)
	\big\rangle
	+ 
	\eta \big\langle
		x_{t+1} - x^\star, V(x_{t-1}) 
	\big\rangle\,,
	\label{eq:oi0}
\end{align}
where the final equality follows from observing that 
$\langle x_{t+1} - x^\star, V(x_{t+1}) \rangle = 0$ since 
$x^\star$ satisfies~\eqref{eq:minty-equality}. 
Both terms in~\eqref{eq:oi0} can then further be expanded as:
\begin{align}
	- 2\eta \big\langle x_{t+1} - x^\star, V(x_t) \big\rangle
	&= 
	- 2\eta \big\langle
		x_t - x^\star + (x_{t+1} - x_t), V(x_t)
	\big\rangle  
		\nonumber \\
	&= 
	- 2\eta \big\langle
		x_{t+1} - x_t, V(x_t)  
	\big\rangle \,, 
		\label{eq:oiet1} \\[0.5em]
	\eta \big\langle x_{t+1} - x^\star, V(x_{t-1}) \big\rangle 
	&= 
	\eta\big \langle 
		x_{t-1} - x^\star + (x_{t+1} - x_{t-1}) , V(x_{t-1}) 
	\big\rangle 
		\nonumber \\
	&= 
	\eta \big\langle
		x_{t+1} - x_{t-1}, V(x_{t-1})
	\big\rangle 
		\nonumber \\
	&= 
	\eta \big\langle
		x_{t+1} - x_{t-1}, V(x_{t-1}) - V(x_t)
	\big\rangle
	+ 
	\eta \big\langle
		x_{t+1} - x_{t-1}, V(x_t)
	\big\rangle .
		\label{eq:oiet2}
\end{align}
In the expansion of both terms, the second equalities
follow from the fact that $x^\star$ satisfies~\eqref{eq:minty-equality}.
Substituting~\eqref{eq:oiet1} and~\eqref{eq:oiet2} into~\eqref{eq:oi0}
and simplifying then yields
\begin{equation}
	C_t 
	= 
	\underbrace{\eta \big \langle 
		V(x_t), 
		x_t - x_{t-1} + x_t - x_{t+1}
	\big\rangle}_{\text{Term I}}
	+ 
	\eta \big \langle
		x_{t+1} - x_{t-1}, V(x_{t-1}) - V(x_t)
	\big\rangle\;.
	\label{eq:oi1}
\end{equation}
We then further simplify the first term of~\eqref{eq:oi1}
as follows:

\smallskip
\noindent
\textit{Expansion of Term I.}
First recall the primal-dual relationships
$x_{t-1} = \nabla F(y_{t-1})$, $x_t = \nabla F(y_t)$,
and $x_{t+1} = \nabla F(y_{t+1})$. 
Using the definition of the Hessian remainder terms
$\G_F(y_{t+1}, y_t)$ and $\G_F(y_{t-1}, y_t)$, it follows that
\begin{align}
	x_t - x_{t-1} &+ x_t - x_{t+1}  
    \nonumber \\
	&= 
	\nabla F(y_t) - \nabla F(y_{t-1}) 
	+ 
	\nabla F(y_t) - \nabla F(y_{t+1}) 
		\nonumber \\
	&= 
		\nabla^2 F(y_t)(y_t - y_{t-1}) - \G_F(y_{t-1}, y_t)
		+ 
		\nabla^2 F(y_t)(y_t - y_{t+1}) - \G_F(y_{t+1}, y_t)
		\nonumber \\
	&= 
		\nabla^2 F(y_t)	(y_t - y_{t-1} + y_t - y_{t+1})
		- \G_F(y_{t-1}, y_t) - \G_F(y_{t+1}, y_t) \;.
		\label{eq:oit10}
\end{align}
By definition of~\eqref{eq:oftrl}, observe that 
\begin{align}
	y_t - y_{t-1} + y_t - y_{t+1} 
	&= 
	- 2 \eta V(x_{t-1}) + \eta V(x_{t-2})
	- \eta V(x_{t-1}) + 2\eta V(x_t)
		\nonumber \\
	&= 
	2 \eta \big(V(x_t) - V(x_{t-1})\big)
	+ \eta \big( V(x_{t-2}) - V(x_{t-1})\big) \;.
	\label{eq:oit11}
\end{align}
Substituting~\eqref{eq:oit11} into~\eqref{eq:oit10}
then allows for expanding (Term I) in~\eqref{eq:oi1} as
\begin{align}
	\text{(Term I)} 
	= 
	2 \eta^2 \big\langle 
		V(x_t), V(x_t) - V(x_{t-1})
	\big\rangle_{y_t}
	&+ 
	\eta^2 \big \langle 
		V(x_t), V(x_{t-2}) - V(x_{t-1})
	\big\rangle_{y_t} 
		\nonumber \\ 
	&- \eta \big\langle
		V(x_t),
		\G_F(y_{t+1}, y_t)
		+ \G_F(y_{t-1}, y_t) 
	\big\rangle  \;.
	\label{eq:oit1expand}
\end{align}

\noindent
\textit{Final expansion of $C_t$.}
Finally, 
substituting~\eqref{eq:oit1expand} into~\eqref{eq:oi1}
yields the full expansion
\begin{equation}
	\begin{aligned}
	C_t 
	&= 
		\underbrace{- \eta \big\langle 
			V(x_t), \G_F(y_{t+1}, y_t) + \G_F(y_{t-1}, y_t)
		\big\rangle}_{\calE_{t,1}} \\
	&\qquad
		+  \underbrace{2 \eta^2 \big\langle 
			V(x_t), V(x_t) - V(x_{t-1})
		\big\rangle_{y_t}}_{\calE_{t, 2}}
		+ \underbrace{
			\eta^2 \big \langle 
				V(x_t), V(x_{t-2}) - V(x_{t-1})
			\big\rangle_{y_t}
		}_{\calE_{t,3}} \\
	&\qquad
		+ \underbrace{
			\eta \big\langle 
				x_{t+1} - x_{t-1}, V(x_{t-1}) - V(x_t)
			\big\rangle
		}_{\calE_{t, 4}} \;.
	\end{aligned}
	\label{eq:oftrl-error-terms}
\end{equation}

\paragraph{2. Bounds on magnitude of coupling error terms.} ~ 

\noindent
We  proceed to derive upper bounds on the magnitude
of the error terms $\calE_{t, 1}$, $\calE_{t, 2}$, 
$\calE_{t, 3}$, and $\calE_{t, 4}$ that comprise the 
exact expansion of the inner  product 
$C_t = \langle x_{t+1} - x^\star, y_{t+1} - \wh y_{t+1} \rangle$
from~\eqref{eq:oftrl-error-terms}.
We assume $\eta$ is constrained as in~\eqref{eq:eta-helper} 
and we also use the definitions of 
the scalar functions $\nu(\cdot)$ and $\gamma(\cdot)$
from~\eqref{eq:scalar-helpers}.

\smallskip
\paragraph{Bound on second-order error terms.}
The first term $\calE_{t, 1}$ scales like $O(\eta^2 \|V(x_t)\|^2_{y_t})$,
and with a leading constant that can be made arbitrarily
small for $\eta$ sufficiently bounded. In particular, 
by an application of the generalized Cauchy-Schwarz
inequality with respect to the local norm $\|\cdot\|_{y_t}$, we can 
write and further simplify
\begin{alignat}{2}
	|\calE_{t, 1}|
	&= 
	\big|\eta \big\langle 
			V(x_t), \G_F(y_{t+1}, y_t) + G_F(y_{t-1}, y_t)
		\big\rangle \big|
		\nonumber  \\
	&\le 
	\eta \|V(x_t)\|_{y_t} \cdot 
	\|\G_F(y_{t+1}, y_t) + \G_F(y_{t-1}, y_t) \|_{y_t, *} 
		\nonumber \\
	&\le 
	\eta \|V(x_t)\|_{y_t} \cdot 
	\big(
		\|\G_F(y_{t+1}, y_t)\|_{y_t, *}
		+ 
		\|\G_F(y_{t-1}, y_{t})\|_{y_t, *}
	\big) 
		\nonumber \\
	&\le
	\eta \|V(x_t)\|_{y_t} \cdot 
	\nu(\eta \wh \tau) \cdot \big(
		\|y_{t+1}-y_t\|_{y_t}
		+ 
		\|y_{t}-y_{t-1}\|_{y_t}
	\big)
	&&\qquad\text{(Lemma~\ref{lem:hessian-remainder-bound})} 
		\nonumber \\
	&\le
		\eta \|V(x_t)\|_{y_t} \cdot 
		4\nu(\eta \wh \tau) \cdot 
		\|y_{t+1}-y_t\|_{y_t}
        &&\qquad\text{(Lemma~\ref{lem:oftrl-dual-stability})} 
		\nonumber \\
	&\le
	 \frac{24}{5} \nu(\eta \wh\tau) \cdot \eta^2 \|V(x_t)\|^2_{y_t} \;. 
	 \nonumber 
	 &&\qquad\text{(Proposition~\ref{prop:oftrl-dual-increment-op})} 
\end{alignat}

\paragraph{Bounds on third-order error terms.}
We now derive bounds on the magnitudes of 
$\calE_{t, 2}$, $\calE_{t, 3}$, and $\calE_{t,4}$
from~\eqref{eq:oftrl-error-terms}.
In particular, we show each of these three terms scales 
like $O(\eta^3 \|V(x_t)\|^2_{y_t})$.

For the term $\calE_{t, 2}$, we can apply
the Cauchy-Schwarz inequality with respect to the local induced norm $\|\cdot\|_{y_t}$ and Lemma~\ref{lem:op-to-dual-diff}
and further simplify to find
\begin{alignat*}{2}
	|\calE_{t, 2}|
	&= 
	\big| 2 \eta^2 \big \langle 
		V(x_t), V(x_t) - V(x_{t-1})
	\big\rangle_{y_t} \big| \\
	&\le
	2\eta^2 
	\|V(x_t)\|_{y_t}
	\cdot 
	\|V(x_t) - V(x_{t-1})\|_{y_t} \\
	&\le 
	2 \eta^2 \|V(x_t)\|_{y_t}
	\cdot 
	\wh L \gamma(\eta \wh \tau) \cdot \|y_t - y_{t-1}\|_{y_t} 
	&&\quad\text{(Lemma~\ref{lem:op-to-dual-diff})} \\
	&\le 
	2 \eta^2 \|V(x_t)\|_{y_t}
	\cdot 
	3 \wh L \gamma(\eta \wh \tau) \cdot \|y_{t+1} - y_{t}\|_{y_t} 
	&&\quad\text{(Lemma~\ref{lem:oftrl-dual-stability})} \\
	&\le
	\frac{36}{5} \wh L \gamma(\eta \wh \tau) \cdot \eta^3 \|V(x_t)\|^2_{y_t}  \;.
	&&\quad\text{(Proposition~\ref{prop:oftrl-dual-increment-op})} 
\end{alignat*}

For the term $\calE_{t, 3}$, we use similar calculations
as for $\calE_{t, 2}$, but now the application 
of Lemma~\ref{lem:oftrl-dual-stability}  gives
$\|y_{t-1}- y_{t-2}\|_{y_t} \le 6 \|y_{t+1}-y_t\|_{y_t}$.
Together with the fact that $\calE_{t, 3}$ does 
not contain leading multiplicative factor 2, 
and using two applications of Lemma~\ref{lem:local-transfer},
we obtain the bound
\begin{equation*}
    \textstyle
	|\calE_{t, 3}| 
	= 
	\big| 
	\eta^2 \big\langle 
	V(x_t), V(x_{t-2}) - V(x_{t-1})
	\big\rangle_{y_t}
	\big|
	\le 
	\frac{36}{5} \wh L \gamma(\eta \wh \tau) \phi(\eta \wh \tau)
	 \cdot \eta^3 \|V(x_t)\|^2_{y_t} \;.
\end{equation*}

For the term $\calE_{t, 4}$, we use generalized
Cauchy-Schwarz with respect to the local norm $\|\cdot\|_{y_t}$, together with applications
of Lemma~\ref{lem:dual-to-primal-local}
and Lemma~\ref{lem:op-to-dual-diff}, to obtain
\begin{alignat*}{2}
	|\calE_{t, 4}|
	&= 
	\big|
	\eta
	\big\langle 
	x_{t+1} - x_{t-1}, V(x_t) - V(x_{t-1})
	\big\rangle
	\big| \\
	&=
	\eta
	\big|
	\big\langle 
	x_{t+1} - x_t + x_t -  x_{t-1}, V(x_t) - V(x_{t-1})
	\big\rangle
	\big| \\
	&\le
	\eta \cdot 
	\big(
		\|x_{t+1} - x_{t}\|_{y_t, *}  + 
		\|x_t - x_{t-1}\|_{y_t, *}
	\big)	
	\cdot 
	\|V(x_t) - V(x_{t-1})\|_{y_t}
	&&\quad\text{(Generalized C-S)} \\
	&\le
	\eta \gamma(\eta \wh \tau) \cdot 
	\big(
		\|y_{t+1} - y_t\|_{y_t}
		+ \|y_t - y_{t-1}\|_{y_t}
	\big)
	\cdot 
	\wh L \gamma(\eta \wh \tau) \|y_{t} -y_{t-1}\|_{y_t} 
	&&\quad\text{(Lemmas~\ref{lem:dual-to-primal-local} 
		and~\ref{lem:op-to-dual-diff})} \\
	&\le
	12 \wh L \gamma(\eta \wh \tau)^2 \eta \cdot \|y_{t+1} -y_t\|^2_{y_t}
	&&\quad\text{(Lemma~\ref{lem:oftrl-dual-stability})} \\
	&\le 
	\frac{432}{5} \wh L \gamma(\eta \wh \tau)^2 \eta^3 \|V(x_t)\|^2_{y_t} \;.
	&&\quad\text{(Proposition~\ref{prop:oftrl-dual-increment-op})} 
\end{alignat*}

Combining the bounds on the third-order terms
$|\calE_{t, 2}|$, $|\calE_{t, 3}|$ and $|\calE_{t, 4}|$,
and also using the fact that $\gamma(\eta \wh \tau) \ge 1$ 
and $\phi(\eta \wh \tau) \ge 1$
for all $\eta >0$, 
we then  have that
\begin{equation*}
    \textstyle
	|\calE_{t, 2}| 
	+ 
	|\calE_{t, 3}| 
	+ 
	|\calE_{t, 4}|
	\le
	\frac{504}{5} \wh L \gamma(\eta \wh \tau)^2 \phi(\eta \wh \tau)
    \cdot \eta^3 \|V(x_t)\|^2_{y_t}
	\le
    110 \wh L \eta^3 \|V(x_t)\|^2_{y_t} \;.
\end{equation*}
Here, the final inequality comes from 
the fact that 
$\frac{504}{5} \cdot \gamma(\eta \wh \tau)^2 \phi(\eta \wh \tau) \le 110$
for $\eta$ as in~\eqref{eq:eta-helper}.

\paragraph{3. Combining the pieces.} ~

\noindent
Together, for $\eta$ as in~\eqref{eq:eta-helper}, 
we conclude that the magnitude of
$C_t$ is bounded above by
\begin{equation*}
    \textstyle
	|C_t| 
	\le 
	|\calE_{t, 1}|
	+ 
	|\calE_{t, 2}|
	+ |\calE_{t, 3}|
	+ |\calE_{t, 4}|
	\le
    \frac{24}{5} \nu(\eta \wh\tau) \cdot \eta^2 \|V(x_t)\|^2_{y_t}
	+ 110 \wh L \cdot \eta^3 \|V(x_t)\|^2_{y_t} \;,
\end{equation*}
which completes the proof of the proposition.
\end{proof}

\subsubsection{Proof of Lemma~\ref{lem:oftrl-dual-stability}}
\label{app:energy:oftrl:stability-proof}

Here, we prove Lemma~\ref{lem:oftrl-dual-stability}, 
which establishes the stability of 
the dual iterates of~\eqref{eq:oftrl}.
Note that this lemma is the analog 
of~\citet[Proposition F.6]{lazarsfeld2026};
but the present proof is generalized 
to handle general non-linear operators. 

\smallskip
\begin{proof}
For readability, we will use the following shorthand notation:
for $t \ge 1$, define
\begin{equation*}
	\Delta_{t} := y_t - y_{t-1} = 
	\begin{cases}
		\; - \eta V(x_{0}) &\text{for $t=1$} \\
		\; - 2 \eta V(x_{t-1}) + \eta V(x_{t-2}) 
	 	&\text{for $t \ge 2$} \;.
 	\end{cases}
\end{equation*}
Thus, our goal is to prove for all $t \ge 2$ that
\begin{equation}
	\text{(i)}\;\;	
	\|\Delta_{t-1}\|_{y_t}  \le 2 \|\Delta_t\|_{y_t} 
	\quad\text{and}\quad
	\text{(ii)}\;\;
	\|\Delta_t\|_{y_t}
	\le 
	3 \cdot \| \Delta_{t+1} \|_{y_t} \;.
	\label{eq:os-goals}
\end{equation}
We start by proving inequality (i) 
from which inequality (ii) will follow. 

\paragraph{Proof of inequality (i).}
First, observe by the triangle inequality that 
\begin{equation}
	\|\Delta_{t-1}\|_{y_t}
	= 
	\|\Delta_{t} + (\Delta_{t-1} - \Delta_{t}) \|_{y_t}
	\le
	\|\Delta_{t}\|_{y_t} + \|\Delta_{t} -  \Delta_{t-1} \|_{y_t}\;.
	\label{eq:os01}
\end{equation}
Now if we can establish for some $C > 0$ that
\begin{equation}
    \|\Delta_{t} - \Delta_{t-1}\|_{y_t}
    \le 
    \eta C \|\Delta_{t-1}\|_{y_t}
    \label{eq:os-ind}
\end{equation}
then assuming $\eta < 1/(2C)$, it will follow
from~\eqref{eq:os01} that 
\begin{equation}
	\|\Delta_{t-1} \|_{y_t}
	\le 
	\frac{1}{1-\eta C} \cdot \|\Delta_{t}\|_{y_t}
	\le 
	2 \|\Delta_{t}\|_{y_t}  \;.
	\label{eq:os-final-i}
\end{equation}
We will prove the inequality of~\eqref{eq:os-ind}
via induction on $t \ge 2$. For this, we first
define several constants:

\smallskip
\paragraph{Definition of constants.}
Fix $\eta > 0$ satisfying~\eqref{eq:eta-helper}.
For the remainder of the proof, we will slightly 
overload notation and
let $\phi$, $\gamma$, and $C$ be
the ($\eta$-dependent) constants
\begin{equation}
    \textstyle
	\phi := \exp(\eta \wh \omega \wh \tau)\;,
	\quad
	\gamma := 
		\frac{\exp(3\eta \wh \omega \wh \tau)- 1}{3\eta \wh \omega \wh \tau} \;,
	\quad 
	\text{and}\;\;
	C := 4 \wh L \phi^4 \gamma \;.
	\label{eq:kappa-gamma}
\end{equation}
Here, $\wh \omega > 0$, $\wh \tau > 0$, and $\wh L \ge 1$ are
the constants from~\eqref{eq:constant-helpers},
where we recall that $\|y_{t}-y_{t-1}\|_{*} \le \eta \wh \tau$
for all $t \ge 0$.
Further observe for all $\eta > 0$ that
both $\phi > 0$ and $\gamma > 0$ are positive
and increasing in $\eta$.
Moreover, for $\eta$ as in~\eqref{eq:eta-helper},
it is also straightforward to check that $\eta C \le 1/2$. 

We now proceed with the inductive 
proof of expression~\eqref{eq:os-ind}
with $C$ as defined in~\eqref{eq:kappa-gamma}:

\smallskip
\noindent
\textbf{Base case.}
At $t = 2$, we have by definition that
\begin{equation}
	\Delta_{2}  - \Delta_1
	= 
	- 2 \eta V(x_1) + \eta V(x_0) - (- \eta V(x_0))
	= 
	- 2 \eta (V(x_1) - V(x_0))\;.
\end{equation}
Then taking local seminorms and applying
Lemma~\ref{lem:op-to-dual-diff}, we have by 
definition of $\gamma$ that 
\begin{equation}
	\|\Delta_2 - \Delta_1 \|_{y_1}
	= 
	2 \eta \|V(x_1) - V(x_0)\|_{y_1} \\
	\le 
	2 \eta \cdot \hat L 
	\gamma
	\cdot 
	\|y_1 - y_0 \|_{y_1} \;.
	\label{eq:os02}
\end{equation}
Then, using two applications of
Lemma~\ref{lem:local-transfer}, we 
conclude that 
\begin{align}
	\|\Delta_2 - \Delta_1 \|_{y_2} 
	\le 
	\phi
	\cdot 
	\|\Delta_2 - \Delta_1 \|_{y_1}  
	\le 
	\eta \cdot 
	2 \wh L
	\phi
	\gamma
	\cdot \|y_1 - y_0 \|_{y_1}  
	\le 
	\eta \cdot  2 \wh L \phi^2 \gamma
	\cdot \|y_1- y_0 \|_{y_2}  \;.
	\label{eq:os03}
\end{align}
It then follows by definition of $\Delta_1$
and using the fact that $\wh L, \phi, \gamma > 0$ that
\begin{equation}
	\|\Delta_2 - \Delta_1\|_{y_2}
	\le 
	\eta \cdot 2 \wh L \phi^2 \gamma  \cdot 
	\|\Delta_1 \|_{y_2} 
	\le 
	\eta \cdot 4 \wh L \phi^4 \gamma
	\cdot 
	\|\Delta_1 \|_{y_2} 
	\;.
	\label{eq:os-base}
\end{equation}

\smallskip
\paragraph{Inductive case.}
Now assume for all $2 \le k \le t-1$ that 
\begin{equation*}
	\|\Delta_k - \Delta_{k-1}\|_{y_k} 
	\le 
	\eta \cdot  
	4 \wh L \phi^4 \gamma
	\cdot \|\Delta_{k-1} \|_{y_k} \;.
\end{equation*}
Our goal is to establish this inequality also holds
at time $t$. For this observe by definition that 
\begin{align*}
	\Delta_{t} - \Delta_{t-1}
	&= 
	- 2 \eta V(x_{t-1})
	+ \eta V(x_{t-2})
	- 
	(-2\eta V(x_{t-2})+ \eta V(x_{t-3})) \\
	&= 
	- 2 \eta (V(x_{t-1}) - V(x_{t-2}))
	+ \eta (V(x_{t-2})  - V(x_{t-3})) \;.
\end{align*}
Then taking local semi-norms and using the triangle 
inequality, we can further write and bound:
\begin{alignat}{2}
	\|\Delta_t - \Delta_{t-1}\|_{y_t}
	&\le
	2 \eta \|V(x_{t-1}) - V(x_{t-2})\|_{y_t}
	+ 
	\eta \|V(x_{t-2}) - V(x_{t-3})\|_{y_t}  
		\nonumber \\
	&\le
	2 \eta \phi
	\|V(x_{t-1}) - V(x_{t-2})\|_{y_{t-1}}
	+ 
	\eta \phi^2 \|V(x_{t-2}) - V(x_{t-3})\|_{y_{t-2}}  
	&\;\;\text{(Lemma~\ref{lem:local-transfer})} 
		\nonumber \\
	&\le 
	2 \eta \phi
	  \wh L \gamma 
	\| y_{t-1} - y_{t-2}\|_{y_{t-1}}
	+ 
	\eta \phi^2 
	\wh L \gamma 
	\| y_{t-2} - y_{t-3}\|_{y_{t-2}}
	&\;\;\text{(Lemma~\ref{lem:op-to-dual-diff})} 
		\nonumber \\
	&= 
		\eta 2\wh L\phi \gamma  
		\|\Delta_{t-1}\|_{y_{t-1}}
		+ 
		\eta \wh L \phi^2 \gamma 
		\|\Delta_{t-2}\|_{y_{t-2}}   
		\nonumber \\
	&\le
		\eta \cdot 
		2 \wh L \phi^2 \gamma  
		\|\Delta_{t-1}\|_{y_{t}}
		+ 
		\eta
		\wh L \phi^3 \gamma 
		\|\Delta_{t-2}\|_{y_{t-1}} \;.
		\qquad\text{(Lemma~\ref{lem:local-transfer})} 
        &
	\label{eq:os-i-01}
\end{alignat}
Now recall by the inductive hypothesis that 
$\|\Delta_{t-1} - \Delta_{t-2}\|_{y_{t-1}} \le 
	\eta
	4 \wh L \phi^4 \gamma 
	\|\Delta_{t-2}\|_{y_{t-1}}$. 
By the triangle inequality, it then follows that 
\begin{equation*}
	\|\Delta_{t-2}\|_{y_{t-1}}
	\le 
	\|\Delta_{t-1}\|_{y_{t-1}}
	+ 
	\|\Delta_{t-1} - \Delta_{t-2}\|_{y_{t-1}}
	\le
	\|\Delta_{t-1}\|_{y_{t-1}}
	+ 
	\eta \cdot 4 \wh L \phi^4 \gamma \cdot 
	\|\Delta_{t-2}\|_{y_{t-1}}  \;.
\end{equation*}
Recall from~\eqref{eq:kappa-gamma} 
that we let $C := 4 \wh L \phi^4 \gamma$. 
Then rearranging and using 
the fact that $C \eta \le 1/2$ 
under the constraint on $\eta$,
we then obtain the bound
\begin{equation}
    \textstyle
	\|\Delta_{t-2}\|_{y_{t-1}}
	\le 
	\frac{1}{1- \eta C} \cdot 
	\|\Delta_{t-1}\|_{y_{t-1}}
	\le 
	\frac{\phi}{1- \eta C} \cdot 
	\|\Delta_{t-1}\|_{y_t}
	\le 2 \phi 
	\cdot \|\Delta_{t-1}\|_{y_t}
	 \;,
	\label{eq:os-i-02} 
\end{equation}
where the second inequality follows
from Lemma~\ref{lem:local-transfer}. 
Substituting~\eqref{eq:os-i-02} into~\eqref{eq:os-i-01}
then yields
\begin{align}
	\|\Delta_{t} - \Delta_{t-1}\|_{y_t}
	\le 
	\eta \cdot 
	(2 \wh L \phi^2 \gamma + 2\wh L \phi^4 \gamma) 
	\cdot 
	\|\Delta_{t-1}\|_{y_t} 
	\le 
	\eta \cdot 
	4 \wh L \phi^4 \gamma \cdot 
	\|\Delta_{t-1}\|_{y_t} \;,
\end{align}
where the final inequality is due to the fact
that $\wh L, \phi, \gamma > 0$.
Thus we conclude the inductive hypothesis
also holds at time $t$.

\paragraph{Combining the pieces for inequality (i).} 
Having established that expression~\eqref{eq:os-ind} holds
for all $t \ge 2$ with $C = 4\wh L \phi^4 \gamma$,
it follows by expression~\eqref{eq:os-final-i}
and the fact that $\eta C \le 1/2$
for $\eta$ as in~\eqref{eq:eta-helper} that 
$\|\Delta_{t-1}\|_{y_{t}}  \le 2 \|\Delta_t\|_{y_t}$, 
which yields the desired inequality (i) from~\eqref{eq:os-goals}.

\smallskip
\paragraph{Proof of inequality (ii).}
To prove the second inequality (ii) from~\eqref{eq:os-goals},
we apply inequality (i) at time $t+1$. 
Together with two applications 
of Lemma~\ref{lem:local-transfer}, we then obtain
\begin{equation}
	\|y_{t}-y_{t-1}\|_{y_{t}}
	\le 
	\phi \cdot 
	\|y_{t}-y_{t-1}\|_{y_{t+1}}
	\le 
	2 \phi \cdot
	\|y_{t+1} - y_t \|_{y_{t+1}}
	\le 
	2 \phi^2 \cdot  
	\| y_{t+1} - y_t\|_{y_t} \;,
	\label{eq:os-ii}
\end{equation}
where $\phi > 0$ is as defined in~\eqref{eq:kappa-gamma}.
Under the setting of $\eta > 0$ as in~\eqref{eq:eta-helper},
it is straightforward
to check that $\phi^2 = \exp( 6 \eta \wh \omega \wh B) \le 3/2$.
Thus we conclude from~\eqref{eq:os-ii} that
$\|y_{t}-y_{t-1}\|_{y_t} \le 3 \|y_{t+1} - y_t\|_{y_t}$.
This yields the desired inequality (ii) from~\eqref{eq:os-goals}
and completes the proof of the lemma. 
\end{proof}

\subsubsection{Bound on initial change in energy}
\label{app:energy:oftrl:initial-bound}

For~\eqref{eq:oftrl}, note that the dissipation 
of Lemma~\ref{lem:energy-dissipation-main}
holds only for times $t \ge 1$. 
At time $t=0$, the shifted energy 
$F_{x^\star}(y_{1}) - F_{x^\star}(y_0)$ 
(for $x^\star \in \MEQ(V)$) may increase.
In the following lemma, we derive an upper bound
on this maximum increase in shifted energy
after the first iteration:

\begin{restatable}[OFTRL energy startup]
    {lem}{lemoftrlstartup}
    \label{lem:oftrl-initial-energy}
    Let $\{(x_t, y_t)\}$ be the
    iterates of~\eqref{eq:oftrl-helper}
    with $\eta$ as in Lemma~\ref{lem:energy-dissipation-main}.
    Then for any $x^\star \in \MEQ(V)$, it holds
    at time $t=1$ that
    \begin{equation*}
        F_{x^\star}(y_1) - F_{x^\star}(y_0) \le \eta^2 \|V(x_0)\|^2_{y_0} \;.
    \end{equation*}
\end{restatable}

\begin{proof}
    Recall for~\eqref{eq:oftrl} that 
    $y_1 - y_0 = - \eta V(x_0)$. 
    Thus at the first step, OFTRL takes 
    a forward, vanilla FTRL step. 
    Then, fix $x^\star \in \MEQ(V)$,
    and using the fact that the initializations 
    satisfy $x_0 = \nabla F(y_0)$, observe that
    \begin{align*}
        F_{x^\star}(y_1)
        - 
        F_{x^\star}(y_0)
        &= 
        F(y_1) - F(y_0)
        - 
        \langle 
        x^\star, y_1- y_0
        \rangle 
        \\
        &= 
        - \eta 
        \langle x_0 - x^\star, V(x_0)\rangle
        + \D_F(y_1, y_0) 
        = 
         \D_F(y_1, y_0) \;.
    \end{align*}
    Here, the final equality is
    due to $x^\star \in \MEQ(V)$. 
    As $\|y_1 - y_0\|_{*} \le \eta \wh \tau$, 
    we apply Lemma~\ref{lem:bregman-local-bound} to obtain
    \begin{equation*}
        F_{x^\star}(y_1)
        - 
        F_{x^\star}(y_0)
        =
        \D_F(y_1, y_0) 
        \le \mu(\eta \wh \tau) \cdot \|y_1 - y_0\|^2_{y_0}
        \le 
        \mu(\eta \wh \tau) \cdot \eta^2 \|V(x_0)\|^2_{y_0} \;,
    \end{equation*}
    where $\mu(\cdot)$ is the scalar function from~\eqref{eq:scalar-helpers}.
    For $\eta$ as in Lemma~\ref{lem:energy-dissipation-main}
    (and in~\eqref{eq:eta-helper}),
    it is straightforward to verify the crude bound
    $\mu(\eta \wh \tau) \le 1$. 
    Substituting this into the prior expression
    yields the desired bound
    and completes the proof.
\end{proof}

\medskip
We also derive the following crude upper bound 
on $\D_R(x^\star, x_1)$
which will be used later:

\smallskip
\begin{restatable}[OFTRL initial Bregman]
    {lem}{lemoftrlstartupbregman}
    \label{lem:oftrl-initial-bregman}
    Let $\{(x_t, y_t)\}$ be the
    iterates of~\eqref{eq:oftrl-helper}
    with $\eta$ as in Lemma~\ref{lem:energy-dissipation-main}.
    Then for every 
    $x^\star \in \MEQ(V) \cap \relint(\calX)$,
    it holds that
    $\D_R(x^\star, x_1) 
    \le 2 \D_R(x^\star, x_0)$. 
\end{restatable}

\begin{proof}
    First, we will derive a general 
    (algorithm-agnostic) upper bound on 
    the term $\|V(x)\|^2_{y}$ 
    that holds for all primal-dual pairs
    $x$ and $y$. 
    For this, fix $y \in \R^D$ and $x = \nabla F(y) \in \relint(\calX)$.
    Further
    fix $x^\star \in \MEQ(V) \cap \relint(\calX)$,
    and recall by Lemma~\ref{lem:wmeq-ne}
    that $V(x^\star) \in \calS$. 
    Thus, as $\Null(\nabla^2 F(y)) = \calS$
    for all $y \in \R^D$,
    we have $\|V(x)\|^2_{y} = \|V(x) - V(x^\star)\|^2_{y}$.
    Then, applying Lemma~\ref{lem:norm-global-vs-local},
    using the $\wh L$-Lipschitzness of $V$, 
    and by 1-strong-convexity of $R$ with respect
    to $\|\cdot\|$, we obtain:
    \begin{equation}
        \|V(x)\|^2_{y}
        = 
        \|V(x)-V(x^\star)\|^2_{y}
        \le 
        \|V(x)-V(x^\star)\|^2_{*}
        \le 
        \wh L^2 \cdot \|x- x^\star\|^2 
        \le 
        2 \wh L^2 \cdot \D_R(x^\star, x) \;.
        \label{eq:dissipation-breg-upper}
    \end{equation}
    Now, for the~\eqref{eq:oftrl-helper}
    iterates, we have by 
    Lemma~\ref{lem:oftrl-initial-energy} that
    $F_{x^\star}(y_1) - F_{x^\star}(y_0) \le 
    \eta^2 \|V(x_0)\|^2_{y_0}$ 
    for all $x^\star \in \MEQ(V)$. 
    Then applying~\eqref{eq:dissipation-breg-upper}
    to the OFTRL pair $(x_0, y_0)$
    and using Proposition~\ref{prop:bregman-shifted-energy},
    we obtain for any 
    $x^\star \in \MEQ(V) \cap \relint(\calX)$:
    \begin{equation}
        \D_R(x^\star, x_1)
        -
        \D_R(x^\star, x_0)
        \le 
        \eta^2 \|V(x_0)\|^2_{y_0}
        \le
        2 \wh L^2 \eta^2 \D_R(x^\star, x_0)
        \le 
        \D_R(x^\star, x_0) \;.
        \label{eq:diss-breg-oftrl}
    \end{equation}
    Here, the final inequality comes
    from the fact that
    (conservatively) $2\wh L^2 \eta^2 \le 1$
    for $\eta$ as in~\eqref{eq:eta-helper}.
    Rearranging~\eqref{eq:diss-breg-oftrl}
    then yields the statement of the lemma.
\end{proof}

\subsection{Instantiating Coupling for EGFTRL}
\label{app:energy:egftrl}

Here we give details on instantiating 
of Lemma~\ref{lem:energy-dissipation-main}
along the iterates of~\eqref{eq:egftrl-helper},
which requires slightly
less work than for the OFTRL case. 

First, the following proposition relates 
the dual increments $y_{t+1} - y_t$
to the operator value $V(x_t)$:

\smallskip
\begin{restatable}
	{prop}{propegftrldualincrement}
	\label{prop:egftrl-dual-increment-op}
	For~\eqref{eq:egftrl-helper}
	with $\eta > 0$ as in~\eqref{eq:eta-helper}, 
	it holds for all $t \ge 0$:
	\begin{equation*}
		\textstyle
		\frac{5}{6} \eta  \|V(x_t)\|_{y_t} 
		\le \|y_{t+1} - y_t \|_{y_t}
		\le \frac{7}{6} \eta \|V(x_t)\|_{y_t} \;.
	\end{equation*}
\end{restatable}

\medskip
The second proposition derives a bound on 
the coupling error terms 
$C_t = \langle x_{t+1} - x^\star, y_{t+1} - \wh y_{t+1}\rangle$
from the implicit-step coupling 
of Proposition~\ref{prop:meta-coupling}.

\begin{restatable}
	{prop}{propegftrlctbound}
	\label{prop:egftrl-dual-ct-bound}
	For~\eqref{eq:egftrl-helper}
	with $\eta > 0$ as in~\eqref{eq:eta-helper}, 
	it holds for all $t \ge 0$
    and $x^\star \in \MEQ(V)$:
	\begin{equation*}
        \textstyle
        |C_t|
        \le 
        9 \wh L \cdot \eta^3 \|V(x_t)\|^2_{y_t} \;.
	\end{equation*}
\end{restatable}

We give the proof of
Propositions~\ref{prop:egftrl-dual-increment-op}
and~\ref{prop:egftrl-dual-ct-bound} 
for~\eqref{eq:egftrl-helper}
below in Section~\ref{app:energy:egftrl:dual-increment}
and Section~\ref{app:energy:egftrl:ct}, respectively.
Compared to the proofs for~\eqref{eq:oftrl-helper},
for the~\eqref{eq:egftrl-helper}, we require slightly 
different machinery given the differences
in the dual updates. In particular, 
we require the following stability bound
on the dual~\eqref{eq:egftrl-helper} iterates
in terms of the halfstep sequence $\{y_{t+1/2}\}$:

\medskip
\begin{restatable}[EGFTRL Dual Stability]{lem}{lemegftrldualstability}
	\label{lem:egftrl-dual-stability}
	Under~\eqref{eq:egftrl} with $\eta$ 
    satisfying~\eqref{eq:eta-helper},
	for all $t \ge 0$:
	\begin{equation*}
        \textstyle
		\tfrac{5}{6} \cdot \|y_{t+1/2}- y_t\|_{y_t}
		\;\le\;
		\|y_{t+1} - y_t\|_{y_t}
		\;\le\;
		\tfrac{7}{6} \cdot \|y_{t+1/2} - y_t\|_{y_t} \;.
	\end{equation*}
\end{restatable}

The proof of Lemma~\ref{lem:egftrl-dual-stability}
is given in Section~\ref{app:energy:egftrl:stability}.
For now, we grant the result as true for use
in the proofs of Propositions~\ref{prop:egftrl-dual-increment-op}
and~\ref{prop:egftrl-dual-ct-bound}.

\subsubsection{Proof of Proposition~\ref{prop:egftrl-dual-increment-op}}
\label{app:energy:egftrl:dual-increment}

Here, we give the short proof of Proposition~\ref{prop:egftrl-dual-increment-op},
which relates for~\eqref{eq:egftrl-helper} the magnitude
of the dual increments $y_{t+1} - y_t$ 
to the vector $V(x_t)$ under the local seminorm $\|\cdot\|_{y_t}$

\smallskip
\begin{proof}
Recall under~\eqref{eq:egftrl} that
$y_{t+1/2} - y_t = - \eta V(x_t)$ for all $t \ge 0$.
Thus substituting this difference into the statement 
of Lemma~\ref{lem:egftrl-dual-stability} yields the claim. 
\end{proof}

\subsubsection{Proof of Proposition~\ref{prop:egftrl-dual-ct-bound}}
\label{app:energy:egftrl:ct}

In this section, we give the proof of
Proposition~\ref{prop:egftrl-dual-ct-bound}.
For this, recall for any $x^\star \in \MEQ(V)$
that we define 
$C_t := \langle x_{t+1} - x^\star, y_{t+1} - \wh y_{t+1} \rangle$.

\begin{proof}
Similar to Proposition~\ref{prop:oftrl-dual-ct-bound},
the proof proceeds by first expanding 
$C_t$ under the~\eqref{eq:egftrl-helper} iterates,
and then bounding the resulting constituent error terms.

\paragraph{1. Expanding the inner product term.} ~

\noindent
Under~\eqref{eq:egftrl-helper}, the difference $y_{t+1} - \wh y_{t+1}$
is given by 
\begin{equation*}
	y_{t+1} - \wh y_{t+1}
	= 
		y_t - \eta V(x_{t+1/2}) - (y_t - \eta V(x_{t+1})) 
	= 	\eta \big( V(x_{t+1})	- V(x_{t+1/2})\big) \;.
\end{equation*}
Fixing $x^\star$ satisfying~\eqref{eq:minty-equality}, 
it follows that we can simplify
\begin{align}
    C_t = 
	\langle x_{t+1} - x^\star, y_{t+1} - \wh y_{t+1} \rangle
	&= 
	\eta \langle x_{t+1} -x^\star, V(x_{t+1}) - V(x_{t+1/2}) \rangle
		\nonumber  \\
	&= 
		- \eta \langle x_{t+1} - x^\star, V(x_{t+1/2}) \rangle 
		\nonumber \\
	&= 
		- \eta \langle x_{t+1/2} - x^\star + x_{t+1} - x_{t+1/2}, V(x_{t+1/2}) \rangle 
		\nonumber \\
	&= 
	- \eta \langle x_{t+1} - x_{t+1/2}, V(x_{t+1/2}) \rangle \;.
		\label{eq:eginner0}
\end{align}
Here, the second and fourth equalities follow from
the fact that $x^\star$ satisfies~\eqref{eq:minty-equality}. 
To further simplify the inner product in~\eqref{eq:eginner0},
we use the definition of the~\eqref{eq:egftrl-helper} iterates 
and the Hessian remainder term $\G_F(y_{t+1}, y_{t+1/2})$ 
to write
\begin{align}
	x_{t+1} - x_{t+1/2}
	&= 
	\nabla F(y_{t+1}) - \nabla F(y_{t+1/2}) 
		\nonumber \\
	&= 
	\nabla^2 F(y_{t+1/2})(y_{t+1} - y_{t+1/2}) + \G_F(y_{t+1}, y_{t+1/2})
		\nonumber \\
	&= 
	\eta \nabla^2 F(y_{t+1/2})(V(x_t) - V(x_{t+1/2})) 
	+ \G_F(y_{t+1}, y_{t+1/2}) 
	\label{eq:eginner1} \;.
\end{align}
Due to the symmetry of the Hessian $\nabla^2 F(y_{t+1/2})$,
substituting~\eqref{eq:eginner1} into~\eqref{eq:eginner0} yields
\begin{equation}
	C_t
	= 
	\underbrace{\eta^2 
	\big\langle V(x_{t+1/2}) - V(x_t), V(x_{t+1/2})\big\rangle_{y_{t+1/2}}}_{\calE_{t,1}}
	- 
    \underbrace{\eta \big\langle \G_F(y_{t+1}, y_{t+1/2}), V(x_{t+1/2}) \big\rangle}_{\calE_{t, 2}}
	\label{eq:eginnerfinal}
	\;.
\end{equation}

\paragraph{2. Bounds on magnitude of coupling error terms.} ~

We now derive bounds on the magnitudes
of the two error terms $\calE_{t, 1}$
and $\calE_{t, 2}$ from~\eqref{eq:eginnerfinal}.
For this, we use the scalar functions
$\phi(\cdot)$, $\gamma(\cdot)$,
and $\nu(\cdot)$ from~\eqref{eq:scalar-helpers}.
Recall also that under~\eqref{eq:egftrl-helper},
we have $\|y_{t+1} -y_t\|_{*}, \|y_{t+1/2} - y_t\|_{*} \le \eta \wh \tau$ 
for all $t \ge 0$.
Thus, for readability, we will 
slightly abuse notation and overload 
$\phi := \phi(\eta \wh \tau)$,
$\gamma := \gamma(\eta \wh \tau)$,
and $\nu := \nu(\eta \wh \tau)$. 
We then proceed with the bounds on $\calE_{t, 1}$, and $\calE_{t,2}$:

First, for the term $\calE_{t, 1}$,
recall that $y_{t+1} - y_t = - \eta V(x_{t+1/2})$. 
Then we can bound:
\begin{alignat}{2}
    |\calE_{t, 1}| 
    &=
    \big|
    \langle \eta (V(x_{t+1/2}) - V(x_t)), \eta V(x_{t+1/2}) \rangle_{y_{t+1/2}}
    \big|
        \nonumber \\
    &\le
    \eta \|V(x_{t+1/2})- V(x_t)\|_{y_{t+1/2}}
    \cdot 
    \|\eta V(x_{t+1/2})\|_{y_{t+1/2}}
    &&\quad\text{(Cauchy-Schwarz)}  
    \nonumber \\
    &\le 
    \phi^2
    \cdot \eta 
    \|V(x_{t+1/2}) - V(x_t)\|_{y_t} \cdot \|y_{t+1} - y_t\|_{y_t} 
    &&\quad\text{(Lemma~\ref{lem:local-transfer}, twice)} 
    \nonumber \\
    &\le 
    \wh L 
    \gamma \phi^2 \cdot 
    \eta 
    \|y_{t+1/2} - y_t\|_{y_t} \cdot \|y_{t+1} - y_t\|_{y_t} 
    &&\quad\text{(Lemma~\ref{lem:op-to-dual-diff})} 
        \nonumber \\
    &\le
    2 \wh L
    \gamma \phi^2 \cdot 
    \eta \|y_{t+1} - y_t\|^2_{y_t} 
    &&\quad\text{(Lemma~\ref{lem:egftrl-dual-stability})}
        \nonumber \\
    &\le 
    \tfrac{9}{2} 
    \wh L \gamma \phi^2
    \cdot 
    \eta^3 \|V(x_t)\|^2_{y_t}
    &&\quad\text{(Proposition~\ref{prop:egftrl-dual-increment-op})}
    \label{eq:eg-et1}
    \;.
\end{alignat}

For the second term $\calE_{t, 2}$, first observe 
under~\eqref{eq:egftrl-helper} that
$y_{t+1} - y_{t+1/2} = \eta (V(x_t) - V(x_{t+1/2}))$.
Then it follows that we can bound the term
\begin{alignat}{2}
    \|\G_F(y_{t+1}, y_{t+1/2})\|_{y_t, *}
    &\le
    \phi \cdot \|\G_F(y_{t+1}, y_{t+1/2})\|_{y_{t+1/2, *}}
    &&\quad\text{(Lemma~\ref{lem:local-transfer})} 
    \nonumber  \\
    &\le
    \phi \nu
    \cdot \|y_{t+1} -y_{t+1/2}\|_{y_{t+1/2}}
    &&\quad\text{(Lemma~\ref{lem:hessian-remainder-bound})} 
    \nonumber \\
    &\le 
    \phi^2 \nu
    \cdot \|y_{t+1} -y_{t+1/2}\|_{y_{t}}
    &&\quad\text{(Lemma~\ref{lem:local-transfer})} 
    \nonumber \\
    &= 
    \phi^2 \nu
    \cdot \eta
    \|V(x_{t+1/2}) - V(x_t)\|_{y_t} 
    \nonumber \\
    &\le 
    \wh L \gamma \phi^2 \nu 
    \cdot \eta \|y_{t+1/2} - y_t\|_{y_t} 
    &&\quad\text{(Lemma~\ref{lem:op-to-dual-diff})} 
    \nonumber \\
    &\le
    2 \wh L \gamma \phi^2 \nu 
    \cdot \eta \|y_{t+1} - y_t \|_{y_t} 
    &&\quad \text{(Lemma~\ref{lem:egftrl-dual-stability})} \;.
    \label{eq:eg-et2-gf}
\end{alignat}
To bound $|\calE_{t, 2}|$, we can then use the generalized
Cauchy-Schwarz inequality to write
\begin{alignat}{2}
    |\calE_{t, 2}|
    &\le
    \|\G_F(y_{t+1}, y_{t+1/2})\|_{y_t, *} 
    \cdot 
    \|\eta V(x_{t+1/2})\|_{y_t} 
    \nonumber \\
    &\le 
    2 \wh L \gamma \phi^2 \nu 
    \cdot \eta \|y_{t+1} - y_t\|^2_{y_t} 
    &&\quad\text{(Substituting~\eqref{eq:eg-et2-gf})} 
    \nonumber \\
    &\le
    \tfrac{9}{2} \wh L \gamma \phi^2 \nu 
    \cdot \eta^3 \|V(x_t)\|^2_{y_t} 
    &&\quad\text{(Proposition~\ref{prop:egftrl-dual-increment-op})}
    \;.
    \label{eq:eg-et2}
\end{alignat}

\paragraph{3. Combining the pieces.} ~   

\nonumber 
Substituting the bounds on $|\calE_{t, 1}|$ and $|\calE_{t, 2}|$
into~\eqref{eq:eginnerfinal}, we obtain
\begin{equation*}
    \textstyle
    |C_t|
    \le
    |\calE_{t,1}| + |\calE_{t, 2}|
    \le
    \frac{9}{2} \wh L \gamma \phi^2  \cdot 
    (1 + \nu) 
    \cdot \eta^3
    \|V(x_t)\|^2_{y_t} \;.
\end{equation*}
For $\eta$ as in~\eqref{eq:eta-helper},
it is straightforward to verify that
$\gamma \phi^2 (1+ \nu) \le 2$.
Thus we obtain
\begin{equation*}
    \textstyle
    |C_t| 
    \le 
    9 \wh L \cdot \eta^3 \|V(x_t)\|^2_{y_t} \;,
\end{equation*}
which concludes the proof.
\end{proof}

\subsubsection{Proof of Lemma~\ref{lem:egftrl-dual-stability}}
\label{app:energy:egftrl:stability}

This section gives the proof of Lemma~\ref{lem:egftrl-dual-stability},
which establishes the stability of 
the dual iterates of~\eqref{eq:egftrl-helper}
in terms of the half-step sequence $\{y_{t+1/2}\}$:

\begin{proof}
For all $t \ge 0$, 
let $\Delta_{t+1/2} = y_{t+1/2} - y_t$
and $\Delta_{t+1} = y_{t+1} - y_t$. 
Then under~\eqref{eq:egftrl-helper}, observe that
\begin{equation*}
	\Delta_{t+1} - \Delta_{t+1/2}
	= 
	\eta (V(x_{t}) - V(x_{t+1/2})) \;.
\end{equation*}
Further recall for all $t \ge 0$ 
that $\|y_t - y_{t+1/2}\|_{*} \le \eta \wh \tau$
(for $\wh \tau$ as in~\eqref{eq:constant-helpers}).
Similar to the proof of Lemma~\ref{lem:oftrl-dual-stability},
we again slightly overload notation and let 
$\gamma := \frac{\exp(\eta \wh \omega \wh \tau) - 1}
{\eta \wh \omega \wh \tau}$.
For $\eta$ as in~\eqref{eq:eta-helper},
it is then straightforward to verify that 
$\eta \wh L \gamma \le 1/6$.
Then, taking local semi-norms and applying 
Lemma~\ref{lem:op-to-dual-diff}, we obtain the bound
\begin{equation}
	\|\Delta_{t+1} - \Delta_{t+1/2}\|_{y_t}
	=
	\eta \|V(x_{t+1/2}) - V(x_t)\|_{y_t}  
	\le 
	\eta \cdot \wh L \gamma \cdot \|y_{t+1/2} - y_t \|_{y_t} \;.
	\label{eq:egs00}
\end{equation}
We now prove the upper and lower inequalities of 
the lemma statement separately:

\paragraph{Upper inequality.}
We have that 
$y_{t+1} - y_t = \Delta_{t+1} = \Delta_{t+1/2} +
(\Delta_{t+1} - \Delta_{t+1/2})$. 
Then by the triangle inequality and substituting
the bound from~\eqref{eq:egs00}, we find
\begin{align*}
	\|y_{t+1} - y_t\|_{y_t}
	&\le
	\|\Delta_{t+1/2}\|_{y_t}
	+ 
	\|\Delta_{t+1} - \Delta_{t+1/2} \|_{y_t} \\
	&\le
	(1 + \eta L \gamma) \cdot \|\Delta_{t+1/2}\|_{y_t} 
	\le 
	\tfrac{7}{6} \cdot \|\Delta_{t+1/2}\|_{y_t} 
	= 
	\tfrac{7}{6} \cdot \|y_{t+1/2} - y_t \|_{y_t}  \;.
\end{align*}
Here, the final inequality comes from 
$\eta \wh L \gamma \le 1/6$
for $\eta$ as in~\eqref{eq:eta-helper}.
This proves the upper inequality. 

\paragraph{Lower inequality.}
We similarly have that
$y_{t+1/2} - y_t = \Delta_{t+1/2} = \Delta_{t+1} + 
(\Delta_{t+1/2} - \Delta_{t+1})$.
Again by the triangle inequality and substituting the
bound from~\eqref{eq:egs00}, we then obtain
\begin{equation*}
	\|\Delta_{t+1/2} \|_{y_t}
	\le 
	\|\Delta_{t+1}\|_{y_t} 
	+ 
	\|\Delta_{t+1} - \Delta_{t+1/2}\|_{y_t}
	\le 
	\|\Delta_{t+1}\|_{y_t}
	+ \eta \wh L \gamma \cdot \|\Delta_{t+1/2}\|_{y_t} \;.
\end{equation*}
Rearranging and using the fact that 
$\eta \wh L \gamma \le 1/6$ then yields
\begin{equation*}
	\|y_{t+1} - y_t\|_{y_t}
	= \|\Delta_{t+1}\|_{y_t}
	\ge 
	(1- \eta \wh L \gamma) \cdot \|\Delta_{t+1/2}\|_{y_t}
	\ge 
	\tfrac{5}{6} \cdot \|y_{t+1/2} - y_t\|_{y_t} \;.
\end{equation*}
This proves the lower inequality and
completes the proof of the lemma.
\end{proof}


\section{Details on Equilibrium Selection}
\label{app:eq-selection}

In this section, we prove
for the symmetrized system
that OFTRL and EGFTRL
converge asymptotically
to a unique limit point,
and we further give an exact characterization of the limit.
Formally:

\begin{restatable}[Equilibrium selection]
	{lem}{lemeqselection}
	\label{lem:eq-selection}
	Assume the setting of Lemma~\ref{lem:energy-dissipation-main}.
	Further, let Assumption~\ref{ass:restricted-regularity} hold,
    and assume that $\MEQ(V) \cap \relint(\calX) \neq \emptyset$.
	Then for any initialization satisfying $x_0 \in \relint(\calX)$, 
    it holds for both
	\eqref{eq:oftrl-helper} and~\eqref{eq:egftrl-helper} 
    that:
	\begin{enumerate}[
		label={(\roman*)},
		leftmargin=2em,
	]
	\item
	There exists a unique limit point
	$x_{\infty} := \lim_{t \to \infty} x_t \in \MEQ(V) \cap \calR(x_0)$.
	\item
	The point $x_{\infty}$ satisfies
	$x_{\infty} = \argmin_{x^\star \in \MEQ(V)} \D_R(x^\star, x_0)$.
	\end{enumerate}
\end{restatable}

\smallskip
\subsection{Proof of Lemma~\ref{lem:eq-selection}}

This section gives the proof of Lemma~\ref{lem:eq-selection}.
We prove the two claims in sequence:

\smallskip
\paragraph{Proof of part (i) -- unique MEQ limit point.} ~

\noindent
First, we show that the iterates $\{x_t\}_{t \geq 1}$ converge
asymptotically to a unique limit point $x_{\infty} \in \MEQ(V)$:
We proceed in several steps:

\medskip
\noindent
\textit{1. Iterates lie in a compact sublevel set.} ~

\smallskip
\noindent
First, using Proposition~\ref{prop:bregman-shifted-energy}
and the upper bound on shifted energy 
dissipation from Lemma~\ref{lem:energy-dissipation-main},
we have for all $t \ge 1$ 
and $x^\star \in \MEQ(V) \cap \relint(\calX)$ that:
\begin{equation}
    \textstyle
    \D_R(x^\star, x_{t+1})
    - 
    \D_R(x^\star, x_t) 
    = 
    F_{x^\star}(y_{t+1})
    - 
    F_{x^\star}(y_{t})
    \le 
    - \frac{1}{8} \eta^2 \|V(x_t)\|^2_{y_t} \;.    
    \label{eq:breg-change}
\end{equation}
Now, fix $\wt x^\star \in \MEQ(V) \cap \relint(\calX)$,
which exists by assumption
of the lemma. 
Further, let $v := \D_R(\wt x^\star, x_1)$, 
where $v$ is finite due to 
$x_1, \wt x^\star \in \relint(\calX)$.
By Proposition~\ref{prop:bregman-boundary-blowup}, 
recall that $\D_R(\widetilde x^\star, x) \to + \infty$
as $x$ approaches the relative boundary of $\calX$.
Together with the continuity of $\D_R(\wt x^\star, \cdot)$
over $\relint(\calX)$, 
this implies compactness of the sublevel set
\begin{equation}
    \calC := \{x \in \relint(\calX) : 
    \D_R(\wt x^\star,x) \le v\}
    \subset \relint(\calX) \;.
    \label{eq:iterates-compact}
\end{equation}
Together with~\eqref{eq:breg-change}, 
this implies that $x_t \in \calC$ for all $t \ge 1$.

\noindent
\textit{2. Primal dissipation function and its zero set.} 

\noindent
First, let $\Phi: \relint(\calX) \to \R$ 
be given by
$\Phi(x)
:= 
\frac{1}{8} \eta^2 
\langle V(x), \nabla^2 F(y) V(x) \rangle
$
for all $x \in \relint(\calX)$, 
where $y = \nabla R(x)$.
We first relate $\Phi(x_t)$ to the dissipation term
in~\eqref{eq:breg-change}. 
For any $t$ and for $x_t = \nabla F(y_t) \in \relint(\calX)$
we have by definition of $\Phi$ that 
$\Phi(x_t) 
= 
\frac{1}{8} \eta^2 
\langle V(x_t), \nabla^2 F(\nabla R(x_t)) V(x_t) \rangle$.
Now, since $x_t = \nabla F(y_t)$, it follows from 
Proposition~\ref{prop:first-order-duality} that 
$y_t - \nabla R(x_t)  \in \calS$. Then by the shift-invariance 
identity of Proposition~\ref{prop:second-order-duality} (i) 
we obtain that 
$\nabla^2 F(\nabla R(x_t)) 
= \nabla^2 F(\nabla R(x_t) + y_t - \nabla R(x_t)) 
= \nabla^2 F(y_t)$. 
Plugging this identity back into the expression of
$\Phi(x_t)$ gives: 
$\Phi(x_t) = \frac{1}{8} \eta^2 
\langle V(x_t), \nabla^2 F(y_t) V(x_t) \rangle 
=  \frac{1}{8} \eta^2 \|V(x_t)\|_{y_t}^2$,
where the final equality follows directly from the 
definition of the local semi-norm. 

By definition of the local norm $\|\cdot\|_{y}$ and continuity of
$V$, $\nabla R$, and $\nabla^2 F(y)$, we have that
$\Phi$ is continuous on $\relint(\calX)$. 
Moreover, recall from 
Section~\ref{app:error-machinery:seminorms} that
$\Phi(x) \ge 0$, and $\Phi(x) = 0$ if and only if 
$V(x) \in \Null(\nabla^2 F(y)) = \calS$. 
Thus by Proposition~\ref{prop:ne-vs-normal}, we have
$\Phi(x) = 0$ if and only if $x \in \NE(V)$.

\smallskip
\noindent
\textit{3. Subsequence convergence to NE.} ~

\noindent
Now fix the 
point $\wt x^\star \in \MEQ(V) \cap \relint(\calX)$ 
from Step 1.
Then summing~\eqref{eq:breg-change} 
from $t=1$ to $T$ 
and using the definition of $\Phi$ 
yields 
\begin{equation*}
    \textstyle
    \D_R(\wt x^\star, x_{T+1}) - \D_R(\wt x^\star, x_1)
    = 
    \sum_{t=1}^T 
    \big(\D_R(\wt x^\star, x_{t+1})
    - 
    \D_R(\wt x^\star, x_{t})\big)
    \le 
    - \sum_{t=1}^T \Phi(x_t)\;.
\end{equation*}
Rearranging, using the nonnegativity
of $\D_R(\wt x^\star, x_{T+1})$ and 
each $\Phi(x_t)$,
it follows that
\begin{equation*}
\textstyle
0 \le \sum_{t=1}^T \Phi(x_t) \le \D_R(\wt x^\star, x_1)
= v < \infty
\;\;\text{for all $T \ge 1$}\;.
\end{equation*}
Taking $T \to \infty$, this means
that $\sum_{t=1}^\infty \Phi(x_t) < \infty$, which implies 
$\lim_{t \to \infty} \Phi(x_t) = 0$.
Now, since all $x_t \in \calC$, 
it follows by compactness of $\calC$
that there exists a convergent subsequence $\{x_{t_k}\}_{k \geq 1}$ such that
$x_{t_k} \to \overline x$.
By compactness, we then have
$\overline x \in \calC \subset \relint(\calX)$,
and by continuity of $\Phi$ on $\relint(\calX)$,
it thus follows that
$\lim_{k \to \infty} \Phi(x_{t_k}) =
\Phi(\overline x) = 0$.
Here, the second equality holds
due to $\Phi(x_t) \to 0$ along the entire sequence.
By Step 2, this means that
the limit $\overline x$ of any convergent
subsequence $\{x_{t_k}\}_{k \geq 1}$ satisfies
$\overline x \in \NE(V)$.

\smallskip
\noindent
\textit{4. Global convergence to unique MEQ point.} ~ 

\noindent
Finally, we establish that the entire
sequence $\{x_t\}_{t \geq 1}$ converges to a unique 
MEQ point. First, we show that
the limit point $\overline x$ of any
convergent subsequence from Step 3
is in the reachable primal set $\calR(x_0)$. 
Let $\overline x$ be such a limit point,
and recall that $x_t \in \calR(x_0)$
for all $t \ge 0$. 
Then, by Step 3, we also have that 
$\overline x \in \NE(V) \cap \calC \subset \relint(\calX)$.
Recall from Proposition~\ref{prop:reachable-helper}
that for the initialization pair 
$y_0 \in \R^D$ and $x_0 = \nabla F(y_0)$, 
it holds for any $x \in \relint(\calX)$
that 
$x \in \calR(x_0)$
if and only if 
$\Pi_{\calS^\bot} (\nabla R(x) - y_0) \in \calU$,
where $\calU = \Pi_{\calS^\bot} \Span(V(\calX))$.
By continuity of $\nabla R$ on $\relint(\calX)$,
and since $\calU$ is a closed, linear subspace,
it then follows that if
$x_{t_k} \in \calR(x_0)$ for all points
in the convergent subsequence $\{x_{t_k}\}$,
then also $\overline x \in \calR(x_0)$. 

Now, under Assumption~\ref{ass:restricted-regularity},
recall by Corollary~\ref{cor:reachable-meq} that
if $x \in \NE(V) \cap \calR(x_0)$, 
then $x \in \MEQ(V)$. 
Thus $\overline x \in \MEQ(V) \cap \calR(x_0)$. 
However, we have by Proposition~\ref{prop:unique-reachable-meq}
that 
$|\MEQ(V) \cap \calR(x_0)| \le 1$
(which holds independently of Assumption~\ref{ass:restricted-regularity}).
Since the preceding arguments held 
for the limit of any convergent subsequence,
it follows that all convergent subsequences
have the same limit, which we denote by
$x_{\infty} \in \MEQ(V)$.
Finally, by compactness, 
the entire sequence $\{x_t\}$
must also converge to this unique point 
$x_{\infty} := \lim_{t \to \infty} x_t \in \MEQ(V)$.
\hfill $\blacksquare$

\medskip
\paragraph{Proof of part (ii) -- characterization of limit point.} ~

By Part (i) of Lemma~\ref{lem:eq-selection},
we have that
$x_{\infty} := \lim_{t\to\infty} x_t \in \MEQ(V)$. 
For Part (ii),
we further give an exact characterization 
of $x_{\infty}$ in terms of 
the initialization $x_0$. 
We proceed in several steps:

\smallskip
\noindent
\textit{1. Change in shifted energy is equivalent 
for all MEQ points.} ~

\noindent
Recall for $x^\star \in \MEQ(V)$
and $y \in \R^D$ that
$F_{x^\star}(y) = F(y) - \langle x^\star, y \rangle$.
Then the following helper lemma
establishes that 
the change in shifted energy
with respect to any MEQ point is invariant:

\begin{restatable}{lem}{lemshiftedenergyinvariant}
	\label{lem:shifted-energy-invariant}
    Assume the setting of
    Lemma~\ref{lem:eq-selection}.
    Then for any $x^\star_1, x^\star_2 \in \MEQ(V)$
    and all $t \ge 0$:
	\begin{equation*}
		F_{x^\star_1}(y_{t+1})
		- 
		F_{x^\star_1}(y_{t})
		= 
		F_{x^\star_2}(y_{t+1})
		-
		F_{x^\star_2}(y_{t}) \;.
	\end{equation*}
\end{restatable}

\noindent
\textit{Proof.}
Let 
$\Delta F_{x^\star_1}(y_t) := F_{x^\star_1}(y_{t+1}) - F_{x^\star_1}(y_t)$
and 
$\Delta F_{x^\star_2}(y_t) := F_{x^\star_2}(y_{t+1}) - F_{x^\star_2}(y_t)$.
Taking the difference between
$\Delta F_{x^\star_1}(y_t)$ and $\Delta F_{x^\star_2}(y_t)$
then yields:
\begin{equation}
    \textstyle
	\Delta F_{x^\star_1}(y_t)
	- 
	\Delta F_{x^\star_2}(y_t)
	= 
	- \langle y_{t+1} - y_t, x^\star_1 - x^\star_2\rangle\,.
	\label{eq:bdi-01}
\end{equation}
On the one hand, since $x^\star_1, x^\star_2 \in \MEQ(V)$ it follows
from the definition of~\eqref{eq:minty-equality} that 
\begin{equation}
\langle V(x), x- x^\star_1 \rangle = 0\,, 
\quad \langle V(x), x- x^\star_2 \rangle = 0\,, 
\quad \forall x \in \calX\,. 
\end{equation}
Subtracting both identities then yields 
$\langle V(x), x^\star_1 - x^\star_2 \rangle = 0$ 
for all $x \in \calX$. 
Hence $x^\star_1 - x^\star_2 \in \Span(V(\calX))^{\bot}$. 
On the other hand, recall that for both~\eqref{eq:oftrl}
and~\eqref{eq:egftrl}, we have $y_{t+1} - y_t \in \Span(V(\calX))$
for all $t \ge 0$. 
Hence $\langle y_{t+1} - y_t, x^\star_1 - x^\star_2\rangle = 0$ and it follows from~\eqref{eq:bdi-01} that 
$\Delta F_{x^\star_1}(y_t)	- \Delta F_{x^\star_2}(y_t) = 0$. 
~ \hfill $\square$

\medskip
\noindent
\textit{2. Change in Bregman divergence
is equivalent for all MEQ points.} ~ 

\noindent
Using Lemma~\ref{lem:shifted-energy-invariant}, 
we can further establish the following invariant
between the total change in Bregman divergence
from two interior  MEQ points after $t$ steps: 

\medskip

\begin{restatable}{lem}{lembregmaninvarianttsteps}
	\label{lem:bregman-invariant-t-steps}
	Let $x^\star_1, x^\star_2 \in \MEQ(V) \cap \relint(\calX)$.
	Then for any $t \ge 1$, it holds that
	\begin{equation*}
		\D_R(x^\star_1, x_t) - \D_R(x^\star_1, x_0)
		=
		\D_R(x^\star_2, x_t) - \D_R(x^\star_2, x_0) \;.
	\end{equation*}
\end{restatable}

\noindent
\textit{Proof.} First, observe for any $x^\star_1, x^\star_2 \in \MEQ(V) \cap \relint(\calX)$ 
and for any $t \ge 0$ that 
\begin{align}
	\D_R(x^\star_1, x_{t+1}) - \D_R(x^\star_1, x_t)
	&= 
		F_{x^\star_1}(y_{t+1}) - F_{x^\star_1}(y_t) 
		\nonumber \\
	&= 
		F_{x^\star_2}(y_{t+1}) - F_{x^\star_2}(y_t) 
	= 
	\D_R(x^\star_2, x_{t+1}) - \D_R(x^\star_2, x_t)
	\label{eq:beq01}
\end{align}
The first and last equalities follow
from Proposition~\ref{prop:bregman-shifted-energy},
and the second equality 
by Lemma~\ref{lem:shifted-energy-invariant}.
Telescoping then yields
$\sum_{k=1}^{t}
	\D_R(x^\star_1, x_k) - \D_R(x^\star_1, x_{k-1})
    = 
\sum_{k=1}^{t}
	\D_R(x^\star_2, x_k) - \D_R(x^\star_2, x_{k-1})$
for all $t \ge 1$.
Thus for 
$x^\star_1, x^\star_2 \in \MEQ(V) \cap \relint(\calX)$,
we conclude by simplifying the two telescoping sums that
\begin{equation}
	\D_R(x^\star_1, x_t)  - \D_R(x^\star_1, x_0) 
	= 
	\D_R(x^\star_2, x_t)  - \D_R(x^\star_2, x_0) \;,
	\label{eq:bregman-id}
\end{equation}
which yields the desired equality.
~\hfill $\square$

\medskip
\noindent
\textit{3. Characterization of the limit as the
Bregman projection onto  $\MEQ(V)$.} ~

\noindent
As part (i) of the lemma proves
$x_{\infty} = 
\lim_{t\to \infty} x_t \in\MEQ(V) \cap \relint(\calX)$,
we substitute $x^\star_2 = x_{\infty}$ 
in the right hand side of 
Lemma~\ref{lem:bregman-invariant-t-steps}.
This then yields for all $t \ge 1$
and $x^\star\in \MEQ(V) \cap \relint(\calX)$ that
\begin{equation}
	\D_R(x^\star, x_t) - \D_R(x^\star, x_0)
	= 
	\D_R(x_{\infty}, x_t) - \D_R(x_{\infty}, x_0)
	\label{eq:bid-01}
\end{equation}
Then by continuity of $\D_R$, 
taking $t \to \infty$ yields 
for all $x^\star \in \MEQ(V) \cap \relint(\calX)$:
\begin{equation*}
	\D_R(x^\star, x_\infty) - \D_R(x^\star, x_0)
	= 
	\D_R(x_{\infty}, x_\infty) - \D_R(x_{\infty}, x_0)
	= 
	- \D_R(x_{\infty}, x_0) \;.
\end{equation*}
Rearranging then yields:
$\D_R(x^\star, x_0)
= 
\D_R(x^\star, x_\infty)
+ 
\D_R(x_\infty, x_0)$
for all $x^\star \in \MEQ(V) \cap \relint(\calX)$.
As $\D_R(x^\star, x_\infty) \ge 0$ for
all $x^\star \in \MEQ(V)$ by convexity,
we further obtain
\begin{equation}
	\D_R(x^\star, x_0) 
	\ge
	\D_R(x_{\infty}, x_0) 
	\;\;\text{for all $x^\star \in \MEQ(V) 
    \cap \relint(\calX) $} \;.
    \label{eq:argmin-interior}
\end{equation}
To extend~\eqref{eq:argmin-interior}
to also hold for \textit{non-interior}
$x^\star \in \MEQ(V)$,
observe that if $R(x^\star) = +\infty$, then also
$\D_R(x^\star, x_0) = +\infty$,
and thus~\eqref{eq:argmin-interior}
trivially holds. 
Now suppose 
$x^\star \in \MEQ(V)$ but 
$x^\star \notin \relint(\calX)$
and $R(x^\star) < \infty$ is finite. 
Then for  $s \in (0, 1)$,
let $x^\star_s := (1-s) x^\star + sx_{\infty}$. 
By convexity of $\MEQ(V)$
and since $x_{\infty} \in \relint(\calX)$,
it follows that $x^\star_s \in \MEQ(V) \cap \relint(\calX)$. 
Thus, applying~\eqref{eq:argmin-interior}
to $x^\star_s$ and by Jensen's inequality
(since $\D_R(\cdot, x_0)$ is convex 
in its first argument), we have
\begin{equation*}
    \D_R(x_{\infty}, x_0)
    \le
    \D_R(x^\star_s, x_0)
    \le 
    (1-s) \cdot \D_R(x^\star, x_0)
    + s \cdot \D_R(x_{\infty}, x_0) \;.
\end{equation*}
Rearranging, we then obtain 
$(1-s) \cdot \D_R(x_{\infty}, x_0) 
\le 
(1-s) \cdot \D_R(x^\star, x_0)$,
and thus dividing 
by $1-s > 0$ we obtain
$\D_R(x_{\infty}, x_0) \le \D_R(x^\star,x_0)$. 
Thus it holds also that 
\begin{equation}
    \D_R(x^\star, x_0)
    \ge
    \D_R(x_{\infty}, x_0)
    \;\;\text{for all $x^\star \in \MEQ(V)$}
    \label{eq:meq-projection}
\end{equation}
Moreover, since $\MEQ(V)$ is convex 
and $\D_R(\cdot, x_0)$ is strictly convex, 
\eqref{eq:meq-projection} implies that
$x_{\infty}$ is the unique minimizer
\begin{equation*}  
    \textstyle
    x_{\infty} 
    = 
    \argmin_{x^\star \in \MEQ(V)} \D_R(x^\star, x_0) 
    \;,
\end{equation*}
which concludes the proof.
\hfill $\blacksquare$


\section{Proof of Proposition~\ref{prop:op-dom}: Operator Domination Inequality}
\label{app:op-dom-details}

In this section, we develop the proof
of the operator domination inequality
of Proposition~\ref{prop:op-dom}
that was introduced in Section~\ref{sec:main-proof-overview}.
This inequality is restated here.

\medskip
\propopdom*

\subsection{Proof of Proposition~\ref{prop:op-dom}}
\label{app:op-dom-details:proof}

We now give the proof of Proposition~\ref{prop:op-dom}.
For readability, we first recall several definitions
for the symmetrized system from Setting~\ref{setting:weighted}.

\paragraph{Recap of spaces and maps.} 
Recall the subspaces $\calY = \Span(V(\calX))$ 
and $\calU = \Pi_{\calS^\bot} \calY$.
For the fixed initialization pair
$y_0 \in \R^D$ and $x_0 = \nabla F(y_0)$, 
recall that $\calR(x_0) = \{\nabla F(y_0 + y) : y \in \calY \} \subset \relint(\calX)$.
For $u \in \calU$, recall that:
\begin{equation*}
	x_0(u) = \nabla F(y_0 + u) \in \calR(x_0)  \,,
	\;\;
	\pi_{0}(u)
	= 
	\Pi_{\calS^\bot} V(x_0(u)) \,,
	\;\;
	u_0(x) 
	= 
	\Pi_{\calS^\bot}(\nabla R(x) - y_0) \;.
\end{equation*}
By Proposition~\ref{prop:reachable-helper},
recall the mapping $u \mapsto  x_0(u)$
is a bijection from the effective dual space $\calU$ 
onto the reachable primal set $\calR(x_0)$,
with inverse given by $u_0(x_0(u)) = u$
for $u \in \calU$.
Thus, fix $x^\star \in \MEQ(V) \cap \calR(x_0)$.
We will then let $u^\star = u_0(x^\star)$,
meaning $x_0(u^\star) = x^\star$.

\noindent
\textbf{Proof strategy.}
The proof is most conveniently carried out in the 
effective dual space $\calU$. 
Indeed, the full dual representation is degenerate 
along $\calS$, whereas the maps $u\mapsto x_0(u)=\nabla F(y_0+u)$
and $x\mapsto u_0(x)$ give inverse parameterizations 
between $\calU$ and the reachable set $\calR(x_0)$. 
This restriction to $\calR(x_0)$ is natural for our purposes, 
since the iterates initialized at $x_0$ remain in this set.
Moreover, energy dissipation confines the iterates to a compact
subset $\calK$. Thus, it suffices to establish a uniform
domination bound on the compact reachable region 
$\calK\cap\calR(x_0)$, rather than globally over the 
entire state space.

In the effective coordinates, the equilibrium $x^\star$ is 
represented by $u^\star=u_0(x^\star)$, and the projected operator
$\pi_0(u)=\Pi_{\calS^\bot}V(x_0(u))$ 
has $u^\star$ as its unique zero on the relevant reachable set. 
The restricted regularity assumption implies that 
the Jacobian of $\pi_0$ at $u^\star$ is non-singular, 
and therefore yields a local error bound
$\|\pi_0(u)\|_2
\gtrsim
\|u-u^\star\|_2$
for $u$ sufficiently close to $u^\star$.
Away from $u^\star$, uniqueness of the zero ensures that the ratio $\frac{\|\pi_0(u)\|_2}{\|u-u^\star\|_2}$
is strictly positive, while compactness of the 
relevant effective dual region upgrades this pointwise 
positivity to a uniform positive lower bound. 
Combining the local and non-local regimes
therefore yields a global error bound on the compact reachable region.

Finally, uniform upper and lower bounds on $\nabla^2F$ over
this compact region relate $\|\pi_0(u)\|_2^2$ to the dissipation term $\|V(x_0(u))\|_{y_0+u}^2$, and relate $\|u-u^\star\|_2^2$ to
the Bregman divergence $\D_R(x^\star,x_0(u))$. 
Combining these comparisons gives the desired operator 
domination inequality in effective coordinates. 
We then translate the resulting bound back to an 
arbitrary representation $x=\nabla F(y)\in\calK\cap\calR(x_0)$ 
using the fact that $y$ and the canonical representative $y_0+u_0(x)$ 
differ only by a direction in $\calS$, 
along which $\nabla^2F$ is invariant.

\smallskip
\noindent
\textbf{Proof.} The proof now proceeds via the following steps:

\smallskip
\noindent
\textbf{1. Compactness in effective dual space.} ~ 

\noindent
Let $\calK \subset \relint(\calX)$ be the compact
set from the proposition statement,
where we assume that $x^\star \in \calK$,
and recall that we are interested in establishing
a bound on $\|V(x)\|^2_{y}$
that holds for $y \in \R^D$
and $x = \nabla F(y)$
where $x \in \calK \cap \calR(x_0)$. 
Given the inverse relationship between
$u_0(\cdot)$
and $x_0(\cdot)$
between $\calU$ and $\calR(x_0)$,
it is convenient to consider the subset
of $\calU$ corresponding to 
$\calK \cap \calR(x_0)$.
For this, let $\calK_{\calU}$
be this subset of the effective dual space given by
\begin{equation}
	\calK_{\calU} 
	:= 
	\{
		u_0(x) : x \in \calK \cap \calR(x_0)
	\}
	\subseteq \calU \;.
	\label{eq:calK-calU}
\end{equation}
By Proposition~\ref{prop:reachable-helper},
recall that $x \in \calR(x_0)$
if and only if $u_0(x) \in \calU$. 
As the map $u_0$ is continuous 
on $\relint(\calX)$ and $\calU$ is a
closed linear subspace, 
it follows that $\calR(x_0)$ is
relatively closed in $\relint(\calX)$. 
Thus, $\calK \cap \calR(x_0)$
is a closed subset of the compact set $\calK$
and hence is compact. 
Moreover, $\calK_{\calU}$
and its convex hull $\conv(\calK_{\calU})$
are then both also compact.
Observe further that 
for $x^\star \in \calK \cap \calR(x_0)$
the corresponding $u^\star = u_0(x^\star) \in \calK_{\calU}$.

\smallskip
\noindent
\textbf{2. Relating
dissipation term to Bregman divergence.} ~

\noindent
Fix $u \in \calK_{\calU} \subset \calU$,
and recall that $x_0(u) = \nabla F(y_0 + u)$.
It then suffices to derive a lower bound 
on $\|V(x_0(u))\|^2_{y_0 + u}$
in terms of $\D_R(x^\star, x_0(u))$.
For this, we have the following components:

\smallskip
\noindent
\textit{(i) Lower-bounding $\|V(x_0(u))\|^2_{y_0 + u}$
in terms of $\|\pi_0(u)\|^2_{2}$}.

\noindent
First, recall that $\nabla^2 F(y) \succeq 0$
and $\Null(\nabla^2 F(y)) = \calS$ for
all $y \in \R^D$.
Then it holds that $\nabla^2 F(y)$
is positive definite when restricted to
$\calS^\bot$.
Recall by definition that $\calU \subseteq \calS^\bot$.
Then by the compactness of the
convex hull $\conv(\calK_{\calU})$
and the continuity of $\nabla^2 F(\cdot)$,
it follows that there exist positive constants
$0 < c_{\ell} \le c_{h} < \infty$ such that
\begin{equation}
	c_{\ell} \|v\|^2_2
	\le
	\langle v, \nabla^2 F(y_0 + u) v \rangle
	\le 
	c_{h} \|v\|^2_{2}
	\quad \forall\; u \in \conv(\calK_{\calU}), v \in \calU\;.
	\label{eq:energy-local-sc}
\end{equation}
For $u \in \calK_{\calU}$,
it then follows that
\begin{align}
	\|V(x_0(u))\|^2_{y_0 +u} 
	&= 
	\langle V(x_0(u)), \nabla^2 F(y_0 + u) V(x_0(u)) \rangle \nonumber\\
	&= 
	\langle \Pi_{\calS^\bot} V(x_0(u)),
	 \nabla^2 F(y_0 + u) \Pi_{\calS^\bot} V(x_0(u)) \rangle \nonumber\\ 
	&= 
	\langle \pi_0(u), \nabla^2 F(y_0 + u) \pi_0(u) \rangle \nonumber\\
	&\ge 
	c_{\ell} \|\pi_0(u)\|^2_{2} \;.
	\label{eq:od1}
\end{align}
Here, the second equality follows from the
fact that $\Null(\nabla^2 F(y_0+u)) = \calS$,
the third equality comes from the 
definition of $\pi_0(u)$,
and the final inequality comes from~\eqref{eq:energy-local-sc}.

\smallskip
\noindent
\textit{(ii) Upper-bounding $\D_R(x^\star, x_0(u))$
	in terms of $\|u^\star - u\|^2_2$.}

Next, recall that $x^\star = \nabla F(y_0 + u^\star)$
for $u^\star \in \calK_{\calU}$, $x_0(u) = \nabla F(y_0 +u)$, 
and use Proposition~\ref{prop:bregman-pd} to obtain: 
\begin{equation*}
\D_R(x^\star, x_0(u))
= 
F(y_0 +u) - F(y_0 + u^\star)
- \langle \nabla F(y_0 + u^\star), u - u^\star \rangle \,. 
\end{equation*}
By Taylor's theorem with remainder, 
it then follows that
\begin{equation}
	\textstyle
	\D_R(x^\star, x_0(u))
	=
	\frac{1}{2}
	\langle u - u^\star, 
	\nabla^2 F(y_0+ u^\star + u_s)(u-u^\star) \rangle \;,
	\label{eq:dra}
\end{equation}
where $u_s := s(u-u^\star)$
for some $s \in (0,1)$.

As $u \in \calK_{\calU}$,
then also $u^\star + u_s \in \conv(\calK_{\calU})$.
Moreover, we have $u-u^\star \in \calU$.
Thus, we can apply the upper bound
from~\eqref{eq:energy-local-sc}
to obtain 
\begin{equation}
	\textstyle
	\D_R(x^\star, x_0(u)) 
	\le 
	c_h \cdot \frac{1}{2} \|u - u^\star \|^2_{2} \;.
	\label{eq:od2}
\end{equation}

\smallskip
\noindent
Combining~\eqref{eq:od1} and~\eqref{eq:od2},
we have the inequalities:
\begin{equation}
	\|V(x_0(u))\|^2_{y_0+u}
	\ge 
	c_{\ell} \|\pi_0(u)\|^2_{2}
	\;\;\;\text{and}\;\;
	\|u-u^\star\|^2_{2}
	\ge 
	\tfrac{1}{c_h} \cdot 2 \D_R(x^\star, x_0(u)) \;.
	\label{eq:od3-summary}
\end{equation}
The remaining task is thus to establish 
that there exists a positive constant 
$b > 0$ such that
\begin{equation}
\|\pi_0(u)\|_2 \ge b \|u -u^\star\|_2
\;\;\;\;\text{for all $u \in \calK_{\calU}$}\;.
\label{eq:od3-goal}
\end{equation}
We proceed to establish existence
of such a constant:

\smallskip
\noindent
\textbf{3. Positivity of 
constant in~\eqref{eq:od3-goal}.}~ 

\noindent
\textit{(i) Unique zero of $\pi_0$.}

\noindent
First, we show that 
$u^\star= u_0(x^\star)$ is the unique zero of $\pi_0$.
For this, we have by assumption 
that $x^\star \in \MEQ(V) \cap \relint(\calX)$,
and thus by Proposition~\ref{prop:ne-vs-normal}
and Proposition~\ref{prop:pi0-zeros-ne}
we have that $V(x^\star) \in \calS$ 
and thus $\pi_0(u^\star) = 0$.
For the other direction, 
observe that  if $\pi_0(u) =0$, 
then we must have $x_0(u) \in \NE(V) \cap \calR(x_0)$
by Proposition~\ref{prop:pi0-zeros-ne}.
Under the regularity condition of 
Assumption~\ref{ass:restricted-regularity},
it then follows from 
Corollary~\ref{cor:reachable-meq}
that $x_0(u) \in \MEQ(V)$.
However, by Proposition~\ref{prop:unique-reachable-meq},
$|\MEQ(V) \cap \calR(x_0)| \le 1$,
and thus $x_0(u) = x^\star$.
Because the map $u \mapsto x_0(u)$
is bijective, it then
follows that $u = u_0(x_0(u)) = u_0(x^\star) = u^\star$.
Thus $u^\star \in \calK_{\calU}$ 
is the unique zero of $\pi_0(u)$.

\smallskip
\noindent
\textit{(ii) Local error bound.}  

\noindent
By Assumption~\ref{ass:restricted-regularity},
we thus have that the Jacobian
$J_{\pi_0(u^\star)} : \calU \to \calU$ 
of $\pi_0$ at $u^\star$ is non-singular. 
Then let $\sigma > 0$ denote the smallest 
singular value of $J_{\pi_0(u^\star)}$
(with respect to $\|\cdot\|_2$ on $\calU$).
As $\pi_0(u^\star) = 0$, differentiability of 
$\pi_0$ at $u^\star$ then implies that 
for $u \in \calU$:
\begin{equation*}
	\pi_0(u) = J_{\pi_0(u^\star)}(u-u^\star)
	+ r(u)\,,
	\;\;\text{where}\;\;
	\frac{\|r(u)\|_{2}}{\|u-u^\star\|_2} \to 0
	\;\;\text{as $u \to u^\star$}\;.
\end{equation*}
This implies \textit{locally} that  
there exists some $\delta > 0$ 
such that $\|r(u)\|_2 \le \frac{\sigma}{2} \|u-u^\star\|_2$
for all $u \in \calU$ where $\|u-u^\star\|_2 \le \delta$.
Taking norms, it then follows for all such
$u \in \calU: \|u - u^\star\|_2 \leq \delta$ that 
\begin{equation}
	\|\pi_0(u)\|_2 \ge 
	\|J_{\pi_0(u^\star)}(u-u^\star)\|_2 - \|r(u)\|_2
	\ge 
	\sigma \cdot \tfrac{1}{2} \|u-u^\star\|_2 \;.
	\label{eq:od3-local}
\end{equation}

\smallskip
\noindent
\textit{(iii) Non-local error bound.}

\noindent
For $\delta >0$ that leads to~\eqref{eq:od3-local},
let 
$\calK_{\calU, \delta} := 
\{u \in \calK_{\calU} : \|u -u^\star\|_2 \ge \delta\}$.
Suppose first 
that $\calK_{\calU, \delta} \neq \emptyset$,
which implies $u^\star \notin \calK_{\calU, \delta}$.
Moreover, $\calK_{\calU, \delta}$ is 
a closed subset of the compact set 
$\calK_{\calU}$ and is thus also compact.
Also, since $u^\star$ is the unique zero of $\pi_0$
from part (i), the map 
$u \mapsto \|\pi_0(u)\|_2 / \|u - u^\star\|_2$
is continuous and strictly positive on $\calK_{\calU, \delta}$.
Together, this implies
\begin{equation}
	a := 
	\min\nolimits_{u \in \calK_{\calU, \delta}}\,
	\frac{\|\pi_0(u)\|_2}{\|u-u^\star\|_2} > 0 \;.
	\label{eq:od3-global}
\end{equation}
If $\calK_{\calU,\delta}=\emptyset$,
the local bound~\eqref{eq:od3-local} 
already holds throughout
$\calK_\calU$, so taking $b=\sigma/2$
yields the desired~\eqref{eq:od3-goal}.
Combining the local 
inequality of~\eqref{eq:od3-local}
and the non-local inequality 
of~\eqref{eq:od3-global},
it then follows that 
\begin{equation}
	\|\pi_0(u)\|_2 
	\ge 
	\min(\sigma/2, a) \cdot \|u-u^\star\|_2
	\;\;\;\text{for all $u \in \calK_{\calU}$}\;.
	\label{eq:od3}
\end{equation}
We thus let $b := \min(\sigma/2, a) > 0$.

\smallskip
\paragraph{4. Combining the pieces.} ~

\noindent
Finally, we combine the bounds of~\eqref{eq:od3-summary}
and~\eqref{eq:od3} to obtain
for every $u \in \calK_{\calU}$ that
\begin{equation}
	\|V(x_0(u))\|^2_{y_0 + u}
	\ge 
	c_{\ell} \|\pi_0(u)\|^2_2
	\ge 
	c_{\ell} b^2 \|u -u^\star\|^2_2 
	\ge 
	\frac{2 c_\ell b^2}{c_h} \cdot \D_R(x^\star, x_0(u)) \;,
	\label{eq:od-full}
\end{equation}
where $c_{\ell}, b, c_h > 0$.
Then set $\wt \gamma(x_0, \calK) := (2 c_{\ell} b^2) / c_h > 0$,
and consider $y \in \R^D$ such that 
$x= \nabla F(y) \in \calK \cap \calR(x_0)$. 
Then instantiate $u = u_0(x) \in \calK_{\calU}$,
which implies that $x = x_0(u)$. 
Moreover, since by definition 
$u = \Pi_{\calS^\bot} (\nabla R(x) - y_0)$,
it follows that 
$\nabla R(x) - y_0 - u 
= \Pi_{\calS} (\nabla R(x) - y_0) \in \calS$.  
We also have $y - \nabla R(x) \in \calS$ 
by Proposition~\ref{prop:first-order-duality} 
since $x = \nabla F(y)$. 
Hence $y - (y_0 + u) = (y - \nabla R(x)) + (\nabla R(x) - y_0 - u) \in \calS$.
Thus by Proposition~\ref{prop:second-order-duality},
it follows that $\nabla^2 F(y) = \nabla^2 F(y_0 + u)$. 
Combining this with $x = x_0(u)$ yields:  
\begin{equation}
\|V(x)\|^2_{y} 
= \langle V(x), \nabla^2 F(y) V(x) \rangle 
= \langle V(x_0(u)), \nabla^2 F(y_0+ u) V(x_0(u)) \rangle 
= \|V(x_0(u))\|^2_{y_0+u}\,.
\end{equation}
Using this identity with~\eqref{eq:od-full}, 
we finally obtain
\begin{equation*}
	\|V(x)\|^2_{y}
	= 
	\|V(x_0(u))\|^2_{y_0+u}
	\ge 
	\wt \gamma(x_0, \calK) 
	\cdot \D_R(x^\star,  x)  \;,
\end{equation*}
which completes the proof of the proposition.
\hfill ~ $\blacksquare$

\section{Proof of Theorem~\ref{thm:main-last-iterate}:
    Last-Iterate Convergence Rates}
\label{app:last-iterate-proof}

In this section, we formally give the
proof of Theorem~\ref{thm:main-last-iterate},
restated here:

\medskip
\thmmainlast*

\medskip

\begin{proof}
The proof follows by combining the components
introduced in the overview of Section~\ref{sec:main-proof-overview}.
In particular, recall that our analysis is primarily based on the 
$M$-symmetrized system 
summarized in Setting~\ref{setting:weighted}.
The proof then follows by 
studying the iterates of the symmetrized system,
and using
(i) energy dissipation bounds of 
Lemma~\ref{lem:energy-dissipation-main},
(ii) the asymptotic convergence and equilibrium selection
characterization of Lemma~\ref{lem:eq-selection},
and (iii) by applying the operator domination 
inequality of Proposition~\ref{prop:op-dom} 
over the OFTRL and EGFTRL iterates.
In more detail:

\smallskip
\paragraph{1. Asymptotic convergence 
and equilibrium selection.} ~ 

\noindent
Let $\{(x_t, y_t)\}$ denote the symmetrized iterates.
By Lemma~\ref{lem:eq-selection}, 
we have that there is a unique limit
$x_{\infty} := \lim_{t \to \infty} x_t 
= \argmin_{x^\star \in \MEQ(V)} \, \D_R(x^\star, x_0) 
\in \relint(\calX)$.
By the relationships between the 
original and $M$-symmetrized systems 
of Proposition~\ref{prop:reduction-algorithmic},
this immediately implies
the existence of a unique limit
$w_{\infty} := \lim_{t \to \infty} w_t$
satisfying
$w_\infty = \argmin_{w^\star \in \WMEQ_M(G)}\, D_{Q_M}(w^\star, w_0)$.
Moreover, it thus holds by definition 
that $x_{\infty} = M^{1/2} w_{\infty}$.
This proves part (i) of the Theorem.

\smallskip
\paragraph{2. Convergence rate for symmetrized system.} ~ 

\noindent
For part (ii) of the theorem, we first establish
a linear last-iterate convergence rate
for the symmetrized system.
In particular, combining the equivalence of 
Proposition~\ref{prop:bregman-shifted-energy}
and 
using the energy dissipation upper bounds
of Lemma~\ref{lem:energy-dissipation-main},
we have for all $t  \ge 1$ that,
since $x_\infty \in \MEQ(V) \cap \relint(\calX)$:
\begin{equation}
	\textstyle
	\D_R(x_{\infty}, x_{t+1})
	- 
	\D_R(x_{\infty}, x_{t})
	= 
	F_{x_{\infty}}(y_{t+1}) 
	- 
	F_{x_{\infty}}(y_{t}) 
	\le 
	- 
	\frac{1}{8} 
	\cdot 
	\eta^2 \|V(x_t)\|^2_{y_t}
	\label{eq:last-energy}
\end{equation}
Now we have by definition 
of the iterates that $x_t \in \calR(x_0) \subset \relint(\calX)$ 
for all $t \ge 1$.
Also as established in the proof of
Lemma~\ref{lem:eq-selection}, 
there exists a compact set $\calC \subset \relint(\calX)$
such that $x_t \in \calC$ for all $t \ge 1$,
and moreover that $x_{\infty} \in \calC \cap \calR(x_0)$.
Then applying the operator domination
inequality of Proposition~\ref{prop:op-dom}
with the compact set $\calC$, the reachable primal 
set $\calR(x_0)$, 
for the point $x_{\infty} \in \MEQ(V) \cap \calR(x_0)$,
and for the iterates
$y_t \in \R^D$ and $x_t = \nabla F(y_t) \in \calC \cap \calR(x_0)$
it then follows from~\eqref{eq:last-energy} for
$t \ge 1$ that
\begin{align}
	\D_R(x_{\infty}, x_{t+1})
	- 
	\D_R(x_{\infty}, x_{t})
	&\le 
	- 
	\tfrac{1}{8} \wt \gamma(x_0, \calC)  \eta^2 
	\cdot
	\D_R(x_{\infty}, x_{t}) \;,
    \label{eq:op-dom-app}
\end{align}
where $\wt \gamma(x_0, \calC) > 0$ is positive
and uniform over all $t \ge 1$.

Now let $\eta_{\max}:=1/(880LB\omega\kappa_M)$,
and let 
$\wt \gamma := \min(\wt \gamma(x_0, \calC),
4/\eta^2_{\max}) > 0$.
Then observe that for all $\eta$ 
as in the theorem statement, it holds 
that $\wt \gamma \eta^2 / 8 \le 1/2$,
and $\wt \gamma(x_0, \calC) \eta^2 \ge \wt \gamma \eta^2$.

Then rearranging~\eqref{eq:op-dom-app}, 
it then follows recursively
for all $t \ge 0$ that
\begin{align}
	\D_R(x_{\infty}, x_{t+1})
	&\le
	\D_R(x_{\infty}, x_1)
	\cdot 
	(1- \tfrac{1}{8}  \wt \gamma \eta^2)^t  
	 \nonumber \\
	&\le
	\D_R(x_{\infty}, x_1)
	\cdot \exp\big(
		-  \tfrac{1}{8}  \wt \gamma  \eta^2 \cdot t
	\big)
\end{align}
Now for~\eqref{eq:egftrl-helper},
it is straightforward to check 
that the preceding arguments also directly yield
$\D_R(x_{\infty}, x_1) \le \D_R(x_{\infty}, x_0)$.
On the other hand, 
for~\eqref{eq:oftrl-helper},
we have by Lemma~\ref{lem:oftrl-initial-bregman}
that $\D_R(x_{\infty}, x_1) \le 2 \D_R(x_{\infty}, x_0)$.
Taking the larger of the two, it follows
for both algorithms that 
\begin{equation}
	\D_R(x_{\infty}, x_{t+1})
	\le 
	2 \D_R(x_{\infty}, x_0) 
	\cdot 
	\exp\big(
		- \tfrac{1}{8}  \wt \gamma  \eta^2 \cdot t
	\big)
	\;\;\;\text{for all $t \ge 0$} \;.
	\label{eq:last-iterate-symmetrized}
\end{equation}

\smallskip
\paragraph{3. Convergence rate for original system.} ~ 

\noindent
To complete the proof of part (ii)
of the theorem,
we translate the bound of~\eqref{eq:last-iterate-symmetrized}
back to the original system. 
For this, we use the relationships
$x_{t+1} = M^{1/2} w_{t+1}$,
$x_{0} = M^{1/2} w_{0}$,
and $x_{\infty} = M^{1/2} w_{\infty}$
and the inequalities
of~\eqref{eq:DQ-vs-DR} to obtain 
\begin{equation}
	m_{\min} \cdot \D_Q(w_{\infty}, w_{t+1})
	\le 
	\D_R(x_{\infty}, x_{t+1})
	\;\;\;\text{and}\;\;\;
	\D_R(x_{\infty}, x_{0})
	\le
	m_{\max}
	\D_Q(w_{\infty}, w_{0}) \;.
	\label{eq:drelations}
\end{equation}
Let $\gamma(w_0) := \frac{1}{8} \wt \gamma > 0$
for the positive constant from Step 2. 
Then substituting the bounds of~\eqref{eq:drelations}
into~\eqref{eq:last-iterate-symmetrized}
and rearranging, 
we obtain
\begin{equation}
	\D_Q(w_{\infty}, w_{t+1})
	\le 
	2 \tfrac{m_{\max}}{m_{\min}} \cdot 
	\D_Q(w_{\infty}, w_0)
	\cdot 
	\exp\big(
		- \gamma(w_0) \cdot \eta^2 t
	\big)
	\;\;\;\;\text{for all $t \ge 0$}\;,
\end{equation}
which concludes the proof of the theorem.
\end{proof}

\section{Proof of Theorem~\ref{thm:ftrl-diverge-main}: 
Divergence of FTRL}
\label{app:vanilla-ftrl}

In this section,
we develop the proof of Theorem~\ref{thm:ftrl-diverge-main},
restated here:

\medskip
\thmftrldiverge*

\medskip
The proof of the theorem first requires
establishing a bound on the 
shifted energy \textit{expansion}
under vanilla FTRL (for the symmetrized system).
We present this result in Section~\ref{app:ftrl-main:energy},
and we then give the full
proof of Theorem~\ref{thm:ftrl-diverge-main}
in Section~\ref{app:vanilla:main-proof}.

\smallskip
\subsection{Energy Expansion under FTRL}
\label{app:ftrl-main:energy}

\smallskip
We prove the following
bounds on \textit{expansion}
of shifted-energy along
the vanilla~\eqref{eq:ftrl} iterates
for the symmetrized system from Setting~\ref{setting:weighted}.

\begin{restatable}[Energy expansion]
	{lem}{lemftrlenergyexpand}
	\label{lem:ftrl-energy-expand}
	Let $\{(x_t, y_t)\}$
	be the iterates of~\eqref{eq:ftrl}
	for the symmetrized system
	with $\eta \in (0, 1/(10\wh B \wh \omega)]$. 
	Then for all $x^\star \in \MEQ(V)$
	and all $t \ge 0$:
	\begin{equation*}
		\textstyle
		F_{x^\star}(y_{t+1}) - F_{x^\star}(y_t)
		\ge
		\frac{1}{4} 
		\cdot 
		\eta^2 \|V(x_t)\|^2_{y_t} \;.
	\end{equation*}
\end{restatable}

\begin{proof}
Reusing 
Proposition~\ref{prop:ftrl-iftrl-energy},
we have
for all $x^\star \in \MEQ(V)$
and $t \ge 0$ that:
\begin{equation}
	\textstyle
	F_{x^\star}(y_{t+1})
	- 
	F_{x^\star}(y_{t})
	= 
	\D_F(y_{t+1}, y_t)
	\label{eq:expand-term}
\end{equation}
Now for all $t \ge 0$,
recall that $\|y_{t+1}- y_t\|_{*} \le \eta \wh B$.
Then by the lower bound of Lemma~\ref{lem:bregman-local-bound},
we have 
\begin{equation}
	\textstyle
	\D_F(y_{t+1}, y_t)
	\ge 
	\big(\frac{\exp(-\wh \omega \eta \wh B) + \wh \omega \eta \wh B -1}
	{(\wh \omega \eta \wh B )^2} \big)
	\cdot \|y_{t+1} -y_t\|^2_{y_t}
	= 
	\big(\frac{\exp(-\wh \omega \eta \wh B) + \wh \omega \eta \wh B -1}
	{(\wh \omega \eta \wh B )^2} \big)
	\cdot \eta^2 \|V(x_t)\|^2_{y_t} \;.
	\label{eq:df-term}
\end{equation}
The multiplicative factor in~\eqref{eq:df-term}
is positive and strictly decreasing in $\eta$.
Moreover, for $\eta$ as in the lemma statement,
it is straightforward to check that 
this term is bounded below by $1/4$.
Substituting this bound into~\eqref{eq:df-term}
and back into~\eqref{eq:expand-term}
then concludes the proof.
\end{proof}

\subsection{Proof of Theorem~\ref{thm:ftrl-diverge-main}}
\label{app:vanilla:main-proof}

We now give the full proof of
Theorem~\ref{thm:ftrl-diverge-main},
which proceeds in several steps.

\paragraph{1. Overview and setup.} ~

Similar to the proof of Theorem~\ref{thm:main-last-iterate},
we analyze the FTRL iterates under the symmetrized 
system of Setting~\ref{setting:weighted}
and then translate the guarantees back to the
original system.
Specifically, let $\{(x_t, y_t)\}$
denote the vanilla~\eqref{eq:ftrl}
iterates for the symmetrized system.

Then, fix $w^\star \in \NE(G) \cap \relint(\calW)$,
and let $x^\star = M^{1/2} w^\star \in \NE(V) \cap \relint(\calX)$.
Fix $\alpha$ satisfying the statement of the theorem,
and define $T_{\alpha}(x^\star)
= \sup\{t \ge 0 : \D_R(x^\star, x_t) \le \alpha\}$.
Then by definition of the symmetrized system,
it holds that $T_{\alpha}(x^\star) = T_{\alpha}(w^\star)$,
for $T_{\alpha}(w^\star)$ as defined 
in the theorem statement.
Moreover, we fix an initialization pair 
$(w_0,z_0)$ (respectively, $(x_0, y_0)$)
such that $w_0 \in \relint(\calW) \setminus \NE(G)$
(respectively, $x_0 \in \relint(\calX) \setminus \NE(V)$).

\medskip
\paragraph{2. Convergence rate to
sublevel set boundary for symmetrized system.} ~

\noindent
First, recall that the operator $G$ is linear 
by assumption, which also implies $V$ is linear.
In this case, we have by Proposition~\ref{prop:ne-wmeq-linear}
that $\NE(V) = \MEQ(V)$.
Now recall for $t \ge 0$
that $x_t \in \calR(x_0) \subseteq \relint(\calX)$. 
Moreover, observe that combining 
Proposition~\ref{prop:unique-reachable-meq}
and the result of Lemma~\ref{lem:eq-selection}
establishes the existence and uniqueness
of a point $\wt x \in \MEQ(V) \cap \calR(x_0)$.
In particular, note that it suffices to 
apply Lemma~\ref{lem:eq-selection}
to an auxiliary OFTRL trajectory
initialized at the same pair $(x_0, y_0)$,
where it then follows that its limit
is a point $\wt x \in \MEQ(V) \cap 
\calR(x_0)$, which establishes existence.
Proposition~\ref{prop:unique-reachable-meq}
then implies this point is unique.

Now let $x^\star \in \NE(V) \cap \relint(\calX)$
be the point from the setup above.
As $\NE(V) = \MEQ(V)$,
it follows also that $x^\star \in \MEQ(V)$.
Then following arguments identical
to part (ii) of Lemma~\ref{lem:eq-selection},
we further obtain for all $t \ge 0$ that
\begin{equation}
	\D_R(x^\star, x_t) 
	= 
	\D_R(x^\star, \wt x) 
	+ 
	\D_R(\wt x, x_t) \;.
	\label{eq:vh-01}
\end{equation}
For convenience, we will define
\begin{equation}
	a_t := \D_R(x^\star, x_t) \,,
	\;\;
	b_t:= \D_R(\wt x, x_t)  \,,
	\;\;
	c := \D_R(x^\star, \wt x) \;.
	\label{eq:vh-constants}
\end{equation}
Then~\eqref{eq:vh-01}
says that 
$a_t = b_t + c$,
where $c \ge 0$ by convexity of $R$. 

Now, by Lemma~\ref{lem:ftrl-energy-expand}
and Proposition~\ref{prop:bregman-shifted-energy},
we obtain for all $t \ge 0$ and
$x^\star \in \MEQ(V) = \NE(V)$ that
\begin{equation}
	\D_R(x^\star, x_{t+1}) 
	- \D_R(x^\star, x_t)
	= 
	F_{x^\star}(y_{t+1})
	- 
	F_{x^\star}(y_{t})
	\ge
	\tfrac{1}{4} \cdot \eta^2 \|V(x_t)\|^2_{y_t} \;.
	\label{eq:ftrl-energy}
\end{equation}
As both $\wt x \in \MEQ(V)$
and $x^\star \in \MEQ(V)$,
it follows that we can write  
for all $t \ge 0$ that
\begin{equation}
	a_{t+1} - a_t
	= 
	b_{t+1}- b_t 
	\ge 
	\tfrac{1}{4} \cdot \eta^2 \|V(x_t)\|^2_{y_t}
	\ge 0\;.
	\label{eq:ftrl-energy-simple}
\end{equation}

For $\alpha$ as in the statement of
the theorem, we then define the sublevel set
\begin{equation*}
\calC_{\alpha} 
:= \{x \in \relint(\calX) : \D_R(x^\star,x) \le \alpha \}
\subset \relint(\calX) \;.
\end{equation*}
As $x^\star \in \relint(\calX)$
by assumption, it holds 
from Proposition~\ref{prop:bregman-boundary-blowup}
that 
$\calC_{\alpha} \subset \relint(\calX)$ is compact.
Also, for the assumptions on $x_0$
(as implied by the condition on $w_0$) in the theorem statement, 
we have $x_0 \in \calC_{\alpha}$. 
Moreover, we have from the identity of~\eqref{eq:vh-01} 
at time $t=0$ that 
\begin{equation*}
	\D_R(x^\star, \wt x)
	\le 
	\D_R(x^\star, \wt x)
	+ 
	\D_R(\wt x, x_0)
	= 
	\D_R(x^\star, x_0)
	\le 
	\alpha \;,
\end{equation*}
which also implies that 
$\wt x \in \calC_{\alpha}$.
By definition of $T_{\alpha}(x^\star)$,
and using the monotonicity of $\{a_t\}$
established in~\eqref{eq:ftrl-energy-simple},
we then further have that $x_t \in \calC_{\alpha} \cap \calR(x_0)$
for all $0 \le t \le T_{\alpha}(x^\star)$.

Together, this allows us to apply the operator
domination inequality of Proposition~\ref{prop:op-dom}
for the point $\wt x \in \MEQ(V) \cap \calR(x_0)$,
for the compact set $\calC_{\alpha} \subset \relint(\calX)$,
and for the FTRL iterates ${(x_t, y_t)}$
satisfying $x_t = \nabla F(y_t)  \in \calR(x_0) \cap \calC_{\alpha}$
for all $0 \le t \le T_{\alpha}(x^\star)$.
This yields for all such $t$ that
\begin{equation}
	\|V(x_t)\|^2_{y_t}
	\ge 
	\wt \gamma \cdot 
	\D_R(\wt x, x_t)
	= 
	\wt \gamma \cdot b_t 
	\label{eq:ftrl-opdom-app}
\end{equation}
for a positive constant $\wt \gamma := \wt \gamma(x_0, \calC_{\alpha}) > 0$.

Now, to relate \eqref{eq:ftrl-opdom-app}
back to the term $a_t = \D_R(x^\star, x_t)$,
observe that $b_t \ge b_0$ and $a_t \ge a_0$
 due to~\eqref{eq:ftrl-energy-simple},
and moreover $b_0 = \D_R(\wt x, x_0) > 0$
and $a_0 = \D_R(x^\star, x_0) > 0$
by strict convexity of $R$ and
due to the assumption that $x_0 \notin \NE(V)$.
As also $c \ge 0$ by convexity of $R$, we have 
the relationships:
\begin{equation}
	\frac{b_t}{a_t}
	= 
	\frac{b_t}{b_t+c} 
	\ge 
	\frac{b_0}{b_0+ c}
	= \frac{b_0}{a_0} \;,
	\label{eq:chelper-02}
\end{equation}
where the final equality is
due to $b_0 +c = a_0$ from~\eqref{eq:vh-01}.
Then combining the inequalities~\eqref{eq:ftrl-energy-simple},
\eqref{eq:ftrl-opdom-app}, and~\eqref{eq:chelper-02},
we obtain the bound
\begin{equation}
	a_{t+1}
	\ge
	a_t 
	+ 
	\tfrac{1}{4} \cdot 
	\wt \gamma \eta^2 \cdot  b_t
	\ge 
	\big(
		1+ \tfrac{\wt \gamma b_0}{4 a_0} \eta^2 
	\big)
	\cdot 
	a_t \;,
	\label{eq:ftrl-diverge-t}
\end{equation}
for all $0 \le t \le T_{\alpha}(x^\star)$.
Then define 
the constant 
$\beta := \frac{\wt \gamma b_0}{4 a_0}  > 0$.
From~\eqref{eq:ftrl-diverge-t}, it then follows recursively for
all $0 \le t \le T_{\alpha}(x^\star)$ that
\begin{equation}
	\D_R(x^\star, x_t) 
	= 
	a_t
	\ge 
	a_0 \cdot 	(1 + \beta \eta^2)^t
	= 
	\D_R(x^\star, x_0) \cdot (1+ \beta \eta^2)^t \;.
	\label{eq:ftrl-diverge-sym}
\end{equation}

\medskip
\paragraph{3. Translating back to original system.} ~

\noindent
It remains to translate 
the bound in~\eqref{eq:ftrl-diverge-sym}
back to the original system. 
For this, recall that we have
$x_t = M^{1/2} w_t$,
$x_0 = M^{1/2} w_0$,
and $x^\star = M^{1/2}w^\star$. 
Then using the relationships of~\eqref{eq:DQ-vs-DR}:
\begin{equation}
	m_{\max} \cdot \D_Q(w^\star, w_t) 
	\ge \D_R(x^\star, x_t) 
	\;\;\;\text{and}\;\;\;
	\D_R(x^\star, x_0)
	\ge 
	m_{\min}
	\D_Q(w^\star, w_0) \;.
	\label{eq:ftrl-breg-relate}
\end{equation}
for the constant $\beta$ defined in step 2. 
Then substituting~\eqref{eq:ftrl-breg-relate}
into~\eqref{eq:ftrl-diverge-sym} and rearranging
yields 
for all $0 \le t \le T_{\alpha}(x^\star)$ that 
\begin{equation*}
	\D_Q(w^\star, w_t)
	\ge 
	\tfrac{m_{\min}}{m_{\max}}
	\cdot 
	\D_Q(w^\star, w_0)
	\cdot 
	(1 + \beta \eta^2)^t \;.
\end{equation*}
Setting $\beta(w_0, w^\star, \alpha) := \beta > 0$
and recalling that $T_{\alpha}(x^\star) = T_{\alpha}(w^\star)$
then yields the desired statement.
\hfill $\blacksquare$

\newpage
\bibliographystyle{apalike}
\bibliography{references.bib}

\end{document}